\documentclass[conference]{IEEEtran}
\usepackage[utf8]{inputenc}
\usepackage[pdftex]{graphicx}
\usepackage{soul}
\usepackage{epstopdf}
\usepackage{multirow}
\usepackage{amsmath}
\usepackage{mathtools}
\usepackage{authblk}
\usepackage{cite}
\usepackage{cleveref}
\usepackage{latexsym}
\usepackage[justification=centering]{caption}
\usepackage{caption}
\usepackage{amsthm}
\usepackage{balance}
\usepackage{relsize}
\usepackage{color}
\usepackage{amssymb}
\usepackage{eucal}
\usepackage{url}
\usepackage{comment}
\usepackage{booktabs,subcaption,amsfonts,dcolumn}
\usepackage[nodisplayskipstretch]{setspace}
\usepackage{lipsum}
\usepackage{acronym}
\usepackage{makecell}
\usepackage{mathrsfs}
\usepackage{dsfont}
\usepackage{longtable}
\usepackage{mdframed}
\usepackage{fancybox}

\usepackage{color}

\usepackage{amsmath}
\usepackage{amssymb,mathtools}
\usepackage{amsfonts}
\usepackage{mathrsfs}
\usepackage{bm}
\usepackage{algorithmic}
\usepackage{algorithm}
\usepackage{array}
\usepackage{multirow}
\usepackage{booktabs}
\usepackage{colortbl}

\usepackage{caption}
\usepackage{graphicx}
\usepackage{float} 
\usepackage{stfloats}
\usepackage{subcaption}

\usepackage{tikz}
\usetikzlibrary{matrix}
\usepackage{tikz-cd}
\usepackage{lineno}
\usepackage{xspace}

\usepackage{amsthm}

\newtheorem{theorem}{Theorem}

\usepackage{pifont}

\newcommand\SystemName{\textsc{FASTLMPI}\xspace}

\usepackage
[
        left=1.62cm,
        right=1.62cm,
        top=0.7in,
]{geometry}

\newcolumntype{L}{>{\centering\arraybackslash}m{6cm}}

\newacro{kpi}[KPI]{Key Performance Indicator}
\newacro{tdoc}[Tdoc]{Technical Documents}
\newacro{cr}[CR]{Change Request}

\newacro{rf}[RF]{Radio Frequency}
\newacro{ran}[RAN]{Radio Access Network}
\newacro{sa}[SA]{system architecture}
\newacro{ct}[CT]{core network and terminals}
\newacro{qos}[QoS]{quality of service}
\newacro{wg}[WG]{working group}
\newacro{5ga}[5G-A]{5G advanced}
\newacro{bpe}[BPE]{byte pair encoding}
\newacro{mmb}[MMB]{multi-modal beamforming}
\newacro{jscc}[JSCC]{joint source-channel coding}
\newacro{sft}[SFT]{Supervised Fine-Tuning}
\newacro{dpo}[DPO]{Direct Preference Optimization}
\newacro{llm}[LLM]{Large Language Model}
\newacro{lmm}[LMM]{Large Multi-modal Model}
\newacro{fm}[FM]{Foundation Model}
\newacro{ai}[AI]{Artificial Intelligence}
\newacro{lm}[LM]{Language Modeling}
\newacro{ptlm}[PTLM]{Pre-Trained Language Model}
\newacro{nlp}[NLP]{Natural Language Processing}
\newacro{dl}[DL]{Deep Learning}
\newacro{nn}[NN]{Neural Network}
\newacro{dnn}[DNN]{Deep Neural Network}
\newacro{cnn}[CNN]{Convolutional Neural Network}
\newacro{rnn}[RNN]{Recurrent Neural Network}
\newacro{gnn}[GNN]{Graph Neural Network}
\newacro{ml}[ML]{Machine Learning}
\newacro{cv}[CV]{Computer Vision}
\newacro{ssl}[SSL]{Self-Supervised Learning}
\newacro{tl}[TL]{Transfer Learning}
\newacro{nlm}[NLM]{Neural Language Model}
\newacro{lstm}[LSTM]{Long Short-Term Memory}
\newacro{gpt}[GPT]{Generative Pre-trained Transformer}
\newacro{bert}[BERT]{Bidirectional Encoder Representation from Transformer}
\newacro{nlu}[NLU]{Natural Language Understanding}
\newacro{nlg}[NLG]{Natural Language Generation}
\newacro{t5}[T5]{Text-to-Text Transfer Transformer}
\newacro{icl}[ICL]{In-Context Learning}
\newacro{rlhf}[RLHF]{Reinforcement Learning with Human Feedback}
\newacro{mha}[MHA]{Multi-Head Attention}
\newacro{clm}[CLM]{Causal Language Modeling}
\newacro{mlm}[MLM]{Masked Language Modeling}
\newacro{plm}[PLM]{Permuted Language Modeling}
\newacro{dae}[DAE]{Denoising AutoEncoder}
\newacro{rf}[RF]{Radio Frequency}
\newacro{sota}[SOTA]{state of the art}
\newacro{rag}[RAG]{Retrieval Augmented Generation}
\newacro{moe}[MoE]{Mixture of Expert}
\newacro{peft}[PEFT]{Parameter-Efficient Fine-Tuning}
\newacro{sdo}[SDO]{Standards Developing Organization}
\newacro{cot}[CoT]{Chain-of-Thought}
\newacro{rl}[RL]{Reinforcement Learning}
\newacro{vlm}[VLM]{Visual Language Model}
\newacro{6g}[6G]{Sixth Generation}
\newacro{cv2x}[CV2X]{Cellular Vehicle-to-Everything}
\newacro{esti}[ESTI]{European Telecommunication Standards Institute}
\newacro{oran}[O-RAN]{Open Radio Access Network}
\newacro{qos}[QoS]{Quality of Service}
\newacro{3gpp}[3GPP]{Third Generation Partnership Project}
\newacro{itu}[ITU]{International Telecommunication Union}
\newacro{ran}[RAN]{Radio Access Network}
\newacro{bs}[BS]{Base Station}
\newacro{its}[ITS]{Intelligent Transport System}
\newacro{rrm}[RRM]{Radio Resource Management}
\newacro{lora}[LoRA]{Low Rank Adaptation}
\newacro{mlp}[MLP]{Multi-Layer Perceptron}
\newacro{vit}[ViT]{Vision Transformer}
\newacro{qat}[QAT]{Quantization Aware Training}
\newacro{ptq}[PTQ]{Post-Training Quantization}
\newacro{kv}[KV]{Key-Value}
\newacro{rleif}[RLEIF]{Reinforcement Learning from Evol-Instruct Feedback}
\newacro{v2x}[V2X]{Vehicle to Everything}
\newacro{rag}[RAG]{Retrieval Augmented Generation}
\newacro{fim}[FIM]{Fill-In-the-Middle}
\newacro{mcq}[MCQ]{Multiple-Choice Question}
\newacro{qa}[QA]{Question Answering}
\newacro{ieee}[IEEE]{Institute of Electrical and Electronics Engineers}
\newacro{urllc}[URLLC]{Ultra Reliable and Low Latency Communication}
\newacro{kl}[KL]{Kullback-Leibler}
\newacro{cdf}[CDF]{Cumulative Density Function}
\newacro{rrm}[RRM]{Radio Resource Management}

\title{Accelerating Private Large Transformers Inference through Fine-grained Collaborative Computation}

\author{
Yuntian Chen,
Zhanyong Tang,
Tianpei Lu, 
Bingsheng Zhang, 
Zhiying Shi, 
Zheng Wang
\thanks{Yuntian Chen, Zhanyong Tang, and Zhiying Shi are with the School of Information Science and Technology, Northwest University, Xi’an 710127,  China (e-mail: chenyt$\_$x@163.com, zytang@nwu.edu.cn, shizhiying@stumail.nwu.edu.cn).\\
Tianpei Lu and Bingsheng Zhang are with the College of Computer Science and Technology, Zhejiang University, Hangzhou 310058, China (e-mail: lutianpei@zju.edu.cn, bingsheng@zju.edu.cn).\\
Zheng Wang is with the School of Computer Science, University of Leeds, Leeds, LS2 9JT, United Kingdom (e-mail: z.wang5@leeds.ac.uk).
}
}

\begin{document}

\maketitle

\begin{abstract}
Homomorphic encryption (HE) and secret sharing (SS) enable computations on encrypted data, providing significant privacy benefits for large transformer-based models (TBM) in sensitive sectors like medicine and finance. However, private TBM inference incurs significant costs due to the coarse-grained application of HE and SS. We present \SystemName, a new approach to accelerate private TBM inference through fine-grained computation optimization. Specifically, through the fine-grained co-design of homomorphic encryption and secret sharing, \SystemName achieves efficient protocols for matrix multiplication, SoftMax, LayerNorm, and GeLU. In addition, \SystemName introduces a precise segmented approximation technique for differentiable non-linear, improving its fitting accuracy while maintaining a low polynomial degree. Compared to solution BOLT (S\&P'24), \SystemName shows a remarkable 54\% to 64\%  decrease in runtime and an impressive 72.2\% reduction in communication costs.
\end{abstract}


\section{Introduction}
\IEEEPARstart{D}{eep} learning models are often hosted in the cloud to provide inference services~\cite{DLaaS2023}. This provides a way for the service provider to protect the intellectual property of the trained model. However, sending user data to an untrusted cloud provider can raise privacy concerns as this can compromise data confidentiality \cite{singh2018data}. 

Homomorphic encryption (HE) \cite{HE2024} and secure multi-party computation (MPC) \cite{MPC2024} are two techniques for privacy-preserving deep learning inference. HE enables models to compute on encrypted user inputs and return encrypted results that are decrypted only by the user. MPC allows multiple parties to jointly compute functions on their private inputs by splitting data into shared pieces, ensuring no party has access to the full dataset or intermediate results, thereby maintaining privacy throughout the inference process.

Unfortunately, HE and MPC also have drawbacks, as HE can massively slow down model performance, while MPC introduces substantial communication overhead. Previous efforts have sought to combine HE and MPC by using HE for matrix multiplication - the dominant computation pattern in deep neural networks (DNNs) - and MPC for other DNN operations. This approach aims to leverage the strengths of both techniques. However, it still faces significant computational overhead, as matrix multiplication can account for over 90\% of inference time in modern DNN architectures like Transformer-based models (TBMs) \cite{Attention2017}.

\begin{table}[t!]    
\caption{\footnotesize 
The performance percentage of each operator in BOLT under LAN network setting.}
    \centering \footnotesize
    \begin{tabular}{c|c|ccc}
    \toprule
          & \textbf{HE} & \multicolumn{3}{c}{\textbf{SS}}  \\
        \midrule
       \textbf{Operations}  & \textbf{Linear} & \textbf{SoftMax} & \textbf{GeLU} & \textbf{LayerNorm} \\
       \textbf{Runtime}  & $74.5\%$ & $9.3\%$ & $8.1\%$ &  $8.1\%$ \\
       \textbf{Comm.}  & $1.2\%$ & $40.7\%$ & $41.3\%$ & $16.8\%$\\
       \bottomrule
    \end{tabular}

    \label{tab: bolt_cost}
\end{table}

As a concrete example, Table \ref{tab: bolt_cost} reports the results for applying BOLT \cite{BOLT2024}. 
The high volume of linear operations, specifically matrix multiplications, is one of the primary contributors to inference latency. Non-linear operators, including SoftMax, GeLU, and LayerNorm, are major contributors to increased communication costs, especially in challenging network environments where they become the primary performance bottleneck. End-to-end inference experiments show that although BumbleBee \cite{Bumblebee2023} exhibits advantages in certain linear operations, \SystemName's superior performance in non-linear operations makes it more competitive overall.  

The design of more practical private TBM inference is hindered by the following challenges:

\begin{enumerate}

    \item \textbf{High inference latency}: The attention mechanism has significantly enhanced the learning and analysis capabilities of the model compared to traditional convolution. However, this comes at the cost of substantially increased computational complexity, demanding matrix multiplication operations. 
    with large dimensions of $128\times768\times64$ and $128\times768\times3072$. Large-scale matrix multiplication is a significant optimization problem. Computational efficiency under ciphertext is of particular concern. Most of the existing cryptographic protocols for private matrix multiplication rely on HE. While HE offers the benefit of concise communication cost, it requires time-consuming operation rotation, leading to poor inference efficiency. 

    \textbf{The first challenge} is to design an efficient secure matrix multiplication protocol.

    \IEEEpubidadjcol
    \item \textbf{High communication cost}: Compared to the ReLU non-linear functions commonly used in convolutional neural networks, TBM employs more sophisticated non-linear functions such as SoftMax within attention mechanisms, GeLU in feed-forward networks, and LayerNorm for normalization. 
    These non-linear operations often rely on MPC techniques like secret sharing (SS) \cite{HybridSS2012}. In BERT$_\mathsf{base}$, the SoftMax function is invoked $2.36 \times 10^6$ times, requiring 300 rounds of interaction, resulting in a communication cost of 1447.65 MB.

    \textbf{The second challenge} is to design efficient protocols for non-linear with low communication costs.

    \item \textbf{High-degree, low-accuracy piecewise approximation}: A line of works~\cite{BOLT2024, Nexus2024, Puma2023, Bumblebee2023} typically uses the piecewise fitting method for non-linear e.g, GeLU, often with very similar segment endpoints. However, we found that the errors were huge at these endpoints, which was the primary reason for the decline in the accuracy of the approximation function. Moreover, the use of piecewise approximation necessitates higher-degree polynomials, thereby increasing computational complexity. For instance, Lu et.al~\cite{Puma2023, Bumblebee2023} use sixth-degree polynomials for piecewise approximation of the GeLU function.
    
    \textbf{The third challenge} is to develop piecewise functions with higher precision while maintaining low-degree polynomial coefficients.

\end{enumerate}

Based on the above analysis, we observe a pattern: existing private TBM inference solutions employ HE techniques to design secure matrix multiplication protocols for linear operators and SS techniques to design secure protocols for non-linear operators such as SoftMax, GeLU, and LayerNorm. In other words, after distinguishing between linear and non-linear operators at a high level, different cryptographic primitives are independently used to design secure protocols. While this classic strategy has effectively advanced the inference latency and communication costs of private TBM inference, the independent use of HE and SS is more likely to reveal the drawbacks of slow computation in HE and high communication costs in SS. 

Motivated by this, we aim to provide a highly efficient and communication-friendly solution, called \SystemName, for private TBM inference. \SystemName integrates HE and SS at a finer granularity, proposing novel protocols that eliminate the high-level distinction between linear and nonlinear operations, such as Matmul-LayerNorm and Matmul-SoftMax. Rather than applying HE solely to Matmul and SS to SoftMax, LayrNorm, and GeLU, we break down these operator boundaries and choose between HE and SS at a more granular level. 
% We have observed significant benefits by applying this approach to the design of secure computation protocols for Matmul, Softmax, GeLU, and LayNorm. 
Specifically, the strategic integration of SS into HE-based Matmul protocol eliminates the computationally expensive rotation operation without increasing computational overhead. Moreover, using HE techniques in intermediate steps of non-linear leads to a substantial reduction in communication overhead. For example, in the exponential calculation, we can compute the exponential of additive SS, package intermediate results into SIMD-based HE ciphertexts, and perform ciphertext-plaintext addition to obtain the result, which is known to be more efficient than ciphertext-ciphertext multiplications~\cite{cong2022CCismoreefficient} and we just need to introduce little acceptable communication costs. Compared to Cryptflow2 \cite{Cryptflow22020}, which demands 117 rounds of interaction and incurs a communication cost of $592$KB for the computation of the exponential of a $128 \times 128$ dimensional, our method only needs to send a single SIMD-based ciphertext (approximately $340$KB).

Piecewise approximation of non-linear functions is another focus of our work. The concavity, convexity, and curvature of a function are usually determined by inflection points, which are the points where the second derivative of the function is zero. At these points, the curvature of the function is usually small, indicating that the function tends to be smooth. Therefore, we select these special points as segment endpoints for function fitting, which can solve the problem of large errors at segment endpoints in traditional methods. An additional benefit is that piecewise functions have lower polynomial degrees. This approach applies to GeLU and can be extended to other differentiable non-linear functions such as the conventional activate function, e.g. Sigmoid, Tanh, and Mish.

To summarize, our contributions are as follows:
\begin{enumerate}
    \item \textbf{Novel protocol based on fine-grained integration of HE and SS }: We propose new 2PC protocols, including secure matrix multiplication $\Pi_{\mathsf{matmul}}$, secure Softmax $\Pi_{\mathsf{softmax}}$, secure layerNorm $\Pi_{\mathsf{ln}}$, and secure GeLU $\Pi_{\mathsf{gelu}}$. Compared to BOLT, our protocols achieve $40 \times$, $14\times$, $2\times$, and $2.2 \times$ speedups, respectively. While reducing communication cost of $\Pi_{\mathsf{softmax}}$, $\Pi_{\mathsf{ln}}$ and $\Pi_{\mathsf{gelu}}$ by $260\times$, $5.5\times$, and $3\times$, respectively. 
% \vspace{-0.3em}
    \item \textbf{Better piecewise approximation for non-linear functions}: A novel piecewise approximation method is proposed for the differentiable nonlinear functions, offering enhanced accuracy with a lower polynomial degree. Such as Tanh uses polynomials of degree 4 to match the accuracy of BOLT, which uses polynomials of degree 5.

% \vspace{-0.3em}
    \item \textbf{End-to-end private inference performance improvement}: 
    We implemented the end-to-end private TBM inference and conducted extensive experiments to evaluate the performance of our method thoroughly. \SystemName significantly reduces communication by about $50$GB and end-to-end inference time by $2.2$ to $24$ minutes across various network environments. The code of \SystemName is available on the anonymous GitHub\footnote{https://anonymous.4open.science/r/FASTLMPI}.
   
\end{enumerate}

\section{PRELIMINARIES}
\subsection{Notations}
\label{sec:notations}
                                          
In this paper, we focus on 2PC setting $\mathcal{P} \in \{\mathcal{A}, \mathcal{B}\}$. We use $\mathbf{x} := \{x_0,\ldots, x_{d_m - 1}\} \in \mathbb{Q}^{d_m}$ \footnote{$\mathbb{Q}$ represents the set of all rational numbers.} to denote a vector of length $m$, and $x_i$ represents the $i$-th element of the vector. $\mathbf{X} \in \mathbb{Q}^{d_m*d_n} $ represents a matrix, where $d_m$ denotes the number of rows, $d_n$ denotes the number of columns, and $x_{ij}$ represents the $(i, j)$-th entry of matrix $\mathbf{X}$. We use $[\![\mathbf{m}]\!]_\mathcal{P}$ denotes the message encrypted by participant $\mathcal{P}$. $\langle m\rangle_\mathcal{P}$ and $\langle m \rangle^B_\mathcal{P}$ denote additive sharing and boolean sharing of $m$ held by $\mathcal{P}$, respectively. We use $\boxplus$ and $\boxdot$ to denote homomorphic addition and homomorphic multiplication computations. For plaintext, we denote the Hadamard product by $\odot$, element-wise addition by $\oplus$, matrix multiplication (or vector inner product) by $\otimes$, and XOR is denoted using the symbol $\wedge$. $\widetilde{\mathbf{X}}$ denotes a rearrangement of matrix $\mathbf{X}$ based on different requirements, while $\widehat{\mathbf{X}}$ denotes a copy of the $\mathbf{X}$. $\mathbb{Z}_p$ denotes prime fields ($p$ is prime), and $\mathbb{Z}_{2^k}$ denotes the ring. Given a real number $\overline{x}$, we use fixed-point encoding with $s$ bit of precision, denotes $x = \lfloor \overline{x} \cdot 2^s \rfloor$.

\subsection{TBM architecture}

In most cases, TBM uses the decoder-only
mode, exemplified by models such as $\rm BERT_\mathsf{base}$, as shown in Figure \ref{Fig_Bert}. The following section discusses the key operators of a single $\rm BERT_\mathsf{base}$ block.

\begin{figure}[h!]
    \centering
    \includegraphics[width=0.9\linewidth]{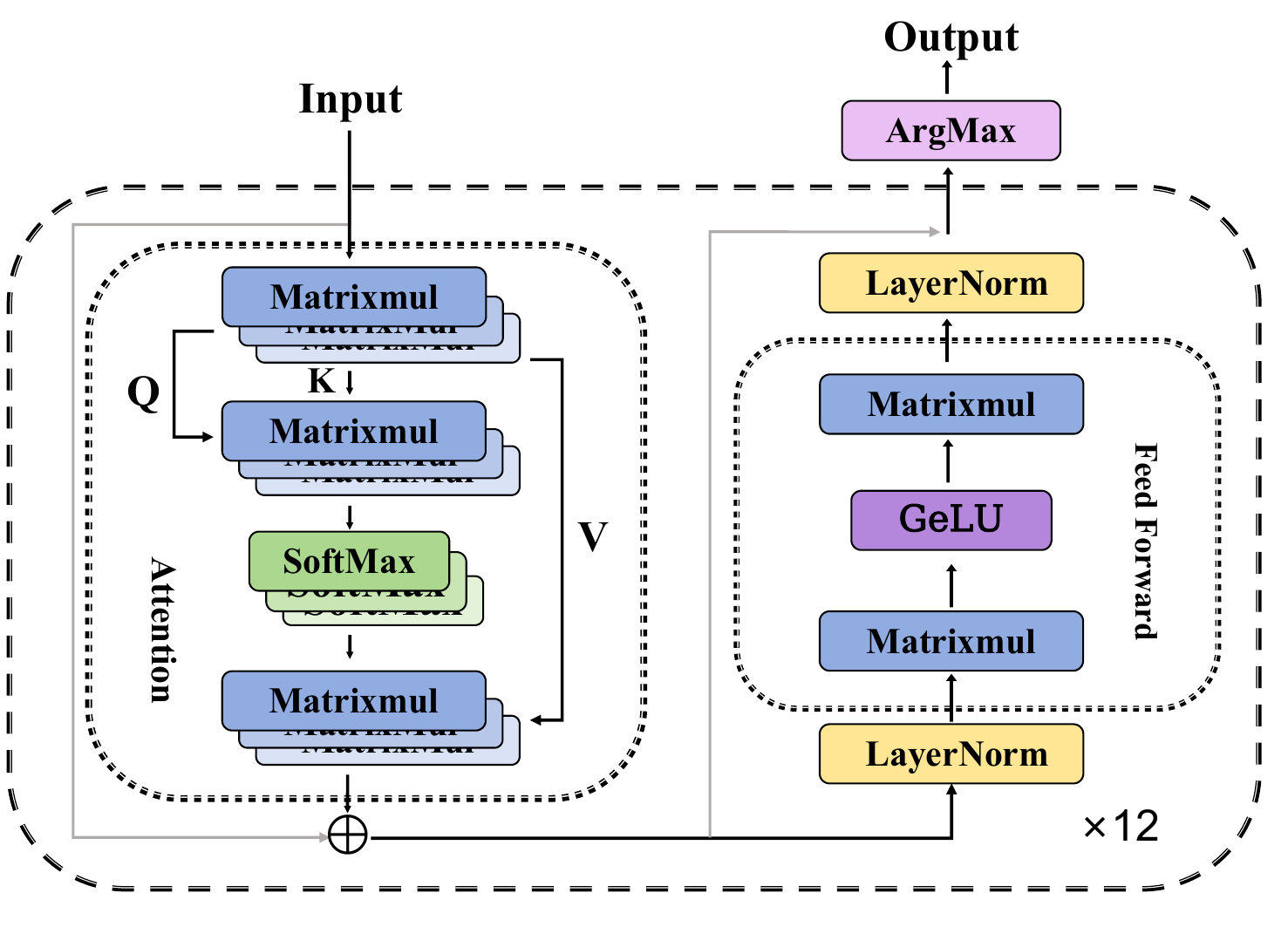}
    % \captionsetup{justification=centering}
    \caption{The architecture of $\rm BERT_\mathsf{base}$.}
    \label{Fig_Bert}
\end{figure}

\subsubsection{Attention Layer}
\label{sec:atten}

Given matrices $\mathbf{Q}$ (query), $\mathbf{K}$ (key), $\mathbf{V}$ (value) in the latent represent, This attention layer can be formalized as Eq.  
\eqref{eq:attention} :
\begin{equation}
\label{eq:attention}
Attention(\mathbf{Q,K,V})=SoftMax\left(\frac{\mathbf{Q \otimes K}^T}{\sqrt{d_k}}\right)\mathbf{V}
\end{equation}
where $\mathbf{Q, K, V} \in \mathbb{Q}^{d_s \times d_k}$. The SoftMax function applies to each row vector of a matrix. For a vector $\mathbf{x}$ $\in \mathbb{Q}^{d_s}$. The SoftMax on $\mathbf{x}$ returns a vector $\mathbf{y}$ $ \in \mathbb{Q}^{d_s}$ can be described as Eq. \eqref{eq:Softmax}:
\begin{equation}
\label{eq:Softmax}
\mathbf{y} = SoftMax(x_i) = \frac{e^{x_i}}{\sum{_{j=0}^{k-1}}{e^{x_j}}}
\end{equation}

The multi-head attention can be further extended by employing the attention mechanism. Considering $h$ different attention mechanisms with weight matrices $\mathbf{W}_i^Q, \mathbf{W}_i^K, \mathbf{W}_i^V$ for $i = 1, ..., h$, we have
\begin{equation}
\label{eq:head}
\mathbf{h}_i = Attention (\mathbf{Q \otimes W}_i^Q, \mathbf{K \otimes W}_i^K, \mathbf{V \otimes W}_i^V)
\end{equation}
where $\mathbf{W}_i^Q, \mathbf{W}_i^K, \mathbf{W}_i^V \in \mathbb{Q}^{d_k \times d_m}$, and $h_i \in \mathbb{Q}^{d_s \times d_m}$. Then, the multi-head attention can be formalized as Eq. \eqref{eq: Multi-attention}:
\begin{equation}
\label{eq: Multi-attention}
MH(\mathbf{Q, K, V}) = (\mathbf{h}_1||\mathbf{h}_2|| \ldots||\mathbf{h}_h) \odot \mathbf{W}^O
\end{equation}
where "||" denotes the horizontal concatenation of matrices, $\mathbf{W}^O \in \mathbb{Q}^{hd_k \times d_m}$, and MH($\mathbf{Q, K, V}$) $\in \mathbb{Q}^{d_s \times d_m}$

\subsubsection{LayerNorm Layer}\label{sec:LNL}
Let $ \mathbf{x} := (x_1, x_2, ..., x_k)$ be a vector representation of an input of size $d_k$ to normalization layers.
LayerNorm re-centers and re-scales input $\mathbf{x \in \mathbb{Q}^{d_k}}$ as Eq.\eqref{eq: LayerNorm}:
\begin{equation}
\label{eq: LayerNorm}
LayerNorm(x_i) = \mathbf{\gamma} \cdot \frac{x_i - \mu}{\sigma} + \mathbf{\beta}
\end{equation}
where $\mu = \frac{1}{N} \sum_{i=1}^{N} x_i $, and $\sigma = \sqrt{\frac{1}{N} \sum_{i=1}^{N} {(x_i - \mu)}^2}$. $\mu$ and $\sigma$ are the mean and standard deviation of $\mathbf{x}$. 

\subsubsection{Feed-forward Layer (FFN)}\label{sec: FFN} The FFN consists of two fully connected layers and is separated by a GeLU activation. Let $\mathbf{X} \in \mathbb{Q}^{d_{s \times d_m}}$ be an input, where $d_s$ is the sequence length. The formula is as Eq. \eqref{eq:FFN}:
\begin{equation}
\label{eq:FFN}
FFN(\mathbf{X})  = GeLU(\mathbf{X \otimes W}^{F1} \oplus \mathbf{B}^{F1}) \otimes \mathbf{W}^{F2} \oplus \mathbf{B}^{F2}
\end{equation}
where $\mathbf{W}^{F1} \in \mathbb{Q}^{d_m \times d_f}$, $\mathbf{W}^{F2} \in \mathbb{Q}^{d_m} \times d_f$ are the weight matrices and $\mathbf{B}^{F1} \in \mathbb{Q}^{d_f}$ $\mathbf{B}^{F2} \in
\mathbb{Q}^{d_m}$ are the bias vectors for the two layers of FFN. GeLU is the Gaussian
Error Linear Unit activation function, which returns a vector $\mathbf{y} \in \mathbb{Q}^{d_k}$ for the input vector $\mathbf{x} \in \mathbb{Q}^{d_k}$, as the Eq. \eqref{eq: GeLU}:
\begin{equation}
\label{eq: GeLU}
% GeLU(x_i) = \frac{x_i}{2}
% \cdot (1 + tanh(\sqrt{\frac{2}{\pi}} \cdot (x_i + 0.044715 \cdot x_i^3)))
GeLU(x_i) = x_i \cdot \Phi(x_i) = x_i \cdot \frac{1}{2} \left[ 1 + erf\left(\frac{x_i}{\sqrt{2}}\right)\right]
\end{equation}
where $\Phi(x_i)$ is the standard Gaussian cumulative distribution function.

\subsection{Threat Model}\label{sec:thread_model}
This work considers two-party computation secure against a semi-honest adversary~\cite{FIT22024}. The security is provided in simulation paradigm~\cite{lindell2017simulate} against static and semi-honest probabilistic polynomial-time (PPT) adversary $\mathsf{Adv}$. That is, the adversary $\mathsf{Adv}$ passively corrupts either $\mathcal{A}$ or $\mathcal{B}$ at the beginning of the protocol to learn something more than expected during interactions but honestly follows the protocol specification. The formal definition of the semi-honest security model is that for protocol $\Pi$ (corresponding to the functionality $\mathcal{F}$), we assume that there is a simulation $\mathsf{Sim}$ can build a real world where the views (including all the intermediate values and outputs) of adversary $\mathsf{Adv}$ are computationally indistinguishable with the views in the real world:

\begin{equation}
\label{eq:semihonest}
\begin{split}
    \{View_\mathsf{Adv}^\Pi(\lambda, \langle x \rangle_\mathcal{A}, \langle x \rangle_\mathcal{B}),
    output^{\Pi}(\mathbf{y})\}
    \stackrel{\text{c}}{\approx} \\
    \mathcal{S}(k, \{\langle x \rangle_\mathcal{P}, \mathcal{F}(\langle x \rangle_\mathcal{A}, \langle x \rangle_\mathcal{B}
     ))
\end{split}
\end{equation}
where $\lambda$ is the secure parameter. $View_\mathsf{Adv}^\Pi$ is the view of $\mathcal{P}$ in the execution of $\Pi$ on $x$, $output^\Pi$ is the output of all parties. A detailed security proof of our proposed protocol is given in Appendix \ref{appendix: proof}.

\subsection{Cryptographic Primitives}

\subsubsection{Fully Homomorphic Encryption}

Fully Homomorphic Encryption (FHE) is first proposed by Rivest et al. \cite{FHEgentry1978}. We chose the Brakerski-Fan-Vercauteren (BFV) scheme \cite{BFV2012, BFV2016}, whose security is based on the Ring Learning With Errors (RLWE) problem proposed by  Lyubashevsky et al. \cite{RLWE02010}. The BFV scheme has five algorithms, 
(\textbf{KeyGen}, \textbf{Encrypt}, \textbf{Decrypt}, \textbf{HAdd}, \textbf{HMult}, \textbf{Square}). \textbf{KeyGen} generates the keys used in the FHE scheme given the parameters chosen. \textbf{Encrypt} and \textbf{Decrypt} are the encyption and decryption algorithms respectively. \textbf{HAdd}, \textbf{HMult}, \textbf{Square} denote homomorphic addition, homomorphic multiplication operations and homomorphic squaring operations, respectively. A formal description with complete details can be found in~\cite{SOMEWHAT2012}. BFV scheme supports packing into ciphertext in a single-instruction-multiple-data (SIMD) \cite{SIMD2014}. We adopted this technique to encrypt the messages.

\subsubsection{Secret Sharing Scheme}

Secret Sharing (SS) \cite{ABY32018}: Allowing a set of $2$ participants $\mathcal{P} \in \{\mathcal{A}, \mathcal{B}\}$ to share a piece of information in such a way that only authorized subsets of the participants can recover the secret.

We have the \textbf{Additive Scheme}: For every secret $s \in \mathbb{Z}_{2^k}$, an additive scheme satisfies $s =\langle s \rangle_\mathcal{A} + \langle s \rangle_\mathcal{B}$ mod $2^k$. When $k=1$, the boolean secret sharing of $s$ can be represented as $\langle s \rangle^B_\mathcal{P}$, satisfying $s  =\langle s \rangle^B_\mathcal{A}$ $\wedge $$\langle s \rangle^B_\mathcal{B}$ mod $\mathbb{Z}_2$.

\subsubsection{Conversion between $\mathbb{Z}_\mathsf{2^k}$ and $\mathbb{Z}_\mathsf{p}$}

FASTLMPI uses fixed-point number expression (Similar to \cite{Privformer2023}, \cite{Pencilfp2024}, \cite{ScaleMPCfp2024}) and is based on the collaborative computation of HE and SS. It is worth noting that the HE is over $\mathbb{Z}_{p}$, while computation in SS is performed over $\mathbb{Z}_{2^k}$. We need to convert between $\mathbb{Z}_{p}$ and $\mathbb{Z}_{2^k}$. In detail, converting from $\mathbb{Z}_{p}$ to $\mathbb{Z}_{2^k}$ requires comparison and a multiplexer. When converting from $\mathbb{Z}_{2^k}$ to $\mathbb{Z}_{p}$, the probability that the sharing overflows is $\frac{|x|}{2^k}$. If $x$ is small or the ring size is large enough, the conversion from $\mathbb{Z}_{2^k}$ to $\mathbb{Z}_{p}$ can be omitted without affecting the accuracy.

\section{SYSTEM DESCIRPTION}
\label{sec:Sysdescirption}

In this section, we introduce the building blocks in FASTLMPI, which include four 2PC protocols: Secure Matrix Multiplication $\Pi_{\mathsf{matmul}}$, secure Softmax $\Pi_\mathsf{softmax}$, secure LayerNorm $\Pi_\mathsf{ln}$, fitting method of differentiable non-linear functions and secure GeLU $\Pi_\mathsf{gelu}$. FASTLMPI is built in a 2-PC scenario where Alice ($\mathcal{A}$) represents the client and Bob ($\mathcal{B}$) represents the server. It relies on the fact that each party $\mathcal{P}$ holds their data as input, and Alice and Bob can generate their private and public keys in the initial phase. FASTLMPI completes the protocol calculation when Alice and Bob interact to obtain the final privacy inference result.

\subsection{Secure Matrix multiplication $\Pi_\mathsf{matmul}$}

This section details a secure matrix multiplication protocol. In the initial stage, party $\mathcal{A}$ holds the input matrix $\mathbf{X} \in \mathbb{Z}_p^{d_m*d_n}$, party $\mathcal{B}$ holds the parameter $\mathbf{W}^I \in \mathbb{Z}_p^{d_n*d_h}$, where $I \in \{Q, K, V\}$. In $\rm BERT_\mathsf{base}$, $\mathbf{X}$ and $\mathbf{W}^I$ can be packaged and directly implemented by SIMD-based homomorphic ciphertext-plaintext multiplication to obtain $\mathbf{Q} \in \mathbb{Z}_p^{d_m*d_h}$, $\mathbf{K} \in \mathbb{Z}_p^{d_m*d_h}$ and $\mathbf{V} \in \mathbb{Z}_p^{d_m*d_h}$, which are in SIMD ciphertext form. However, $\mathbf{Q}$, $\mathbf{K}$, and $\mathbf{V}$ can be constructed in the form of additive secret sharing. $\mathcal{B}$ generates the uniformly random $\mathbf{R} \in \mathbb{Z}_p^{d_m*d_h}$ locally and constructs $\langle \mathbf{Q} \rangle_\mathcal{A}= \mathbf{Q} \boxminus \mathbf{R}$ , then $\langle \mathbf{Q} \rangle_\mathcal{B} = \mathbf{R}$. $\langle \mathbf{K} \rangle_\mathcal{A}$ and $\langle \mathbf{K} \rangle_\mathcal{B}$ can be obtained using the same way. Then $\mathbf{M} \in \mathbb{Z}_p^{d_m*d_m}= \mathbf{Q} \otimes \mathbf{K}^T$\footnote{$T$ denotes the transpose of the matrix} can be expressed as:

\begin{equation}
\begin{split}
\mathbf{M} &= \overbrace{\left(\langle \mathbf{Q} \rangle_\mathcal{A} \otimes \langle \mathbf{K}^T \rangle_\mathcal{A} \right)}^{\text{$\mathcal{A}$ can compute locally}} \oplus   \overbrace{\left( \langle \mathbf{Q} \rangle_\mathcal{B} \otimes \langle \mathbf{K}^T \rangle_\mathcal{B} \right)}^{\text{$\mathcal{B}$ can compute locally}}\\ 
&\quad \oplus \overbrace{\left( \langle \mathbf{Q} \rangle_\mathcal{B} \otimes \langle \mathbf{K}^T \rangle_\mathcal{A} \right)}^{\text{$\mathcal{A}$ and $\mathcal{B}$ compute interactively}}
 \oplus \overbrace{\left( \langle \mathbf{Q} \rangle_\mathcal{A} \otimes \langle \mathbf{K}^T \rangle_\mathcal{B} \right)}^{\text{$\mathcal{A}$ and $\mathcal{B}$ compute interactively}}\\ 
\end{split}
\label{eq: Matmul}
\end{equation}

It's worth noting that we need to deal with the part of interactive computation such as $\mathbf{X} \otimes \mathbf{W}$, $\langle \mathbf{Q} \rangle _\mathcal{A} \otimes  \langle \mathbf{K}^T \rangle_\mathcal{B}$ and $\langle \mathbf{Q} \rangle _\mathcal{B} \otimes  \langle \mathbf{K}^T \rangle_\mathcal{A} $. In subsequent descriptions, we will uniformly describe the matrices involved in the interactive computation: Alice holds matrix $\mathbf{A}\in \mathbb{Z}_p^{d_m*d_n} $, and Bob holds matrix $\mathbf{B} \in \mathbb{Z}_p^{d_n*d_k}$. $\mathbf{A}$ and $\mathbf{B}$ as illustrated in Eq. \eqref{eq: MatAB}.

\begin{equation}\scriptsize
\mathbf{A}=\begin{tikzpicture}[baseline=(m.center)]
 \matrix (m) [matrix of nodes, left delimiter={\{}, right delimiter={\}}] {
$a_{00}$ & $a_{01}$ & $\cdots$ & $a_{0n}$ & $= \mathbf{A}_1$ \\
$a_{10}$ & $a_{11}$ & $\cdots$ & $a_{1n}$ & $\vdots$ \\
$\vdots$ & $\vdots$ & $\ddots$ & $\vdots$ & $\vdots$\\
$a_{m0}$ & $a_{m1}$ & $\cdots$ & $a_{mn}$ & $= \mathbf{A}_m$\\
};
\draw[gray, thick] (m-1-1.north west) rectangle (m-1-4.south east); 
\draw[gray, thick] (m-4-1.north west) rectangle (m-4-4.south east); 
 \end{tikzpicture}
,
\mathbf{B}=\begin{tikzpicture}[baseline=(m.center)]
\matrix (m) [matrix of nodes, left delimiter={\{}, right delimiter={\}}] {
$b_{00}$ & $b_{01}$ & $\cdots$ & $b_{0k}$ \\
$b_{10}$ & $b_{11}$ & $\cdots$ & $b_{1k}$ \\
$\vdots$ & $\vdots$ & $\ddots$ & $\vdots$ \\
$b_{n0}$ & $b_{n1}$ & $\cdots$ & $b_{nk}$ \\
};
% \draw[red, thick] (m-1-1.north west) rectangle (m-4-1.south east);
\end{tikzpicture}
\label{eq: MatAB}
\end{equation}

$\mathbf{A}$ is first flattened row-wise to obtain $\mathbf{A}_i$, $i \in \{1, \ldots, m\}$, which is then copied $k$ times to form $\widetilde{\mathbf{A}}$. Meanwhile $\mathbf{B}$ copied $m$ times form $\widehat{\mathbf{B}}$. $\widetilde{\mathbf{A}}$ and $\widehat{\mathbf{B}}$ as shown in Eq. \eqref{eq: Mcopy} and Eq. \eqref{eq: Hcopy}, respectively.

\begin{equation}\footnotesize
\widetilde{\mathbf{A}} = \begin{bmatrix}
(\mathbf{A}_1^T  & \cdots &\mathbf{A}_1^T )_{d_k}&
(\mathbf{A}_2^T 
\cdots \mathbf{A}_2^T )_{d_k}&
(\mathbf{A}_m^T 
\cdots \mathbf{A}_m^T)_{d_k}
\end{bmatrix}
\label{eq: Mcopy}
\end{equation}

\begin{equation}\footnotesize
\widehat{\mathbf{B}} = \begin{bmatrix}
(\mathbf{B} & \mathbf{B} & \cdots &\mathbf{B})_{d_m}
\end{bmatrix}
\label{eq: Hcopy}
\end{equation}

As $\mathbf{A}$ is $\mathcal{A}$'s private data, it cannot be directly sent to $\mathcal{B}$. To preserve data confidentiality, we need to encrypt $\widetilde{\mathbf{A}}$ and send $[\![\widetilde{\mathbf{A}}]\!]_\mathcal{A}$ to $\mathcal{B}$. $\mathcal{B}$ computes $\sum_{i=0}^m ([\![\widetilde{a}_{ij}]\!]_\mathcal{A}\times \widehat{b}_{ij})$ to get $[\![\mathbf{C}]\!]_\mathcal{A}$, which is shown in Eq. \eqref{eq: MatMH}. 

\begin{equation}\footnotesize
[\![\mathbf{C}]\!]_\mathcal{A} =\begin{tikzpicture}[baseline=(m.center)]
 \matrix (m) [matrix of nodes, left delimiter={\{}, right delimiter={\}}] {
$[\![\widetilde{a}_{00}]\!]_\mathcal{A}  \widehat{b}_{00}$ & $[\![\widetilde{a}_{01}]\!]_\mathcal{A}  \widehat{b}_{01}$ & $\cdots$ & $[\![\widetilde{a}_{0 (m \times k)}]\!]_\mathcal{A}  \widehat{b}_{0 (m \times k)}$\\
 &  & $\boxplus$ &  \\
$[\![\widetilde{a}_{10}]\!]_\mathcal{A}  \widehat{b}_{10}$ & $[\![\widetilde{a}_{11}]\!]_\mathcal{A}  \widehat{b}_{11}$ & $\cdots$ & $[\![\widetilde{a}_{1(m \times k)}]\!]_\mathcal{A}  \widehat{b}_{1 (m \times k)}$\\
 &  & $\boxplus$ &  \\
$\vdots$ & $\vdots$ & $\ddots$ & $\vdots$ \\
 &  & $\boxplus$ &  \\
$[\![\widetilde{a}_{m0}]\!]_\mathcal{A}  \widehat{b}_{m0}$ & $[\![\widetilde{a}_{m1}]\!]_\mathcal{A}  \widehat{b}_{m1}$ & $\cdots$ & $[\![\widetilde{a}_{m (m \times k)}]\!]_\mathcal{A}  \widehat{b}_{m (m \times k)}$\\
};
% \draw[red, thick] (m-1-1.north west) rectangle (m-1-4.south east); 
\end{tikzpicture}
\label{eq: MatMH}
\end{equation}

Fig \ref{Fig: Boxmatmul} shows a toy example with $\mathbf{A}^{3*4}$ and $\mathbf{B}^{4*2}$.

\begin{figure}[h]
    \centering
    \includegraphics[width=0.9\linewidth]{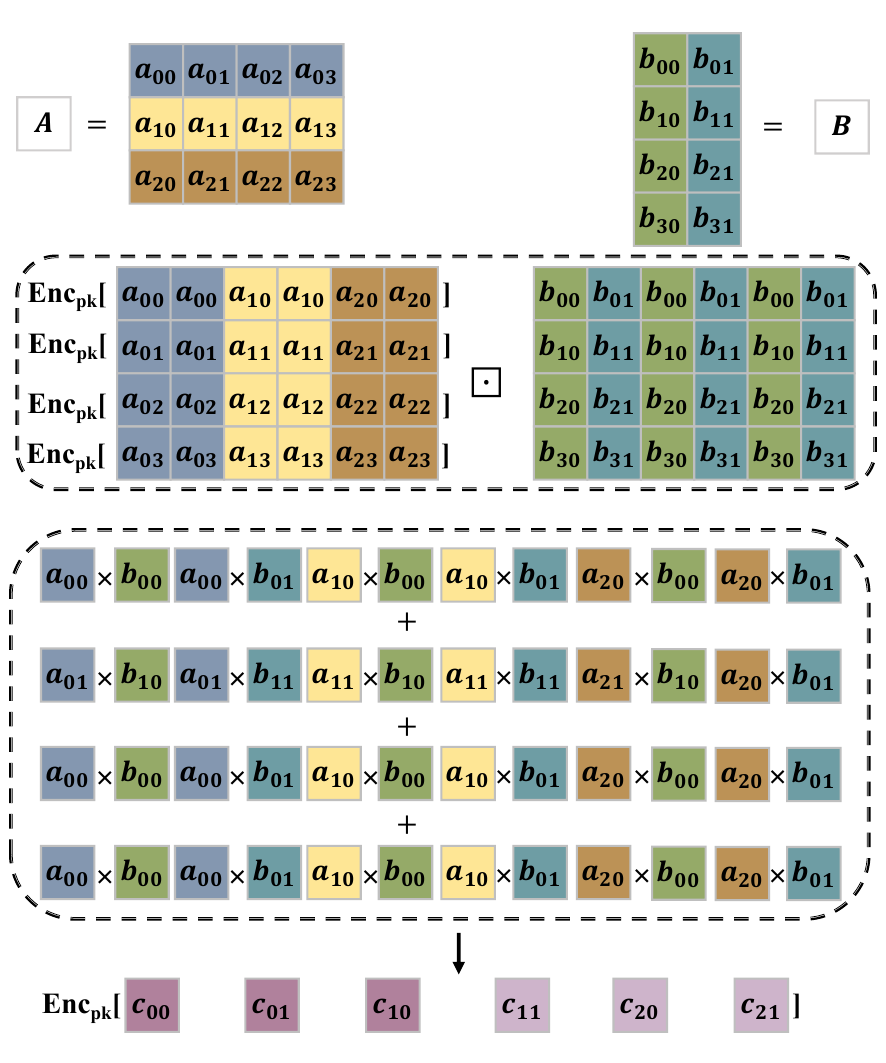}
     \caption{A toy example of SIMD-based matrix multiplication}
    \label{Fig: Boxmatmul}
\end{figure}
The matrix multiplication result obtained using protocol $\Pi_\mathsf{matmul}$ is a SIMD-based ciphertext without any wasted ciphertext slots. Compared with plaintext matrix multiplication, no additional multiplication operations are introduced. Employing the method mentioned earlier, the SIMD-based ciphertext $[\![\mathbf{C}]\!]_\mathcal{A}$ can be expressed as an additive secret sharing. This ensures that subsequent matrix multiplications can still be efficiently performed using homomorphic ciphertext-plaintext multiplication while remaining rotate-free. The protocol process is shown in Fig. \ref{pro: matmul}. Now, $\langle \mathbf{Q} \rangle_\mathcal{A}$, $\langle \mathbf{K} \rangle_\mathcal{A}$ and $\langle \mathbf{Q} \rangle_\mathcal{B}$, $\langle \mathbf{K} \rangle_\mathcal{B}$ can be computed by $\Pi_\mathsf{matmul}(\mathbf{X} \otimes \mathbf{W}_Q)$, $\Pi_\mathsf{matmul}(\mathbf{X} \otimes \mathbf{W}_K)$, respectively. Alice and Bob locally compute $\langle \mathbf{Q} \rangle_\mathcal{A} \odot \langle \mathbf{K} \rangle_\mathcal{A})$ and $\langle \mathbf{Q} \rangle_\mathcal{B} \odot \langle \mathbf{K} \rangle_\mathcal{B}$, respectively. Subsequently, Bob converts $\langle \mathbf{Q} \rangle_\mathcal{B} \odot \langle \mathbf{K} \rangle_\mathcal{B}$ to additive secret-sharing $\langle \mathbf{L} \rangle _\mathcal{A} = (\langle \mathbf{Q} \rangle_\mathcal{B} \odot \langle \mathbf{K} \rangle_\mathcal{B}) \ominus \mathbf{E}$ and $\langle \mathbf{L} \rangle _\mathcal{B} = \mathbf{E}$ ($\mathbf{E}$ is a uniformly random matrix sampled from $\mathbb{Z}_{p}$, and has the same dimension with $\langle \mathbf{Q} \rangle_\mathcal{B} \odot \langle \mathbf{K} \rangle_\mathcal{B}$). Then Alice can obtain $\langle \mathbf{M} \rangle_\mathcal{A}$:
\begin{equation}\begin{split}
        \langle \mathbf{M} \rangle_\mathcal{A} &= (\langle \mathbf{Q} \rangle_\mathcal{A} \odot \langle \mathbf{K} \rangle_\mathcal{A}) \oplus \Pi_\mathsf{matmul}(\langle \mathbf{Q}\rangle_\mathcal{A} \otimes \langle \mathbf{K}\rangle_\mathcal{B}) \\
        &
        \oplus \Pi_\mathsf{matmul}(\langle  \mathbf{Q} \rangle_\mathcal{B} \boxtimes \langle \mathbf{K}\rangle_\mathcal{A}) \oplus \langle \mathbf{L} \rangle _\mathcal{A} 
\end{split}
\end{equation}
and Bob can obtain $\langle \mathbf{M} \rangle_\mathcal{B}$. :
\begin{equation}\begin{split}
        \langle \mathbf{M} \rangle_\mathcal{B} &=   \Pi_\mathsf{matmul}(\langle \mathbf{Q}\rangle_\mathcal{A} \otimes \langle \mathbf{K}\rangle_\mathcal{B}) 
        \oplus \Pi_\mathsf{matmul}(\langle  \mathbf{Q} \rangle_\mathcal{B} \boxtimes \langle \mathbf{K}\rangle_\mathcal{A})  \\&
        \oplus
        \langle \mathbf{L} \rangle _\mathcal{B} 
\end{split}
\end{equation}

We opt for communication over the HE rotation. A single round of communication allows HE ciphertexts to be decomposed into the form of additive secret sharing, thus eliminating the need for rotation operations. Moreover, for specific matrix sizes, the communication cost is also acceptable.

%############## input the protocol file %###############
 
\begin{figure}[h]
    \centering\scalebox{0.82}{
    \begin{tikzpicture}
        \node[line width=1pt,draw, rectangle, fill=gray!20, rounded corners = 3pt, minimum height=0.5cm, minimum width=2.4cm] at (-0.45, 5.81) 
        {Protocol $\Pi_\mathsf{matmul}$};
        \node[line width=1pt, draw, rectangle, rounded corners = 5pt, minimum height=6cm, minimum width=7cm, text width=9cm, align=left] at (2.8, 1) 
        {
            \textbf{Initialization}: \\
            ~~$\mathcal{A}$ and $\mathcal{B}$ generate each private key $sk_\mathcal{A}$, $sk_\mathcal{B}$ and public key $pk_\mathcal{A}$, $pk_\mathcal{B}$.\\
            % \textbf{Preprocessing}:\\
            % ~~$\mathcal{B}$ generates uniformly random $\mathbf{R} \in \mathbb{Z}_p^{d_{m}*d_{h}}$.\\
            \textbf{Private inputs}: \\
            ~~$\bm{\mathcal{A}}$: $\mathbf{A} \in \mathbb{Z}_p^{d_{m}*d_{n}}$\\
            ~~$\bm{\mathcal{B}}$:  $\mathbf{B} \in \mathbb{Z}_p^{d_{n}*d_{h}}$\\
            \textbf{Protocol}:
            \begin{itemize}
            \resetlinenumber
            \internallinenumbers
                \item[$\bm{\mathcal{A}}$]: Flattens $\mathbf{A}$ by rows and copys $k$ times to obtain $\widetilde{\mathbf{A}}$; Encrypts each row of $\widetilde{\mathbf{A}}$ by $pk_\mathcal{A}$ yields SIMD-ciphertesxts $[\![ \widetilde{\mathbf{A}}_\mathsf{i} ]\!]_\mathcal{A}$ \footnote{$i \in \{0, 1, \ldots, m\}$}; Sends $[\![ \widetilde{\mathbf{A}}_\mathsf{i} ]\!]_\mathcal{A}$ to $\mathcal{B}$.

                 \item[$\bm{\mathcal{B}}$]: Receives $[\![ \widetilde{\mathbf{A}}_\mathsf{i} ]\!]_\mathcal{A}$; Copies $ \mathbf{B}$ $m$ time to get $\widehat{\mathbf{B}}$; Computes 
                 $[\![\mathbf{C}_\mathsf{i}]\!]_\mathcal{A} = [\![ \widetilde{\mathbf{A}}_\mathsf{i} ]\!]_\mathcal{A} \boxdot \widehat{\mathbf{B}}_\mathsf{i}$; Computes $[\![\mathbf{C}]\!]_\mathcal{A}= \sum_\mathsf{i=0}^\mathsf{m} [\![\mathbf{C}_\mathsf{i}]\!]_\mathcal{A} $; Generates uniformly random $\mathbf{R} \in \mathbb{Z}_p^{d_{m}*d_{h}}$; Flattens $\mathbf{R}$ to obtain $\widetilde{\mathbf{R}}$; Computes $[\![\langle \mathbf{C} \rangle_\mathcal{A}]\!]_\mathcal{A} =$ $[\![\mathbf{C}]\!]_\mathcal{A} \boxminus \widehat{\mathbf{R}}$; Sends $[\![\langle \mathbf{C} \rangle_\mathcal{A}]\!]_\mathcal{A} $ to $\mathcal{A}$.
                 
                \item[$\bm{\mathcal{A}}$]: Receives $[\![\langle \mathbf{C} \rangle_\mathcal{A}]\!]_\mathcal{A} $.

            \end{itemize}  
            \textbf{Private outputs}: \\
            $\bm{\mathcal{A}}$: $\langle \mathbf{C} \rangle_\mathcal{A}$$ =$ $\textbf{BFV.Decrypt}(sk_\mathcal{A}$, $[\![\langle \mathbf{C} \rangle_\mathcal{A}]\!]_\mathcal{A} )$$ \in \mathbb{Z}_p$ \\
            $\bm{\mathcal{B}}$: $\langle \mathbf{C} \rangle_\mathcal{B}= \widehat{\mathbf{R}}\in \mathbb{Z}_p$\\
        };
    \end{tikzpicture}
}
    \caption{Secure matrix multiplication $\Pi_\mathsf{matmul}$}
    \label{pro: matmul}
\end{figure}

%############## input the protocol file%###############

\subsection{Secure SoftMax $\Pi_\mathsf{softmax}$}

This subsection describes the secure SoftMax $\Pi_\mathsf{softmax}$. As we know, mixing ring and prime shares will perform better. ie., evaluate non-linear function over ring ($\mathbb{Z}_{2^k}$) shares, but switch to prime field ($\mathbb{Z}_p$) shares for arithmetic computations. Therefore, before computing function $exp$, we convert the shares from $\mathbb{Z}_p$ to $\mathbb{Z}_{2^k}$. In subsequent calculations, it is assumed that the computational domain requires conversion, and we omit the description of the conversion process.

To describe the protocol implementation, assuming shares are already converted to $\mathbb{Z}_{2^k}$ from $\mathbb{Z}_p$, $\mathcal{A}$ and $\mathcal{B}$ hold $\langle \mathbf{X} \rangle_\mathcal{A}\in \mathbb{Z}_{2^k}^{d_{m}*d_{m}}$ and $\langle \mathbf{X} \rangle_\mathcal{B}\in \mathbb{Z}_{2^k}^{d_{m}*d_{m}}$, respectively. We can directly call the $\Pi_\mathsf{rExp}$ protocol in SIRNN \cite{SIRNN2021} to compute the $exp$ function. $\mathcal{A}$ and $\mathcal{B}$ gets $\langle exp(\mathbf{X} )\rangle_\mathcal{A}\in \mathbb{Z}_{p}^{d_{m}*d_{m}}$ and $\langle exp(\mathbf{X} )\rangle_\mathcal{B}\in \mathbb{Z}_{p}^{d_{m}*d_{m}}$, respectively. Let $\mathcal{B}$ flatten $\langle exp(\mathbf{X} )\rangle_\mathcal{B}$ row-wise, As shown in Eq. \eqref{eq: Exprow}:

\begin{equation}\footnotesize 
 [exp(\mathbf{X} )]_\mathcal{B}
=
\begin{tikzpicture}[baseline=(m.center)]
\matrix (m) [matrix of nodes, left delimiter={\{}, right delimiter={\}}] {
$\langle e^{ x_{00} }\rangle_\mathcal{B}$ & $\langle e^{ x_{01} }\rangle_\mathcal{B}$ & $\cdots$ &$\langle e^{ x_{0k} }\rangle_\mathcal{B}$ & $= \left(\langle exp(\mathbf{X} )\rangle_\mathcal{B}\right)_{1}$ \\
$\langle e^{ x_\mathsf{10} }\rangle_\mathcal{B}$ & $\langle e^{ x_{11} }\rangle_\mathcal{B}$ & $\cdots$ & $\langle e^{ x_{1k} }\rangle_\mathcal{B}$ & $= \left(\langle exp(\mathbf{X} )\rangle_\mathcal{B}\right)_{2}$ \\
$\vdots$ & $\vdots$ & $\ddots$ & $\vdots$ \\
$\langle e^{ x_{n0} }\rangle_\mathcal{B}$ & $\langle e^{ x_{n1} }\rangle_\mathcal{B}$ & $\cdots$ & $\langle e^{ x_{nk} }\rangle_\mathcal{B}$ & $= \left(\langle exp(\mathbf{X} )\rangle_\mathcal{B}\right)_{m}$ \\
};
\draw[gray, thick] (m-1-1.north west) rectangle (m-1-4.south east); 
\end{tikzpicture}
\label{eq: Exprow}
\end{equation}
we have: 

\begin{equation}\footnotesize 
 \widetilde{\langle exp(\mathbf{X})\rangle_\mathcal{B}}
=
\begin{tikzpicture}[baseline=(m.center)]
\matrix (m) [matrix of nodes, left delimiter={\{}, right delimiter={\}}] {
$\left(\langle exp(\mathbf{X} )\rangle_\mathcal{B}\right)_{1}$ & $\left(\langle exp(\mathbf{X} )\rangle_\mathcal{B}\right)_{2}$& $\cdots$ & $\left(\langle exp(\mathbf{X} )\rangle_\mathcal{B}\right)_{m}$\\
};
\end{tikzpicture}
\label{eq: Expvector}
\end{equation}

Then $\mathcal{B}$ encrypts $\widetilde{\langle exp(\mathbf{X})\rangle_\mathcal{B}}$, the result can be simply expressed as $[\![\widetilde{\langle \mathbf{E}\rangle_\mathcal{B}}]\!]_\mathcal{B}$. Party $\mathcal{B}$ sends the SIMD-based ciphertexts $[\![\widetilde{\langle \mathbf{E}\rangle_\mathcal{B}}]\!]_\mathcal{B}$ to $\mathcal{A}$, which flattens $\langle exp(\mathbf{X} )\rangle_\mathcal{A}$ row-wise similar to $\mathcal{B}$ to get $\widetilde{\langle exp(\mathbf{X})\rangle_\mathcal{A}}$. Next step $\mathcal{A}$ computes $[\![\widetilde{\mathbf{E}}]\!]_\mathcal{B}$ $=[\![\widetilde{\langle \mathbf{E}\rangle_\mathcal{B}}]\!]_\mathcal{B} \boxplus  \widetilde{\langle exp(\mathbf{X})\rangle_\mathcal{A}}$, which serves as the numerator of the Softmax function.

To compute the denominator, $\mathcal{A}$ generates a uniformly random matrix $\mathbf{R} \in \mathbb{Z}_p^{d_{m}* d_{h}}$ and then flattens it row-wise using a similar pattern to obtain $\widetilde{\mathbf{R}}$. After that, $\mathcal{A}$ computes $[\![\widetilde{\mathbf{E}}]\!]_\mathcal{B} \boxplus \widetilde{\mathbf{R}}$ and the row sum of $\mathbf{R}$, and the resulting vector is denoted as $\mathbf{SR} = \sum_{i=0}^m (\mathbf{R})_{i}$. Subsequently, $\mathcal{A}$ encrypts $\mathbf{SR}$, sends $[\![\widetilde{\mathbf{E}} \oplus \widetilde{\mathbf{R}}]\!]_\mathcal{B}$ and $[\![ \mathbf{SR} ]\!]_\mathcal{A}$ to $\mathcal{B}$. To achieve this, $\mathcal{B}$ first decrypts the SIMD-based ciphertext $[\![\widetilde{\mathbf{E}} \oplus \widetilde{\mathbf{R}}]\!]_\mathcal{B}$ by $sk_\mathcal{B}$ and restoring it to the original dimensions to get $(\mathbf{E}\oplus\mathbf{R}) \in \mathbb{Z}_p^{d_{m} * d_{m}}$. Next, the row sums of $\mathbf{E}\oplus\mathbf{R}$ are calculated, yielding $\sum_{i=0}^m \left(\mathbf{E}_{i} \oplus \mathbf{R}_{i} \right)$, which is expressed as Eq. \eqref{eq: expsum}:

\begin{equation}\footnotesize 
 \sum \mathbf{E} \oplus \sum \mathbf{R}
=
\begin{tikzpicture}[baseline=(m.center)]
\matrix (m) [matrix of nodes, left delimiter={\{}, right delimiter={\}}] {
$e_{00}  + \cdots e_{0m}$ \\
$e_{10}  + \cdots e_{1m}$ \\
$\vdots$ \\
$e_{m0}  + \cdots e_{mm}$ \\ 
};
\end{tikzpicture}
\oplus
\begin{tikzpicture}[baseline=(m.center)]
\matrix (m) [matrix of nodes, left delimiter={\{}, right delimiter={\}}] {
$r_{00}  + \cdots r_{0m}$ \\
$r_{10}  + \cdots r_{1m}$ \\
$\vdots$ \\
$r_{m0} + \cdots r_{mm}$ \\ 
};
\end{tikzpicture}
\label{eq: expsum}
\end{equation}
The denominator $ [\![\sum(\mathbf{E})]\!]_\mathcal{A} $ is then obtained by subtracting the previously computed $[\![ \mathbf{SR} ]\!]_\mathcal{A}$ from this sum.

Up to this point, we have obtained the numerator and denominator components of the SoftMax function, but direct division cannot be performed yet. Therefore, we construct $ [\![\sum(\mathbf{E})]\!]_\mathcal{A} \boxdot \mathbf{v}$, where $\mathbf{v} \in \mathbb{Z}_p ^{d_{m}}$ is another uniformly random vector generated by $\mathcal{B}$, and through an interactive round, obtain the modular inverse of the denominator at party $\mathcal{B}$. Coincidentally, the numerator $[\![\widetilde{\mathbf{E}}]\!]_\mathcal{B}$ of the Softmax function is also SIMD-based ciphertext held by $\mathcal{A}$, allowing us to obtain the final result of Softmax easily.

Specifically, party $\mathcal{B}$ first generates a uniformly random vector $\mathbf{v}\in \mathbb{Z}_p^{d_{m}}$. Simultaneously, $\mathbf{v}$ is copied $m$ times and tiled to obtain $\widehat{\mathbf{V}} \in \mathbb{Z}_p^{d_{m}*d_{m}}$. Then the encrypted $\widehat{\mathbf{V}}$ is sent to $\mathcal{A}$ along with $ [\![\sum(\mathbf{E})]\!]_\mathcal{A} \boxdot \mathbf{V}$. $\mathcal{A}$ can decrypt $ [\![\sum(\mathbf{E})]\!]_\mathcal{A} \boxdot \mathbf{V}$ and also copy $\left( \sum(\mathbf{E}) \boxdot \mathbf{V}\right)$ $m$ times to obtain $ \widehat{\left(\sum(\mathbf{E}) \boxdot \mathbf{V}\right)} \in \mathbb{Z}_p^{d_{m}*d_{m}} $, and then compute the modular inverse  $1 / \widehat{\left(\sum(\mathbf{E}) \boxdot \mathbf{V}\right)}$. Next, by calculating $1 / \widehat{\left(\sum(\mathbf{E}) \boxdot \mathbf{V}\right)}$ multiplied by $[\![\widehat{\mathbf{V}}]\!]_\mathcal{B}$, we can obtain the modular inverse of the denominator of the Softmax function, $[\![1 / \sum(\widehat{\mathbf{E}})]\!]_\mathcal{B}$. Since the numerator $[\![\widetilde{\mathbf{E}}]\!]_\mathcal{B}$ of the Softmax function is also held by $\mathcal{A}$, we can directly compute $[\![\widetilde{\mathbf{E}}]\!]_\mathcal{B}$ multiplied by $[\![1 / \sum(\widehat{\mathbf{E}})]\!]_\mathcal{B}$ to obtain the final Softmax result $[\![\mathbf{Y}]\!]_\mathcal{B}$. Finally, we can obtain the additive secret sharing $\langle \mathbf{Y} \rangle_\mathcal{A}$ and $\langle \mathbf{Y} \rangle_\mathcal{B}$ at partys $\mathcal{A}$ and $\mathcal{B}$ respectively, by adopting a method similar to that used in $\Pi_\mathsf{matmul}$. The specific steps of the protocol are shown in Fig. \ref{pro: softmax}. To facilitate understanding, a $3 \times 3$ matrix example for Softmax calculation is provided in Fig. \ref{Fig: ieSoftmax} of Appendix \ref{appendix: softmax&ln}. 

%############## input the protocol file %###############
 
\begin{figure}[h]
    \centering\scalebox{0.82}{
    \begin{tikzpicture}
        \node[line width=1pt,draw, rectangle, fill=gray!20, rounded corners = 3pt, minimum height=0.5cm, minimum width=2.4cm] at (-0.93, 9.64) 
        {Protocol $\Pi_\mathsf{softmax}$};
        \node[line width=1pt, draw, rectangle, rounded corners = 5pt, minimum height=10cm, minimum width=7cm, text width=9.8cm, align=left] at (2.8, 1) 
        {
            \textbf{Initialization}: \\
            ~~$\mathcal{A}$ and $\mathcal{B}$ generate each private key $sk_\mathcal{A}$, $sk_\mathcal{B}$ and public key $pk_\mathcal{A}$, $pk_\mathcal{B}$.\\
            % \textbf{Preprocessing}:\\
            %     ~~$\mathcal{A}$ generates a uniformly random matrix $\mathbf{R}$, $\mathbf{M} \in \mathbb{Z}_p^{d_{m}* d_{m}}$; $\mathcal{B}$ generates a uniformly random vector $\mathbf{v} \in \mathbb{Z}_p^{d_{m}}$\\
            \textbf{Private inputs}:\footnote{Before calculating function $exp$, the share is transferred from the prime field $\mathbb{Z}_p$ to Ring $\mathbb{Z}_{2^k}$, and after calculating exp, the share is transferred from $\mathbb{Z}_{2^k}$ to $\mathbb{Z}_p$.} \\
            ~~$\bm{\mathcal{A}}$: $\langle \mathbf{X} \rangle_\mathcal{A}$ $\in \mathbb{Z}_{2^k}^{d_{m}* d_{m}}$ \\
            ~~$\bm{\mathcal{B}}$: $\langle \mathbf{X} \rangle_\mathcal{B}$ $\in \mathbb{Z}_{2^k}^{d_{m}* d_{m}}$\\
            \textbf{Protocol}: 
            \begin{itemize}
            \resetlinenumber
            \internallinenumbers

            \item[$\bm{\mathcal{P}}$]\footnote{Since $\mathcal{A}$ and $\mathcal{B}$ follow the same computation steps}: Participants $\mathcal{A}$ and $\mathcal{B}$ do the same as follow: \\
            $\langle exp(\mathbf{X})\rangle_\mathcal{P}$ $=$ $\Pi_\mathsf{rExp}(\langle \mathbf{X} \rangle_\mathcal{P})$

            \item[$\bm{\mathcal{B}}$]: Flattens $\langle exp(\mathbf{X})\rangle_\mathcal{B}$ row-wise to get $\widetilde{\langle exp(\mathbf{X})\rangle_\mathcal{B}}$; Encrypts $\widetilde{\langle exp(\mathbf{X})\rangle_\mathcal{B}}$ by $pk_\mathcal{A}$ yields $[\![\widetilde{\langle exp(\mathbf{X})\rangle_\mathcal{B}} ]\!]_\mathcal{B}$ (simply expressed as $[\![\widetilde{\langle\mathbf{E}\rangle_\mathcal{B}}]\!]_\mathcal{B}$); Sends $[\![\widetilde{\langle\mathbf{E}\rangle_\mathcal{B}}]\!]_\mathcal{B}$ to $\mathcal{A}$.

            \item[$\bm{\mathcal{A}}$]: Receives $[\![\widetilde{\langle\mathbf{E}\rangle_\mathcal{B}}]\!]_\mathcal{B}$; Flattens $\langle exp(\mathbf{X})\rangle_\mathcal{A} $ row-wise to get $\widetilde{\langle exp(\mathbf{X})\rangle_\mathcal{A}}$; Computes $[\![\widetilde{\mathbf{E}}]\!]_\mathcal{B}$ = $[\![\widetilde{\langle \mathbf{E}\rangle_\mathcal{B}}]\!]_\mathcal{B}$ $\boxplus$ $\widetilde{\langle exp(\mathbf{X})\rangle_\mathcal{A}}$; Generates a uniformly random matrix $\mathbf{R} \in \mathbb{Z}_p^{d_{m}* d_{m}}$; Computes $[\![\widetilde{\mathbf{E}}]\!]_\mathcal{B} \boxplus \mathbf{R}$ and $\mathbf{SR}=\sum_\mathsf{i=0}^m\left( \mathbf{R}\right)_\mathsf{i}$; Encrypts $\mathbf{SR}$; Sends set $S_{1}=\{$ $[\![\widetilde{\mathbf{E}} \oplus \mathbf{R}]\!]_\mathcal{B} $, $[\![ \mathbf{SR} ]\!]_\mathcal{A} \}$ to $\mathcal{B}$

            \item[$\bm{\mathcal{B}}$]: Receives $S_1$ from $\mathcal{A}$; Decrypts $[\![\widetilde{\mathbf{E}} \oplus \mathbf{R}]\!]_\mathcal{B} $ by $sk_\mathcal{B}$ yields $\widetilde{\mathbf{E}} \oplus \mathbf{R}$; Computes $[\![\sum(\mathbf{E})]\!]_\mathcal{A} =$ $\sum_{i=0}^m \left(\mathbf{E}_{i}+\mathbf{R}_{i}\right)$ $  \boxminus  \left([\![ \mathbf{SR} ]\!]_\mathcal{A}\right)$; Generates a uniformly random vector $\mathbf{v} \in \mathbb{Z}_p^{d_{m}}$; Computes $ [\![\sum(\mathbf{E})]\!]_\mathcal{A} \boxdot \mathbf{v}$; $\mathbf{v}$ is copied $m$ times to get $\widehat{\mathbf{V}}$; Encrypts $\widehat{\mathbf{V}}$; Sends set $S_{2}=$ $\{ [\![ \sum(\mathbf{E}) \odot \mathbf{v}]\!]_\mathcal{A}$,  $[\![ \widehat{\mathbf{V}} ]\!]_\mathcal{B}\}$ to $\mathcal{A}$.

            \item[$\bm{\mathcal{A}}$]: Receives $S_\mathsf{2}$ from $\mathcal{B}$; Decrypts $[\![\sum(\mathbf{E}) \odot \mathbf{V}]\!]_\mathcal{A} $; Copies $\left(\sum(\mathbf{E}) \odot \mathbf{V}\right)$ also $m$ times to get $ \widehat{\left(\sum(\mathbf{E}) \boxdot \mathbf{V}\right)}$; Computes the modular inverse $1 / \widehat{\left(\sum(\mathbf{E}) \boxdot \mathbf{V}\right)}$; Computes $[\![1 / \sum(\widehat{\mathbf{E}})]\!]_\mathcal{B}$ $= 1 / \widehat{\left(\sum(\mathbf{E}) \boxdot \mathbf{V}\right)} \boxdot $ $[\![ \widehat{\mathbf{V}} ]\!]_\mathcal{B}\}$; Computes $[\![\mathbf{Y}]\!]_\mathcal{B}$ $=$ $[\![\widetilde{\mathbf{E}}]\!]_\mathcal{B}$ $ \boxdot $ $[\![1 / \sum(\widehat{\mathbf{E}})]\!]_\mathcal{B}$; Generates a uniformly random matrix $\mathbf{M} \in \mathbb{Z}_p^{d_{m}* d_{m}}$; Flattens $\mathbf{M}$ row-wise to get $\widetilde{\mathbf{M}}$; Computes $[\![\langle \mathbf{Y} \rangle_\mathcal{B}]\!]_\mathcal{B}= $ $[\![\mathbf{Y}]\!]_\mathcal{B} \boxminus $$\widetilde{\mathbf{M}}$; Sends $[\![\langle \mathbf{Y} \rangle_\mathcal{B}]\!]_\mathcal{B} $ to $\mathcal{B}$.
                 
            \item[$\bm{\mathcal{B}}$]: Receives $[\![\langle \mathbf{Y} \rangle_\mathcal{B}]\!]_\mathcal{B} $.

            \end{itemize}  
            \textbf{Private outputs}: \\
            $\bm{\mathcal{A}}$: $\langle \mathbf{Y} \rangle_\mathcal{A}= \widetilde{\mathbf{M}} \in \mathbb{Z}_p$\\
            $\bm{\mathcal{B}}$: $\langle \mathbf{Y} \rangle_\mathcal{B}$$ =$ $\textbf{BFV.Decrypt}(sk_\mathcal{B}$, $[\![\langle \mathbf{Y} \rangle_\mathcal{B}]\!]_\mathcal{B} )$$ \in \mathbb{GF}(p)$ \\

        };
 
    \end{tikzpicture}
 }
    \caption{Secure softmax $\Pi_\mathsf{softmax}$}
    \label{pro: softmax}
\end{figure}

%############## input the protocol file%###############

\subsection{secure LayerNorm $\Pi_\mathsf{ln}$}

This section describes the secure LayerNorm $\Pi_\mathsf{ln}$. Eq. \eqref{eq: LayerNorm} presents the LayerNorm function.

Given a matrix $\mathbf{X} \in \mathbb{Z}_p^{d_{m}*d_{n} }$, $x_{ij}$ is the element at the ($i$, $j$) position. Equivalently, LayerNorm can be expressed as \eqref{eq:layerNormTrans}:
\begin{equation}
\label{eq:layerNormTrans}
LayerNorm(x_{ij}) = \mathbf{\gamma}\sqrt{n} \odot \frac{a_{ij}}{\sqrt{\sum_{j=0}^n a_{ij}^2}} \oplus \mathbf{\beta}
\end{equation}
where $a_{ij} = nx_{ij} - n \mu_{i} = nx_{ij} - \sum_{j=0}^n x_{ij}$. Then it can be inferred that $\mathbf{A}$, and shown as:

\begin{equation}\footnotesize
\mathbf{A} = \begin{bmatrix}
\mathbf{A}_{1} \\
\mathbf{A}_{2} \\
\vdots \\
\mathbf{A}_{m}
\end{bmatrix}, 
\quad \mathbf{A}_{i}^T 
=\begin{bmatrix}
a_{i1}\\
a_{i2}\\
\vdots \\
a_{in} 
\end{bmatrix} 
=\begin{bmatrix}
(n-1)x_{i1} - x_{i2} - \cdots - x_{in}\\
- x_{i1} + (n-1)x_{i2} - \cdots - x_{in} \\
\vdots \\
- x_{i1} - x_{i2} - \cdots + (n-1)x_{in} 
\end{bmatrix} 
\label{eq: MHflatten}
\end{equation}

Consequently, the computation of $a_{ij}$ in Eq. \eqref{eq:layerNormTrans} can be readily computed locally and represented as a secret-sharing $\langle \mathbf{A} \rangle_\mathcal{P}\in \mathbb{Z}_{2^k}^{d_{m}*d_{n}}, \mathcal{P} \in \{\mathcal{A}, \mathcal{B}\}$. The primary bottleneck lies in the evaluation of the denominator $\sqrt{\sum_{j=0}^n a_{ij}^2}$. Specifically, after parties $\mathcal{A}$ and $\mathcal{B}$ obtain $\langle \mathbf{A}\rangle_\mathcal{A} $ and $\langle \mathbf{A}\rangle_\mathcal{B} $ respectively, they can flatten the secret sharing row-wise to form $\langle \widetilde{\mathbf{A}} \rangle_\mathcal{A}$ and $\langle \widetilde{\mathbf{A}}\rangle_\mathcal{B}$. This is convenient for encrypting the ciphertexts to the SIMD form. Subsequently, $\mathcal{B}$ encrypts $\langle \widetilde{\mathbf{A}} \rangle_\mathcal{B}$ and send $[\![\langle \widetilde{\mathbf{A}} \rangle_\mathcal{B}]\!]_\mathcal{B}$ to $\mathcal{A}$. Party $\mathcal{A}$ then sum $\langle \widetilde{\mathbf{A}} \rangle_\mathcal{A}$ and $[\![\langle \widetilde{\mathbf{A}} \rangle_\mathcal{B}]\!]_\mathcal{B}$ to obtain $[\![ \widetilde{\mathbf{A}}]\!]_\mathcal{B}$. Given the \textbf{BFV.square} operation, which is more than twice as efficient as the multiplication operation for two ciphertexts, so we employ \textbf{BFV.square} to efficiently compute $[\![ \widetilde{\mathbf{A}}^2]\!]_\mathcal{B}$. At this point, a direct summation is not feasible. $\mathcal{A}$ generates a uniformly random vector $\mathbf{R}\in \mathbb{Z}_p^{d_{m}*d_{n}} $ and computes $[\![ \widetilde{\mathbf{A}}^2]\!]_\mathcal{B} \boxplus \mathbf{R}$, as outlined in Eq. \eqref{eq: lnAandR}:

\begin{equation}\footnotesize
[\![ \widetilde{\mathbf{A}}^2 \oplus \mathbf{R}]\!]_\mathcal{B} =\begin{tikzpicture}[baseline=(m.center)]
 \matrix (m) [matrix of nodes, left delimiter={\{}, right delimiter={\}}] {
$[\![ \widetilde{\mathbf{A}}^2_{1}]\!]_\mathcal{B} $ & $[\![ \widetilde{\mathbf{A}}^2_{2}]\!]_\mathcal{B}$ & $\cdots$ & $[\![ \widetilde{\mathbf{A}}^2_{m}]\!]_\mathcal{B} $ \\
$\boxplus$& $\boxplus$ & $\cdots$& $\boxplus$\\
$  \mathbf{R}_{1} $ & $  \mathbf{R}_{2} $ & $\cdots$ & $ \mathbf{R}_{m} $ \\
};
% \draw[red, thick] (m-1-1.north west) rectangle (m-1-4.south east); 
\end{tikzpicture}
\label{eq: lnAandR}
\end{equation}
where $ \mathbf{R}_{i}$, $i\in \{1, 2, \ldots, m\}$ denotes the i-th row of matrix $\mathbf{R}$. While $[\![ \widetilde{\mathbf{A}}^2 \oplus \mathbf{R}]\!]_\mathcal{B}$ can be directly sent to $\mathbf{B}$, who can compute the row sum after decryption. However, $\mathcal{B}$ can not eliminate $\mathbf{R}$ directly. There, we let $\mathcal{A}$ compute the row sum of $\mathbf{R}$, obtaining Eq. \eqref{eq: treesumlnR}:

\begin{equation}\footnotesize
\mathbf{SR} = \begin{bmatrix}
\sum_{j=0}^{n} r_{1j} & \sum_{j=0}^{n} r_{2j} & \cdots &\sum_{j=0}^{n} r_{mj}
\end{bmatrix}
\label{eq: treesumlnR}
\end{equation}
Then, $\mathcal{A}$ encrypts $\mathbf{SR} \in \mathbb{Z}_p^{d_{m}}$ and sends $[\![\mathbf{SR}]\!]_\mathcal{A}$ to $\mathcal{B}$. Next step, $\mathcal{B}$ can compute the row sum like \eqref{eq: expsum} after decrypts $[\![ \widetilde{\mathbf{A}}^2 \oplus \mathbf{R}]\!]_\mathcal{B}$, and obtain $\sum_{j=0}^{n}\widetilde{a}_{ij}^2 \oplus \mathbf{SR}$ as Eq. \eqref{eq: treesumsquare}:

\begin{equation}\footnotesize
\sum_{j=0}^{n}\widetilde{a}_{ij}^2 \oplus \sum_{j=0}^{n} r_{ij} =\begin{tikzpicture}[baseline=(m.center)]
 \matrix (m) [matrix of nodes, left delimiter={\{}, right delimiter={\}}] {
$\sum_{j=0}^{n} \widetilde{a}_{1j}^2 $ & $\sum_{j=0}^{n}\widetilde{a}_{2j}^2$ & $\cdots$ & $\sum_{j=0}^{n} \widetilde{a}_{mj}^2$ \\
$\boxplus$& $\boxplus$ & $\cdots$& $\boxplus$\\
$\sum_{j=0}^{n} r_{1j}$ & $\sum_{j=0}^{n} r_{2j} $ & $\cdots$ &$\sum_{j=0}^{n} r_{mj}$\\
};
% \draw[red, thick] (m-1-1.north west) rectangle (m-1-4.south east); 
\end{tikzpicture}
\label{eq: treesumsquare}
\end{equation}

By subtracting $[\![\mathbf{SR}]\!]_\mathcal{A}$ from $\sum_{j=0}^{n}\widetilde{a}_{ij}^2 \oplus \mathbf{SR}$, $\mathcal{B}$ can subsequently obtain the sum of squares $\sum_{j=0}^{n}\widetilde{a}_{ij}^2$ of $[\![ \widetilde{\mathbf{A}}^2 ]\!]_\mathcal{B}$, abbreviated as $[\![\mathbf{SA}^2 ]\!]_\mathcal{A}$.

To proceed, the square root is calculated. $\mathcal{B}$ begins by generating a uniformly random vector $\mathbf{v} \in \mathbb{Z}_p^{d_{m}}$, which can be viewed as a additive sharing held by $\mathcal{B}$, denoted as $\langle \mathbf{K}\rangle_\mathcal{B} = \mathbf{v}$. $\mathcal{B}$ also can forming $[\![\mathbf{SA}^2 ]\!]_\mathcal{A} \boxminus \mathbf{v}$, which is sent to $\mathcal{A}$. $\mathcal{A}$ decrypts this to obtain $\langle \mathbf{K}\rangle_\mathcal{A} = \left(\mathbf{SA}^2 \ominus \mathbf{v}\right) \in \mathbb{Z}_p^{d_{m}}$. With $\langle \mathbf{K}\rangle_\mathcal{A}$ and $\langle \mathbf{K}\rangle_\mathcal{B}$ constituting additive secret shares over $\mathbb{Z}_p$, a conversion to the ring $\mathbb{Z}_{2^k}$ is undertaken before computing the square root. At this point, $\mathcal{A}$ and $\mathcal{B}$ can invoke $\Pi_{invsqrt}$ \footnote{The secure inverse square root protocol $\Pi_{invsqrt}$, sourced from the $\rm SCI_{lib}$ \cite{SIRNN2021}} to get the secret sharing $\langle 1 / \sqrt{\mathbf{K}}\rangle_\mathcal{A}$ and $\langle 1 / \sqrt{\mathbf{K}}\rangle_\mathcal{B}$, respectively.

Given the preceding computation, the party $\mathcal{A}$, holds the ciphertext $[\![ \widetilde{\mathbf{A}}]\!]_\mathcal{B}$. Consequently, we let party $\mathcal{B}$ copy $\langle 1 / \sqrt{\mathbf{K}}\rangle_\mathcal{B}$ $n$ times and  encrypt the additive sharing $\langle \widehat{1 / \sqrt{\mathbf{K}}}\rangle_\mathcal{B} \in \mathbb{Z}_p^{d_{m} * d_{n}}$ and send it to $\mathcal{A}$. Then party $\mathcal{A}$ also copy $\langle 1 / \sqrt{\mathbf{K}}\rangle_\mathcal{A}$ $n$ times and calculates $[\![ \widehat{1 / \sqrt{\mathbf{K}}}]\!]_\mathcal{B} = \langle \widehat{1 / \sqrt{\mathbf{K}}}\rangle_\mathcal{A} \boxplus [\![\langle \widehat{1 / \sqrt{\mathbf{K}}\rangle}_\mathcal{B}]\!]_\mathcal{B}$. Now, $\mathcal{A}$ can directly compute:

\begin{equation}
    [\![\frac{a_{i}}{\sqrt{\sum^n a_{i}^2}}]\!]_\mathcal{B} = [\![ \widetilde{\mathbf{A}}]\!]_\mathcal{B} \boxdot \left([\![ \widehat{1 / \sqrt{\mathbf{K}}}]\!]_\mathcal{B}\right)
\label{eq: lnadivsuma}
\end{equation}

%############## input the protocol file%###############

\begin{figure}[h!]
    % \centering
    \centering\scalebox{0.82}{
    \begin{tikzpicture}
        \node[line width=1pt,draw, rectangle, fill=gray!20, rounded corners = 3pt, minimum height=0.5cm, minimum width=2.4cm] at (-1.05, 11.25) 
        {Protocol $\Pi_\mathsf{ln}$};
        \node[line width=1pt, draw, rectangle, rounded corners = 5pt, minimum height=10cm, minimum width=7cm, text width=9.8cm, align=left] at (2.8, 1) 
        {
            \textbf{Initialization}: \\
                ~~$\mathcal{A}$ and $\mathcal{B}$ generate each private key $sk_\mathcal{A}$, $sk_\mathcal{B}$ and public key $pk_\mathcal{A}$, $pk_\mathcal{B}$.\\
            % \textbf{Preprocessing}:\\
            %     ~~$\mathcal{A}$ generates a uniformly random matrix $\mathbf{R}$, $\mathbf{S} \in \mathbb{Z}_p^{d_{m} * d_{n}}$; $\mathcal{B}$ generates uniformly random vector $\mathbf{v}\in \mathbb{Z}_p^{d_{m}}$, random matrix $\mathbf{M} \in \mathbb{Z}_p^{d_{m} * d_{n}}$.\\
            \textbf{Private inputs}: \\
                ~~$\bm{\mathcal{A}}$: $\langle \mathbf{X} \rangle_\mathcal{A}\in \mathbb{Z}_{2^k}^{d_{m}*d_{n}}$ \\
                ~~$\bm{\mathcal{B}}$: $\langle \mathbf{X} \rangle_\mathcal{B}\in \mathbb{Z}_{2^k}^{d_{m}*d_{n}}$, parameter $\mathbf{\gamma} \in \mathbb{Z}_p^{d_{m}}$, $\mathbf{\beta}\in \mathbb{Z}_p^{d_{m}}$ \\

            \textbf{Protocol}: 
            \begin{itemize}
            \resetlinenumber
            \internallinenumbers
                \item[$\bm{\mathcal{B}}$]: Computes $\langle \mathbf{A} \rangle_\mathcal{B}$; Flatten $\langle \mathbf{A} \rangle_\mathcal{B}$ row-wise to get $\langle \widetilde{\mathbf{A}} \rangle_\mathcal{B}$; Encrypts $\langle \widetilde{\mathbf{A}} \rangle_\mathcal{B}$ through $pk_\mathcal{B}$ to obtain $[\![\langle \widetilde{\mathbf{A}} \rangle_\mathcal{B}]\!]_\mathcal{B}$; Sends $[\![\langle \widetilde{\mathbf{A}} \rangle_\mathcal{B}]\!]_\mathcal{B}$ to $\mathcal{A}$.

                \item[$\bm{\mathcal{A}}$]: Receives $[\![\langle \widetilde{\mathbf{A}} \rangle_\mathcal{B}]\!]_\mathcal{B}$; Computes $\langle \mathbf{A} \rangle_\mathcal{A}$; Flattens it row-wise to get $\langle \widetilde{\mathbf{A}} \rangle_\mathcal{A}$; Computes $[\![ \widetilde{\mathbf{A}} ]\!]_\mathcal{B}$=$\langle \widetilde{\mathbf{A}} \rangle_\mathcal{A} \boxplus [\![\langle \widetilde{\mathbf{A}} \rangle_\mathcal{B}]\!]_\mathcal{B}$; Calls BFV.square \footnote{this function from SEAL \cite{SEAL2023}} to get $[\![ (\widetilde{\mathbf{A}})^2]\!]_\mathcal{B}$; Generates a uniformly random matrix $\mathbf{R} \in \mathbb{Z}_p^{d_{m} * d_{n}}$; Computes the sum of each row to get $\mathbf{SR} \in \mathbb{Z}_p^{d_{m}}$, $[\![ \widetilde{\mathbf{A}} ^2]\!]_\mathcal{B} \boxplus \mathbf{R}$; Encrypts $\mathbf{SR}$; Sends set $S_1 =$ $\{ [\![ \widetilde{\mathbf{A}} ^2 \oplus \mathbf{R} ]\!]_\mathcal{B} $, $[\![ \mathbf{SR}]\!]_\mathcal{A}\}$ to $\mathcal{B}$.

                \item[$\bm{\mathcal{B}}$]: Receives $S_1$; Decrypts $\{ [\![ \widetilde{\mathbf{A}} ^2 \oplus \mathbf{R} ]\!]_\mathcal{B} $; Computes row sum of $(\widetilde{\mathbf{A}} ^2 \oplus \mathbf{R} )$ to get $\sum_{j=0}^{n}\widetilde{a}_{ij}^2 \oplus \mathbf{SR}$; Computes $\sum_{j=0}^{n}\widetilde{a}_{ij}^2 $ $\oplus \mathbf{SR} $ $\boxminus [\![ \mathbf{SR} ]\!]_\mathcal{A}$ to get $[\![\sum_{j=0}^{n}\widetilde{a}_{ij}^2]\!]_\mathcal{A}$; Generates uniformly random vector $\mathbf{v}\in \mathbb{Z}_p^{d_{m}}$; Computes $[\![\sum_{j=0}^{n}\widetilde{a}_{ij}^2]\!]_\mathcal{A} \boxminus \mathbf{v}$; $\mathcal{B}$ has $\langle \mathbf{K}\rangle_\mathcal{B}=\mathbf{v}$; Calls $\Pi_\mathsf{invsqrt}$ to get $\langle 1/ \sqrt{\mathbf{K}}\rangle_\mathcal{B}$, copies $\langle 1/ \sqrt{\mathbf{K}}\rangle_\mathcal{B}$ $n$ times and encrypts $\langle \widehat{1 / \sqrt{\mathbf{K}}}\rangle_\mathcal{B} \in \mathbb{Z}_p^{d_{m} * d_{n}}$; Sends set $S_2 = \{$ $[\![ \langle \widehat{1 / \sqrt{\mathbf{K}}}\rangle_\mathcal{B}]\!]_\mathcal{B}$ , $\{ [\![\sum_{j=0}^{n}\widetilde{a}_{ij}^2]\!]_\mathcal{A}\ominus \mathbf{v}]\!]_\mathcal{A}\}$, to $\mathcal{A}$.

                \item[$\bm{\mathcal{A}}$]: Receives $S_2$; Decrytps $[\![\sum_{j=0}^{n}\widetilde{a}_{ij}^2 \ominus \mathbf{v}]\!]_\mathcal{A}$ to get $\langle \mathbf{K}\rangle_\mathcal{A} = \sum_{j=0}^{n}\widetilde{a}_{ij}^2 $ $\ominus \mathbf{v}$; Calls $\Pi_\mathsf{invsqrt}$ to get $\langle 1/ \mathbf{K}\rangle_\mathcal{A}$, also copies it $n$ times to get $\langle \widehat{1 / \sqrt{\mathbf{K}}}\rangle_\mathcal{A}$; Computes $[\![ \widehat{1 / \sqrt{\mathbf{K}}}]\!]_\mathcal{B}  = \langle \widehat{1 / \sqrt{\mathbf{K}}}\rangle_\mathcal{A} \boxplus [\![ \langle \widehat{1 / \sqrt{\mathbf{K}}}\rangle_\mathcal{B}]\!]_\mathcal{B}$ and $ [\![\frac{a_{i}}{\sqrt{\sum^n a_{i}^2}}]\!]_\mathcal{B} = $ $[\![ \widetilde{\mathbf{A}} ]\!]_\mathcal{B} \boxdot [\![ \widehat{1 / \sqrt{\mathbf{K}}}]\!]_\mathcal{B}$; Generates a uniformly random matrix $\mathbf{S} \in \mathbb{Z}_p^{d_{m} * d_{n}}$; Flattens $\mathbf{S}$ row-wise to get $\widetilde{\mathbf{S}}$; Computes $[\![\frac{a_{i}}{\sqrt{\sum^n a_{i}^2}}]\!]_\mathcal{B} \boxminus \widetilde{\mathbf{S}}$; Sends set $S_3 =$ $\{[\![\frac{a_{i}}{\sqrt{\sum^n a_{i}^2}} \boxminus \widetilde{\mathbf{S}}]\!]_\mathcal{B} $, $[\![ \widetilde{\mathbf{S}}]\!]_\mathcal{A}\}$ to $\mathcal{B}$.

                \item[$\bm{\mathcal{B}}$]: Receives $S_3$; Decrypts $[\![\frac{a_{i}}{\sqrt{\sum^n a_{i}^2}} \boxminus \widetilde{\mathbf{S}}]\!]_\mathcal{B}$ to get $\frac{a_{i}}{\sqrt{\sum^n a_{i}^2}} \ominus \widetilde{\mathbf{S}}$; Computes $[\![\frac{a_{i}}{\sqrt{\sum^n a_{i}^2}} ]\!]_\mathcal{A}= \frac{a_{i}}{\sqrt{\sum^n a_{i}^2}} \ominus \widetilde{\mathbf{S}}$ $\boxplus [\![ \widetilde{\mathbf{S}}]\!]_\mathcal{A}$; Computes $[\![\mathbf{Y}]\!]_\mathcal{A} =$ $\mathbf{\gamma} \odot \sqrt{n} \boxdot [\![\frac{a_{i}}{\sqrt{\sum^n a_{i}^2}} ]\!]_\mathcal{A} \boxplus \mathbf{\beta}$. 
                Generates a uniformly random matrix $\mathbf{M} \in \mathbb{Z}_p^{d_{m} * d_{n}}$. Computes $[\![\mathbf{Y}]\!]_\mathcal{A} \boxminus \mathbf{M}$ and Sends to $\mathcal{A}$.

                \item[$\bm{\mathcal{A}}$]: Receives $[\![\mathbf{Y}]\!]_\mathcal{A} \boxminus \mathbf{M}$.

            \end{itemize}  
            \textbf{Private outputs}: \\
            $\bm{\mathcal{A}}$: 
            $\langle \mathbf{Y} \rangle_\mathcal{A}$ = $\textbf{BFV.Decrypt}(sk_\mathcal{A}$, $[\![\mathbf{Y}]\!]_\mathcal{A} \boxminus \mathbf{M})$ \\
            $\bm{\mathcal{B}}$: $\langle \mathbf{Y} \rangle_\mathcal{B}$ = $\mathbf{M}$ \\
            };
    \end{tikzpicture}}
    \caption{Protocol secure layernorm $\Pi_\mathsf{ln}$}
    \label{pro:ln}
\end{figure}

%############## input the protocol file%###############

The parameters $\gamma$ and $\beta$ of LayerNorm are held by party $\mathcal{B}$. Therefore, $\mathcal{A}$ generates a uniformly random matrix $\mathbf{S} \in \mathbb{Z}_p^{d_{m}*d_{n}}$, flattens it row-wise to get $\widetilde{\mathbf{S}}$ and calculates $[\![\frac{a_{i}}{\sqrt{\sum^n a_{i}^2}}]\!]_\mathcal{B} \boxminus \widetilde{\mathbf{S}}$ and encrypts $\widetilde{\mathbf{S}}$. These SIMD-based ciphertexts are then sent to $\mathcal{B}$, which decrypts $[\![\frac{a_{i}}{\sqrt{\sum^n a_{i}^2}} \ominus \widetilde{\mathbf{S}}]\!]_\mathcal{B}$ and add $[\![\widetilde{\mathbf{S}}]\!]_\mathcal{A}$ to obtain $[\![\frac{a_{i}}{\sqrt{\sum^n a_{i}^2}}]\!]_\mathcal{A}$. Finally, $\mathcal{B}$ can compute the LayerNorm output $[\![Y]\!]_\mathcal{A} = \mathbf{\gamma} \sqrt{n} \boxdot [\![\frac{a_{i}}{\sqrt{\sum^n a_{i}^2}}]\!]_\mathcal{A} \boxplus \mathbf{\beta} $, , where $\mathbf{\gamma} \in \mathbb{Z}_p^{d_{m}}$ $\mathbf{\beta}\in \mathbb{Z}_p^{d_{m}}$.

Subsequently, through a single round of communication, $\mathcal{B}$ can distribute $\langle \mathbf{Y} \rangle_\mathcal{A}$ and $\langle \mathbf{Y} \rangle_\mathcal{B}$ to $\mathcal{A}$ and $\mathcal{B}$, respectively. The specific process of $\Pi_{ln}$ is shown in Fig. \ref{pro:ln}. A $3 \times 3$ matrix is presented as an illustrative example to demonstrate the calculation of LayerNorm. Specific details can be found in Fig. \ref{Fig: ieln} of Appendix \ref{appendix: softmax&ln}.

\subsection{Fitting of differentiable non-linear functions}
This section introduces piecewise approximation of non-linear functions, with a focus on GeLU as a specific example due to its application in subsequent protocols. Due to the incompatibility of cryptographic techniques with non-linear functions, existing methods usually approximate the GeLU function in segments to achieve secure computation, but this approximation often leads to computational errors \cite{Bumblebee2023, BOLT2024}. By analyzing different approaches, we found that these errors usually stem from the vicinity of the boundary points of the intervals. This indicates that choosing more reasonable boundary points is necessary and can greatly reduce the error. Thus, starting from the nature of the function itself, we provide a high-precision approximate function for the GeLU. 

As is well known, the second derivative of a function can determine the inflection points (or concavity) of the function within a given interval, indicating the curvature of the function in that interval. Inflection points, where the second derivative is zero, are special points that distinguish between concave and convex regions. Typically, the function's curve changes relatively smoothly at these points, meaning the function fluctuates relatively little. Therefore, we select inflection points as the endpoints of our piecewise function intervals to minimize the impact of piecewise approximation on function accuracy. Based on this, we approximate and replace the function. If no inflection points exist, we can calculate the third derivative. The third derivative represents the curvature change rate, and a smaller value indicates a smoother function. It can also be used to select segmentation points. This method applies to a wide range of differentiable non-linear functions, such as Sigmoid, Tanh, and Mish (see Appendix \ref{appen: non-linear appro} for details). 

Here, we will focus on the piecewise approximation of GeLU. We compute its second derivative, denoted as:

\begin{equation}
    F''(x) = \frac{e^{-\frac{x^2}{2}}}{\sqrt{2\pi} } \cdot ( 2-x^2)
\label{eq: lnadivsuma}
\end{equation}
By setting $F''(x)= 0$, we obtain the inflection point $x_1 = \pm \sqrt{2}$. Through the computation of the third derivative $F'''(x)$, we observe that the $F''(x)$ and $F'''(x)$ tend towards zero when $x_2 = \pm 5.075$. Consequently, we choose this point as the second endpoint. With the segmentation points determined $x_1$ and $x_2$, we can leverage the $\rm{numpy.polyfit}$ API to find the approximates $F_1$, $F_2$ and $F_3$. And to ensure the robustness of the approximate GeLU function, we add $\epsilon = 10^{-5}$. The approximate substitution of the Gelu function is as follows:

\begin{equation}
seg5GeLU(x) = 
\begin{cases}
    \epsilon, & \text{if } x  <  -5.075 \\
    F_1(x), & \text{if } -5.075 \leq x  <  -\sqrt{2} \\
    F_2(x), & \text{if } -\sqrt{2} \leq 0 < \sqrt{2} \\
    F_3(x), & \text{if } \sqrt{2} \leq 0 < 5.075 \\
    x + \epsilon, & \text{if } x \geq 5.075
\end{cases}
\label{eq:seg5GeLU}
\end{equation}
For the interval functions $F_1$, $F_2$ and $F_3$, we refer the reader to Appendix~\ref{appen: geluinterval}

\subsection{Secure GeLU $\Pi_\mathsf{gelu}$}

This subsection introduces the secure GeLU $\Pi_\mathsf{gelu}$. Before calculating $\Pi_\mathsf{gelu}$, we first attempt to handle the GeLU function, described in Eq. \eqref{eq: GeLU}.

To compute seg5GeLU, we initially calculate $F_1$, $F_2$, and $F_3$. Given the SIMD ciphertext $[\![\mathbf{X}]\!]_\mathcal{A}$, $\epsilon$, and $[\![\mathbf{X} \boxplus \epsilon]\!]_\mathcal{A}$ are held by $\mathcal{B}$, we calculate $[\![\mathbf{X}^2]\!]_\mathcal{A} = \rm BFV.Square([\![\mathbf{X}]\!]_\mathcal{A})$, $[\![\mathbf{X}^3]\!]_\mathcal{A}= \rm BFV.HMult([\![\mathbf{X}^2]\!]_\mathcal{A}, [\![\mathbf{X}]\!]_\mathcal{A})$ , $[\![\mathbf{X}^4]\!]_\mathcal{A} = \rm BFV.Square([\![\mathbf{X}^2]\!]_\mathcal{A})$, subsequently derive $[\![F_1]\!]_\mathcal{A}$, $[\![F_2]\!]_\mathcal{A}$, and $[\![F_3]\!]_\mathcal{A}$. at party $\mathcal{B}$. Then, we need to select the segment. Hence, we construct $[\![\mathbf{X}]\!]_\mathcal{A}\boxminus \mathbf{R}$, where $\mathbf{R}$ is a uniformly random sampled from $\mathbb{Z}_p$. Through a round of communication, $\mathcal{A}$ and $\mathcal{B}$ get the secret sharing $\langle \mathbf{X} \rangle_\mathcal{A} = \mathbf{X} \ominus \mathbf{R}$, and $\langle \mathbf{X} \rangle_\mathcal{B}=\mathbf{R}$, respectively.

These shares must be transformed into the $\mathbb{Z}_{2^K}$ from $\mathbb{Z}_p$. We can directly employ the protocol $\Pi_\mathsf{LT}$ in CrypTFlow2 \cite{Cryptflow22020} for segment selection. $\mathcal{A}$ and $\mathcal{B}$ then locally compute:
\begin{equation}
\begin{split}
\langle \mathbf{S}_0 \rangle^B_\mathcal{P}&= \Pi_\mathsf{LT}( \langle \mathbf{X} \rangle_\mathcal{P}, -5.075) \quad ~~  \textcolor{gray}{~~if ~\langle \mathbf{S}_0 \rangle^B_\mathcal{P} = 1, ~x < -5.075}\\ 
\langle \mathbf{S}_1 \rangle^B_\mathcal{P} &= \Pi_\mathsf{LT}( \langle \mathbf{X} \rangle_\mathcal{P}, -\sqrt{2}) \quad  \quad ~ \textcolor{gray}{~~if ~\langle \mathbf{S}_1 \rangle^B_\mathcal{P} = 1, ~x < -\sqrt{2}}\\ 
\langle \mathbf{S}_2 \rangle^B_\mathcal{P} &= \Pi_\mathsf{LT}( \sqrt{2}, \langle \mathbf{X} \rangle_\mathcal{P}) \quad\quad \quad \textcolor{gray}{~~if ~\langle \mathbf{S}_2 \rangle^B_\mathcal{P} = 1, ~ \sqrt{2} < x }\\ 
\langle \mathbf{S}_3 \rangle^B_\mathcal{P} &= \Pi_\mathsf{LT}( 5.075, \langle \mathbf{X} \rangle_\mathcal{P}) \quad\quad \textcolor{gray}{~~if ~\langle \mathbf{S}_3 \rangle^B_\mathcal{P} = 1, ~ 5.075 < x}\\ 
\end{split}
\label{eq: signfunction}
\end{equation}

%############## input the protocol file%###############
\begin{figure}[h!]
    % \centering
    \centering\scalebox{0.80}{
    \begin{tikzpicture}
        \node[line width=1pt,draw, rectangle, fill=gray!20, rounded corners = 3pt, minimum height=0.5cm, minimum width=2.4cm] at (-0.00, 10.4) 
        {Protocol $\Pi_\mathsf{gelu}$};
        \node[line width=1pt, draw, rectangle, rounded corners = 5pt, minimum height=5cm, minimum width=7cm, text width=9cm, align=left] at (3.42, 1) 
        {  \textbf{Initialization}: \\
                ~~$\mathcal{A}$ and $\mathcal{B}$ generate each private key $sk_\mathcal{A}$, $sk_\mathcal{B}$ and public key $pk_\mathcal{A}$, $pk_\mathcal{B}$.\\
            % \textbf{Preprocessing}:\\
            %     ~~$\mathcal{B}$ generates a uniformly random $\mathbf{R} \in \mathbb{Z}_p^{d_m*d_c}$.\\
            \textbf{inputs}: \\
                ~~$\bm{\mathcal{B}}$: SIMD-based ciphertext $[\![\mathbf{X}]\!]_\mathcal{A}$\\
            \textbf{Process}: 
            \begin{itemize}
            \resetlinenumber
            \internallinenumbers
                \item[$\bm{\mathcal{B}}$]: Computes $[\![\mathbf{X}]\!]_\mathcal{A}$, $[\![\mathbf{X}^2]\!]_\mathcal{A}$, $[\![\mathbf{X}^3]\!]_\mathcal{A}$, $[\mathbf{X}^4]\!]_\mathcal{A}$, $[\![\mathbf{F}_1]\!]_\mathcal{A}$, $[\![\mathbf{F}_2]\!]_\mathcal{A}$, $[\![\mathbf{F}_3]\!]_\mathcal{A}$; Generates a uniformly random $\mathbf{R} \in \mathbb{Z}_p^{d_m*d_c}$; Flattens $\mathbf{R}$ row-wise to get $\widetilde{\mathbf{R}}$; Computes $[\![\mathbf{X}]\!]_\mathcal{A}\boxminus \widetilde{\mathbf{R}}$; Denotes $\langle \mathbf{X}\rangle_\mathcal{A} =\widetilde{\mathbf{R}}$;

                 \item[$\bm{\mathcal{A}}$]: Receives $[\![\mathbf{X}]\!]_\mathcal{A}\boxminus \widetilde{\mathbf{R}}$; Decrypts $[\![\mathbf{X}]\!]_\mathcal{A}\boxminus \mathbf{R}$ to obtain $\langle \mathbf{X}\rangle_\mathcal{B} = \mathbf{X}\ominus \widetilde{\mathbf{R}}$;

                \item[$\bm{\mathcal{P}}$]\footnote{Since $\mathcal{A}$ and $\mathcal{B}$ follow the same computation steps, we denote them as $\mathcal{P}\in \{\mathcal{A}, \mathcal{B}\}$. Before execution, $\langle \mathbf{X}\rangle_\mathcal{A}$ and $\langle \mathbf{X}\rangle_\mathcal{B}$ are transformed from the $\mathbb{Z}_p$ to the $\mathbb{Z}_{2^k}$.}: Participants $\mathcal{A}$ and $\mathcal{B}$ do the same as follow: \\
                $\langle \mathbf{S}_0 \rangle^B_\mathcal{P}= \Pi_\mathsf{LT}( \langle \mathbf{X} \rangle_\mathcal{P}, -5.075) ~  \textcolor{gray}{if ~\langle \mathbf{S}_0 \rangle^B_\mathcal{P} = 1, ~x < -5.075}$\\ 
                $\langle \mathbf{S}_1 \rangle^B_\mathcal{P} = \Pi_\mathsf{LT}( \langle \mathbf{X} \rangle_\mathcal{P}, -\sqrt{2})  ~ \textcolor{gray}{if ~\langle \mathbf{S}_1 \rangle^B_\mathcal{P} = 1, ~x < -\sqrt{2}}$\\
                $\langle \mathbf{S}_2 \rangle^B_\mathcal{P} = \Pi_\mathsf{LT}( \sqrt{2}, \langle \mathbf{X} \rangle_\mathcal{P}) \textcolor{gray}{if ~\langle \mathbf{S}_2 \rangle^B_\mathcal{P} = 1, ~ \sqrt{2} < x }$\\ 
                $\langle \mathbf{S}_3 \rangle^B_\mathcal{P} = \Pi_\mathsf{LT}( 5.075, \langle \mathbf{X} \rangle_\mathcal{P})  \textcolor{gray}{if ~\langle \mathbf{S}_3 \rangle^B_\mathcal{P} = 1, ~ 5.075 < x}$\\
                
                Then $\mathcal{P}$ locally sets:\\
                $\langle \mathbf{b}_0 \rangle^B_\mathcal{P}= \langle \mathbf{S}_0 \rangle^B_\mathcal{P} ~~  \textcolor{gray}{if ~  \langle \mathbf{b}_0 \rangle_\mathcal{P} = 1,\quad x \leq -5.075}$\\ 
                $\langle \mathbf{b}_1\rangle^B_\mathcal{P}= \langle \mathbf{S}_0 \rangle^B_\mathcal{P} \rm \lozenge \langle \mathbf{S}_1 \rangle^B_\mathcal{P}  \textcolor{gray}{~if ~ \langle \mathbf{b}_1 \rangle^B_\mathcal{P} = 1,-5.075 < x \leq -\sqrt{2}}$\\ 
                $\langle \mathbf{b}_2\rangle^B_\mathcal{P}= \langle \mathbf{S}_1 \rangle^B_\mathcal{P}  \rm \lozenge \langle \mathbf{S}_2 \rangle^B_\mathcal{P}  ~ \textcolor{gray}{if ~ \langle \mathbf{b}_2 \rangle^B_\mathcal{P} = 1,~ -\sqrt{2} < x \leq \sqrt{2}}$\\  
                $\langle \mathbf{b}_3\rangle^B_\mathcal{P}= \langle \mathbf{S}_2 \rangle^B_\mathcal{P}  \rm \lozenge \langle \mathbf{S}_3 \rangle^B_\mathcal{P}  ~ \textcolor{gray}{if ~~ \langle \mathbf{b}_3 \rangle^B_\mathcal{P} = 1,~~ \sqrt{2} < x \leq 5.075}$\\ 
                $\langle \mathbf{b}_4 \rangle^B_\mathcal{P}= \langle \mathbf{S}_3 \rangle^B_\mathcal{P} ~~  \textcolor{gray}{if  \langle \mathbf{b}_4 \rangle^B_\mathcal{P} = 1,\quad  5.075 \leq x}$\\ 
                $\mathcal{P}$ converts the bool secret sharing to additive sharing, gets $\mathbf{b}_i + 2^f$  $\equiv$ $\langle \mathbf{b}_i \rangle_\mathcal{A} + \langle \mathbf{b}_i \rangle_\mathcal{B}$ mod $2^k$.

                \item[$\bm{\mathcal{A}}$]: Computes $\langle \mathbf{b}_i \rangle_\mathcal{A} \ominus 2^f$; Encrypts $\langle \mathbf{b}_i \rangle_\mathcal{A} \ominus 2^f$ to get $ [\![\langle \mathbf{b}_i \rangle_\mathcal{A} \ominus 2^f ]\!]_\mathcal{A}$; Sends $ [\![\langle \mathbf{b}_i \rangle_\mathcal{A} \ominus 2^f ]\!]_\mathcal{A}$ to $\mathcal{B}$.
                
                \item[$\bm{\mathcal{B}}$]: Receives $ [\![\langle \mathbf{b}_i \rangle_\mathcal{A} \ominus 2^f ]\!]_\mathcal{A}$; Computes $ [\![ \mathbf{b}_i ]\!]_\mathcal{A}= \langle \mathbf{b}_i \rangle_\mathcal{B}$ $ \boxplus [\![\langle \mathbf{b}_i \rangle_\mathcal{A} \ominus 2^f ]\!]_\mathcal{A}$ to get $[\![\mathbf{b}]\!]_\mathcal{A} =[\![$ $\mathbf{b}_0$, $\mathbf{b}_1$, $\mathbf{b}_2$, $\mathbf{b}_3$, $\mathbf{b}_4]\!]_\mathcal{A}$; Computes $[\![\mathbf{Y}]\!]_\mathcal{A}$ = $[\![\mathbf{b}_0]\!]_\mathcal{A} \boxdot \epsilon \boxplus $ $[\![\mathbf{b}_1]\!]_\mathcal{A} \boxdot$ $[\![\mathbf{F}_1]\!]_\mathcal{A}$ $\boxplus$ $[\![\mathbf{b}_2]\!]_\mathcal{A} \boxdot$ $[\![\mathbf{F}_2]\!]_\mathcal{A}$ 
                $\boxplus$ $[\![\mathbf{b}_3]\!]_\mathcal{A} \boxdot$ $[\![\mathbf{F}_3]\!]_\mathcal{A}$ 
                $\boxplus$ $[\![\mathbf{b}_4]\!]_\mathcal{A} \boxdot$ $[\![\mathbf{F}_4]\!]_\mathcal{A}$; Generates uniformly random $\mathbf{N} \in \mathbb{Z}_p^{d_m*d_c}$; Computes $[\![\langle \mathbf{Y}\rangle_\mathcal{A}]\!]_\mathcal{A}= $ $[\![\mathbf{Y}]\!]_\mathcal{A} \boxminus \mathbf{N}$; Sends $[\![\langle \mathbf{Y}\rangle_\mathcal{A}]\!]_\mathcal{A}$ to $\mathcal{A}$.

                \item[$\bm{\mathcal{A}}$]: Receives $[\![\mathbf{Y}]\!]_\mathcal{A} \boxminus \mathbf{N}$.

            \end{itemize}  
            \textbf{Private outputs}: \\
            $\bm{\mathcal{A}}$: 
            $\langle \mathbf{Y} \rangle_\mathcal{A}$ = $\textbf{BFV.Decrypt}(sk_\mathcal{A}$, $[\![\mathbf{Y}]\!]_\mathcal{A} \boxminus \mathbf{N})$ \\
            $\bm{\mathcal{B}}$: $\langle \mathbf{Y} \rangle_\mathcal{B}$ = $\mathbf{N}$ \\
            };
    \end{tikzpicture}}
    \caption{Secure GeLU $\Pi_\mathsf{gelu}$}
    \label{pro: gelu}
\end{figure}
%############## input the protocol file%###############

Now, we need to compute the sign encoding of the input. We use the XOR ($\lozenge$) operation to obtain sign encoding: $b_0$, $b_1$, $b_2$, $b_3$, $b_4$, s.t., 

\[
\begin{aligned}
& b_i = 1
& \text{iff x belongs to the i-th segment.}
\end{aligned}
\]
$\mathcal{A}$ and $\mathcal{B}$  locally sets : 
\begin{equation}
\begin{split}
\langle \mathbf{b}_0 \rangle^B_\mathcal{P}&= \langle \mathbf{S}_0 \rangle^B_\mathcal{P} \quad ~~  \textcolor{gray}{if  \langle \mathbf{b}_0 \rangle_\mathcal{P} = 1,\quad x \leq -5.075}\\ 
\langle \mathbf{b}_1\rangle^B_\mathcal{P}&= \langle \mathbf{S}_0 \rangle^B_\mathcal{P} ~\rm \lozenge ~\langle \mathbf{S}_1 \rangle^B_\mathcal{P} ~~ \textcolor{gray}{if ~ \langle \mathbf{b}_1 \rangle^B_\mathcal{P} = 1,-5.075 < x \leq -\sqrt{2}}\\ 
\langle \mathbf{b}_2\rangle^B_\mathcal{P}&= \langle \mathbf{S}_1 \rangle^B_\mathcal{P} ~ \rm \lozenge ~\langle \mathbf{S}_2 \rangle^B_\mathcal{P}  \quad \textcolor{gray}{if ~~ \langle \mathbf{b}_2 \rangle^B_\mathcal{P} = 1,~~ -\sqrt{2} < x \leq \sqrt{2}}\\  
\langle \mathbf{b}_3\rangle^B_\mathcal{P}&= \langle \mathbf{S}_2 \rangle^B_\mathcal{P} ~ \rm \lozenge ~\langle \mathbf{S}_3 \rangle^B_\mathcal{P}  \quad \textcolor{gray}{if ~~ \langle \mathbf{b}_3 \rangle^B_\mathcal{P} = 1,~~ \sqrt{2} < x \leq 5.075}\\ 
\langle \mathbf{b}_4 \rangle^B_\mathcal{P}&= \langle \mathbf{S}_3 \rangle^B_\mathcal{P} \quad ~~  \textcolor{gray}{if  \langle \mathbf{b}_4 \rangle^B_\mathcal{P} = 1,\quad  5.075 \leq x}\\ 
\end{split}
\label{eq: segsel}
\end{equation}

Subsequently, the bool secret sharing is converted to additive secret sharing by using $\Pi_\mathsf{B2A}$ \footnote{The protocol $\Pi_\mathsf{B2A}$ can convert the bool secret sharing to additive secret sharing, which is also from Cryptflow2 \cite{Cryptflow22020}}. Since we employ a fixed-point representation where an input number $x$ is defined as $\lfloor \overline{x} \cdot 2^s \rceil$ mod ${2^k}$, the additive secret sharing yields a result of $\mathbf{b}_i + 2^s$  $\equiv$ $\langle \mathbf{b}_i \rangle_\mathcal{A} + \langle \mathbf{b}_i \rangle_\mathcal{B}$ mod $ p$, $\mathbf{b}_i \in \{0,1\}$. encrypting $\langle \mathbf{b}_i \rangle_\mathcal{A}$, party $\mathcal{A}$ sends $[\![\langle \mathbf{b}_i \rangle_\mathcal{A}]\!]_\mathcal{A}$ to $\mathcal{B}$. $\mathcal{B}$ computes: 
\begin{equation}
    [\![\mathbf{b}_i]\!]_\mathcal{A} = \langle \mathbf{b}_i \rangle_\mathcal{B} \boxplus [\![\langle \mathbf{b}_i \rangle_\mathcal{A}]\!]_\mathcal{A} \boxminus 2^s
\end{equation}
Finally, the GeLU function can be represented as:
\begin{equation}
\begin{split}
GeLU(x) = b_0 \epsilon + b_1F_1(x) + b_2F_2(x)  + b_3F_3(x)
&+ b_4  (x+\epsilon)  
\label{eq:approxGeLU}
\end{split}
\end{equation}
In conclusion, $\mathcal{B}$ obtains the ciphertext result of GeLU, and the specific procedure is shown in Fig \ref{pro: gelu}.

\section{EVALUATION}
\subsection{Experimental Setup}

\textbf{Implementation}: FASTLMPI is implemented in C++ and we use EMP toolkit\footnote{\url{https://github.com/emp-toolkit/emp-tool}} \cite{emptool2016} to implement communication between the parties. Leveraging the SCI\footnote{ \url{https://github.com/mpc-msri/EzPC/tree/master/SCI}} library in EzPC \cite{EzPC2019} for fixed-point numerical computations and SEAL\footnote{\url{https://github.com/microsoft/SEAL}} \cite{SEAL2023} library for BFV homomorphic encryption. FASTLMPI and extern libraries are compiled by GCC (version 11.4.0) on Ubuntu 22.04.\\
\textbf{Testbed Environment}: All the following experiments were performed on a rack server with the Intel(R) Xeon(R) Gold 6238R CPU @ 2.20GHz and 377 gigabytes of RAM. We control the communication bandwidth using the Linux Traffic Control (tc) command. We simulate two different WAN environments, which are $\rm WAN_1 = $\{400Mbps, 10 ms\} and $\rm WAN_2 = $\{200Mbps, 40 ms\}.\\
\textbf{Concrete SEAL’s and SS's Parameters}: This paper uses bit-length $k=37$ and precision $s=12$ for the expressed of fixed-point numbers, and using the SEAL parameters $\rm HE_{SCI}$$ = \{N=8192, q \approx 180, p \approx 29 \}$ (Similar to \cite{BOLT2024}, \cite{Cheetah2022}) for HE.\\
\textbf{Model Architectures}: We measure the performance of FASTLMPI on $\rm BERT_{BASE}$ \cite{Bertbase2018} and denote the number of layers (i.e., Transformer blocks) as $L$, the hidden size as $H$, and the number of the self-attention head as $A$. The size of feed-forward is $4H$, i.e., 3072 for the $H=768$ and $L=12$. \\
\textbf{Metrics}: We test on four datasets from the GLUE \cite{GLUE2018} benchmark, which is widely used to evaluate the model performance. We use four datasets in GLUE, consisting of three classification tasks: MRPC, RTE, and SST-2.

\subsection{Microbenchmarks}

\subsubsection{Secure matrix multiplication}
We first compare the performance of secure matrix multiplication with $\rm SCI_\mathsf{OT}$ \cite{Cryptflow22020} and $\rm FIT_\mathsf{FC}$ \cite{FIT22024} under three network environments: LAN, $\rm WAN_1$, and $\rm WAN_2$. The results are shown in Tab. \ref{tab: protocol_matmul}.
\begin{table}[h]
    \caption{Performance comparison of secure matrix multiplication. Timing results are
averaged from 5 runs}
    \centering
    \setlength{\tabcolsep}{20pt}
    \scalebox{0.83}{
    \begin{tabular}{@{}c|ccc@{}}
    % \begin{tabular}{@{}c|ccc@{}} 
    \toprule
    \multirow{2}{*}{\textbf{Method }} & \multicolumn{3}{c}{\textbf{Time (ms) under LAN} }  \\
    %\multicolumn{3}{c}{\multirow{2}{*}{\textbf{Optimization}}} \\
    \cline{2-4}
     & \textbf{FIT}$_\mathsf{ FC}$ \cite{FIT2024}  & \textbf{SCI}$_\mathsf{ OT}$ \cite{Cryptflow22020} & 
    \textbf{$\Pi_\mathsf{matmul}$} \\
    \midrule
    % data start----------------------
   { \footnotesize $32\times 32 \times 64$} &
    1296  & 84 & \textbf{37} \\
    { \footnotesize $ 128 \times 64 \times 64$}& 
    1370  & 288 & \textbf{41}  \\
   { \footnotesize $128 \times 256 \times 64$} & 2462 & 962 & \textbf{161}  \\
      {\footnotesize $256 \times 256 \times 64$}  & 2492  & 1729 & \textbf{210}  \\
   % data end----------------------
\midrule
   & 
   \multicolumn{3}{c}{\textbf{Time (ms) under  WAN$_1$} } \\
    \midrule
    % data start----------------------
   { \footnotesize $32\times 32 \times 64$} &
    2470  & 533 & \textbf{335} \\
     { \footnotesize $ 128 \times 64 \times 64$}& 
    2964  & 1151 & \textbf{489}  \\
   { \footnotesize $128 \times 256 \times 64$} & 4670 & 4576 & \textbf{2014}  \\
      {\footnotesize $256 \times 256 \times 64$}  & 4421  & 10931 & \textbf{4123}  \\
   % data end----------------------

\midrule
   & 
   \multicolumn{3}{c}{\textbf{Time (ms)} under  WAN$_2$} \\
    \midrule
    % data start----------------------
   { \footnotesize $32\times 32 \times 64$} &
    4796  & 2032 & \textbf{1075} \\
     { \footnotesize $ 128 \times 64 \times 64$}& 
    5236  & 4045 & \textbf{1301}  \\
   { \footnotesize $128 \times 256 \times 64$} & 7772 & 12733 & \textbf{4647}  \\
      {\footnotesize $256 \times 256 \times 64$}  & \textbf{7919}  & 17203 & 9913  \\
   % data end----------------------

\midrule
   & 
   \multicolumn{3}{c}{\textbf{Comm. (MB)} } \\
    \midrule
    % data start----------------------
     { \footnotesize $32\times 32 \times 64$}  &
        62.13  & \textbf{4.12} & 11.65   \\
      { \footnotesize $ 128 \times 64 \times 64$}& 67.78 & 29.12 & \textbf{22.95}  \\
         { \footnotesize $128 \times 256 \times 64$} & 101.67  & 123.93 & \textbf{90.72}  \\
    {\footnotesize $256 \times 256 \times 64$}  & \textbf{101.67}  & 241.95 & 181.44 \\
    % data end----------------------

    \bottomrule
    \end{tabular}
    }
    \label{tab: protocol_matmul}
\end{table}
Our method outperforms the baseline methods in terms of communication and running time in most scenarios. Particularly, in the LAN network environment with a matrix dimension of $128\times256\times64$, our method achieves a $15.3\times$ and $6.0 \times$ speedup compared to $\rm FIT_\mathsf{FC}$ and $\rm SCI_\mathsf{OT}$, respectively.

We also benchmarked the performance of $\rm linear_1$ \footnote{$\rm linear_1$ is the matrix multiplication operation in the first layer of $\rm BERT_\mathsf{base}$: $\mathbf{X}^{128*768} \otimes \mathbf{W}^{768*64}$} against SOTA methods IRON \cite{IRON2022} and BOLT \cite{BOLT2024}. The results are presented in Tab. \ref{tab: protocollinear1}
\begin{table}[h]
    \caption{Comparison of linear$_1$ with SOTA in terms of running time and communication costs. Timing results are averaged from 5 runs}
    \centering
    \setlength{\tabcolsep}{6pt}
    \scalebox{0.86}{
    \begin{tabular}{@{}cc|ccc|c@{}}
    \toprule
    
    \multicolumn{2}{c}{\multirow{2}{*}{\textbf{Method}}} & \multicolumn{3}{|c|}{\textbf{Time (s)} } & \multirow{2}{*}{\textbf{Comm.(MB)}} \\
    \cline{3-5}
   & & $\mathrm{LAN}$ & $\mathrm{WAN_1}$ & $\mathrm{WAN_2}$ & \\
    \midrule

%linear1
  \multirow{3}{*}{\textbf{Linear}$_1$}  
    & \textbf{FIT}$_\mathsf{FC}$ \cite{FIT2024}
   & 5.09  & 8.51 & 13.92 &  192.04\\
  & \textbf{SCI}$_\mathsf{ OT}$ \cite{Cryptflow22020}
   & 2.56  & 9.92 & 25.68 &  406.16\\
  { \footnotesize $ 128 \times 768 \times 64$} & \textbf{IRON \cite{IRON2022}}
   & 6.86  & 8.35 & 33.44 &  4844.14\\
   &  \textbf{BOLT \cite{BOLT2024}} & 12.70 & 13.02 & 13.78  & \textbf{7.06} \\
   & \textbf{$\Pi_\mathsf{matmul}$}  & \textbf{0.32} & \textbf{6.09}  & \textbf{12.62} & 271.47 \\

    \bottomrule
    \end{tabular}
    }
    \label{tab: protocollinear1}
\end{table}

The key takeaway is that our run time is $0.32$ seconds, and the communication cost is $271.47$ MB. In contrast to FIT$_\mathsf{FC}$ and BOLT that leverage HE, our approach offers a $15.9$-fold and $39.7$-fold improvement in runtime under LAN setting, respectively. Furthermore, in comparison with SS-based \textbf{SCI}$_\mathsf{ OT}$, our communication costs are reduced by $1.5\times$. Our approach offers substantial improvements in terms of runtime and communication cost, particularly in LAN settings. Given the other three linear layers in FeedForward, we leverage the BSGS algorithm from BOLT to execute matrix multiplication directly on encrypted data.

\subsubsection{Secure Non-linear}

We conducted a comprehensive evaluation of the non-linear functions (SoftMax, LayerNorm, and GeLU\footnote{The GeLU in BOLT is the version without preprocessing}) embedded in $\rm Bert_\mathsf{base}$ under three network environments (LAN, WAN$_1$, WAN$_2$), addressing the long-standing challenge of nonlinearity in TBM private inference. Results are presented in Tab. \ref{tab: prononlinear}.

%%%%%%%%%%%%%%%%%%%%%%%%%%%%intput file%%%%%%%%%%%%%%%%
\begin{table}[h]
    \caption{Comparison of non-linear computation with SOTA in terms of running time and communication costs. Timing results are averaged from 5 runs}
    \centering
    \setlength{\tabcolsep}{6pt}
    \scalebox{0.86}{
    \begin{tabular}{@{}cc|ccc|c@{}}
    \toprule
    
    \multicolumn{2}{c}{\multirow{2}{*}{\textbf{Method}}} & \multicolumn{3}{|c|}{\textbf{Time (s)} } & \multirow{2}{*}{\textbf{Comm.(MB)}} \\
    \cline{3-5}
   & & $\mathrm{LAN}$ & $\mathrm{WAN_1}$ & $\mathrm{WAN_2}$ & \\
    \midrule

%softmax
 \multirow{3}{*}{\textbf{SoftMax}}  & \textbf{IRON \cite{IRON2022}}
   & 5.36  & 23.25 & 129.82 &  3356.92 \\
 &  \textbf{BOLT \cite{BOLT2024}} & 1.96 & 15.92 & 76.67  & 1285.13  \\
 { \footnotesize $ 128 \times 128$}  & \textbf{FASTLMPI}  & \textbf{0.14} & \textbf{0.31}  & \textbf{0.99} & \textbf{4.94} \\
 \midrule
%layerNorm
  \multirow{3}{*}{\textbf{LayerNorm}}  & \textbf{IRON \cite{IRON2022}}
   & 10.02  & 23.25 & 57.11 &  786.38 \\
   &  \textbf{BOLT \cite{BOLT2024}} & 1.82 & 15.92 & 51.31  & 844.6  \\
    { \footnotesize $ 128 \times 768$}  & \textbf{FASTLMPI}  & \textbf{0.92} & \textbf{4.25}  & \textbf{23.23} & \textbf{154.69} \\
 \midrule
%GeLU
  \multirow{3}{*}{\textbf{GeLU}}  & \textbf{IRON \cite{IRON2022}}
   & 10.30  & 106.31 & 246.39 &  7524 \\
   &  \textbf{BOLT \cite{BOLT2024}} & 4.33 & 36.18 & 101.84  & 2769.89  \\
 { \footnotesize $ 128 \times 3072$}   & \textbf{FASTLMPI}  & \textbf{1.99} &  \textbf{14.50}  & \textbf{45.39} & \textbf{928.77} \\

    \bottomrule
    \end{tabular}
    }
    \label{tab: prononlinear}
\end{table}

%%%%%%%%%%%%%%%%%%%%%%%%%%%%intput file%%%%%%%%%%%%%%%%

Our findings reveal that the non-linear protocol within FASTLMPI offers substantial performance gains across various network environments. Specifically, $\Pi_\mathsf{softmax}$ surpasses IRON by achieving a $38.3 \times$, $75 \times$, and $131.13\times$ speedup in runtime under LAN, WAN$_1$, WAN$_2$, respectively, while reducing communication costs by $679.5 \times$. Compared to the state-of-the-art BOLT, $\Pi_\mathsf{softmax}$ demonstrates even more impressive results, with a $14\times$, $ 51.4\times$, and $ 77.4\times$ speedup and a $260.2$-fold reduction in communication costs.

\vspace{-0.2em}
\subsubsection{Accuracy evaluation of seg5GeLU}

Fig. \ref{fig_seg5GeLUalalysis} presents the visualization results and accuracy evaluation of seg5Gelu within a specific range. Specifically, we evaluated the Mean Absolute Error (MAE) of various segmentation methods using a dataset of 10,000 points spanning the range of -6 to 6. The conclusion is that our proposed method achieves a $4.5$-fold and $1.3$-fold MAE reduction compared to Bumblebee and BOLT, respectively. More notably, by carefully selecting specific points for piecewise polynomial fitting of GeLU, we achieve accuracy comparable to BubbleBee \cite{Bumblebee2023}, PUMA \cite{Puma2023} (with 6-degree polynomials), and SecFormer \cite{Secformer2024} (with 7-degree polynomials) using only 3-degree polynomials. Consequently, our approach naturally surpasses BOLT when using 4-degree polynomials.

%%%%%%%%%%%%%%%%%%%%%%%%%%%%%%%%%%%%%%%input file %%%%%%%%%%%%%%%%%%%%%%%
\begin{figure}[!t]
\centering
\subfloat[\scriptsize]{
		\includegraphics[width=1.56in]{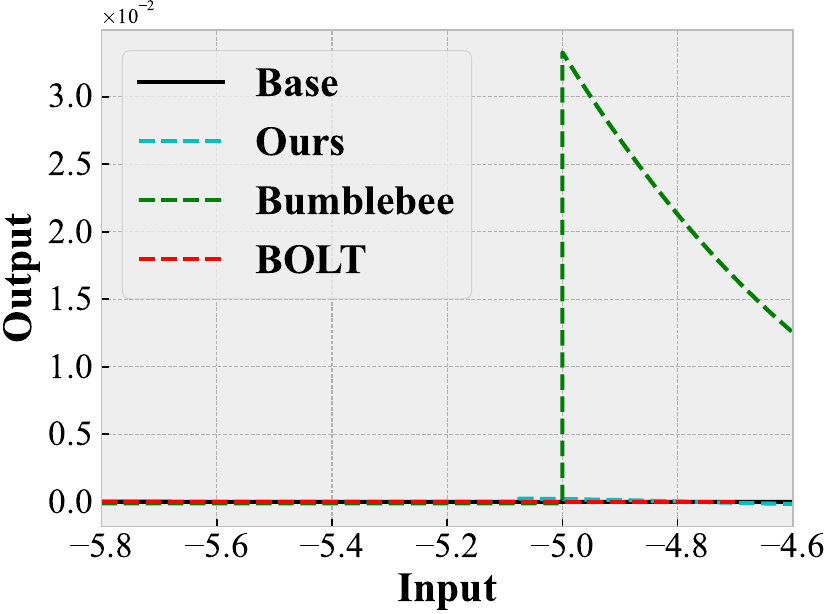}
		\label{}}
\hfil
\subfloat[\scriptsize]{
        \includegraphics[width=1.54in]{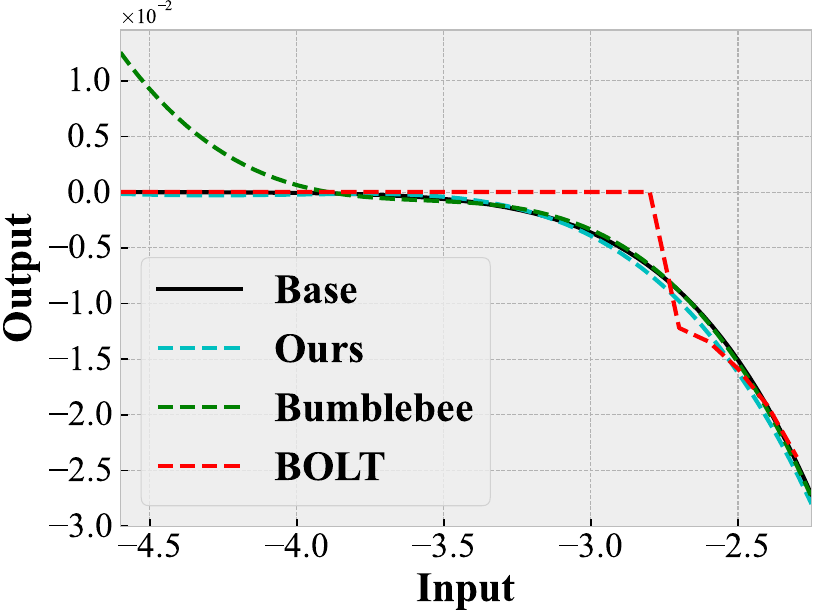}
		\label{}}
\hfil

\subfloat[\scriptsize]{
		\includegraphics[width=1.58in]{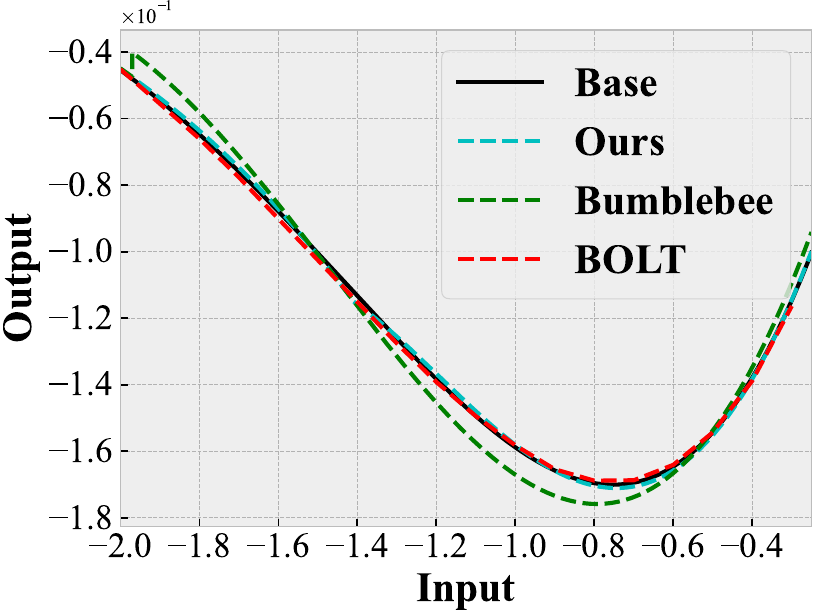}
		\label{}}
\hfil
\subfloat[\scriptsize]{
        \includegraphics[width=1.45in]{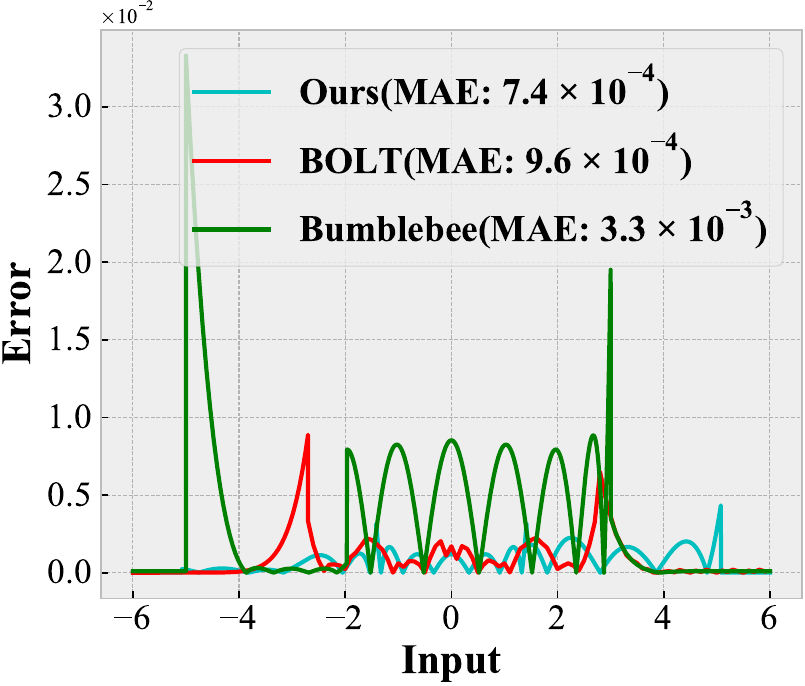}
		\label{}}
\caption{Comparative analysis of GeLU function}
\label{fig_seg5GeLUalalysis}
\end{figure}

%%%%%%%%%%%%%%%%%%%%%%%%%%%%%%%%%%%%%%%input file %%%%%%%%%%%%%%%%%%%%%%%

%%%%%%%%%%%%%%%%%%%%%%input file %%%%%%%%%%%%%%%

\begin{figure}[h!]
\centering
\subfloat[\scriptsize Plot: w/ dif. method]{
		\includegraphics[width=1.45in]{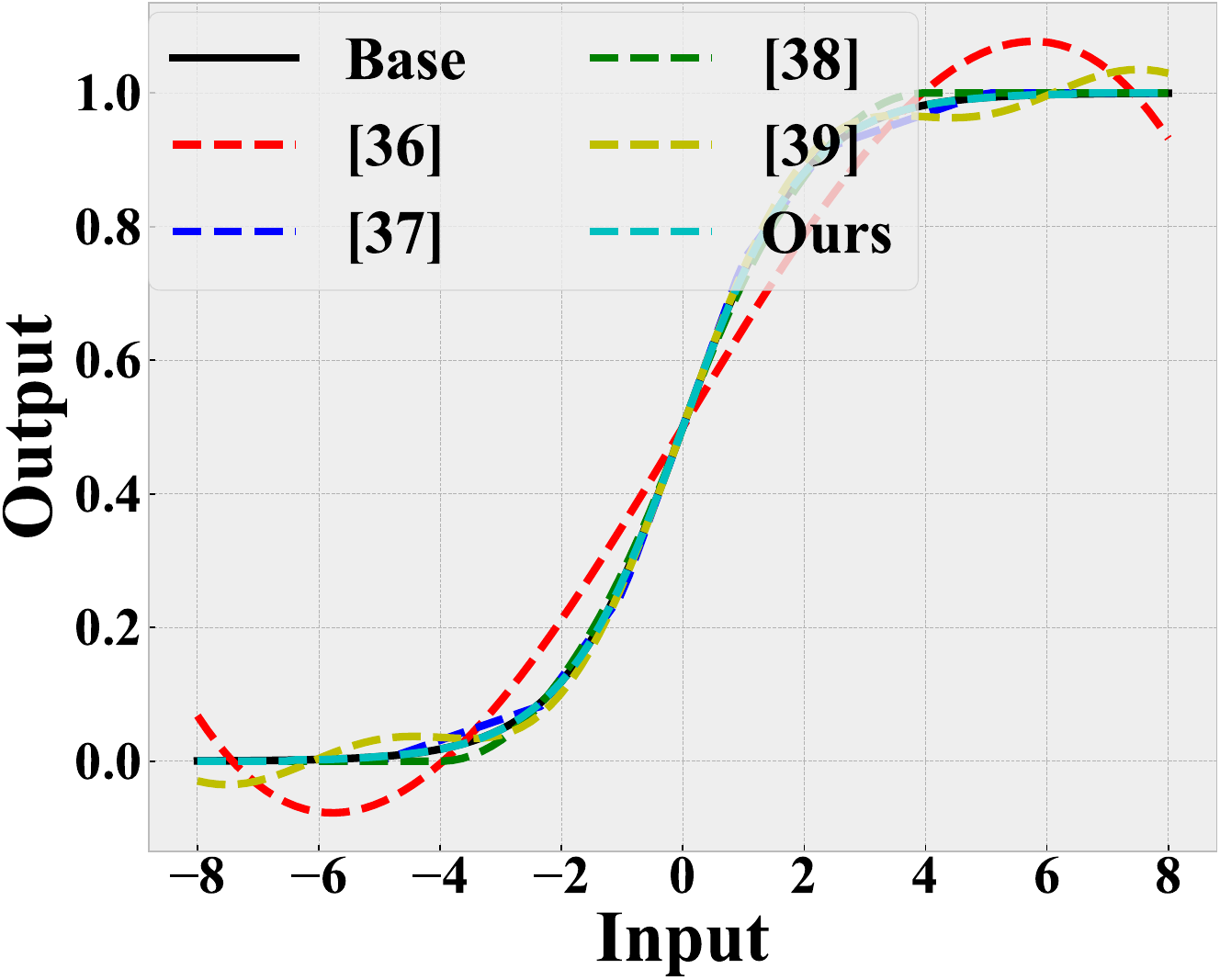}
		\label{fig: sigmoidalalysisa}}
\hfil
\subfloat[\scriptsize MAE: w/ dif. method]{
        \includegraphics[width=1.35in]{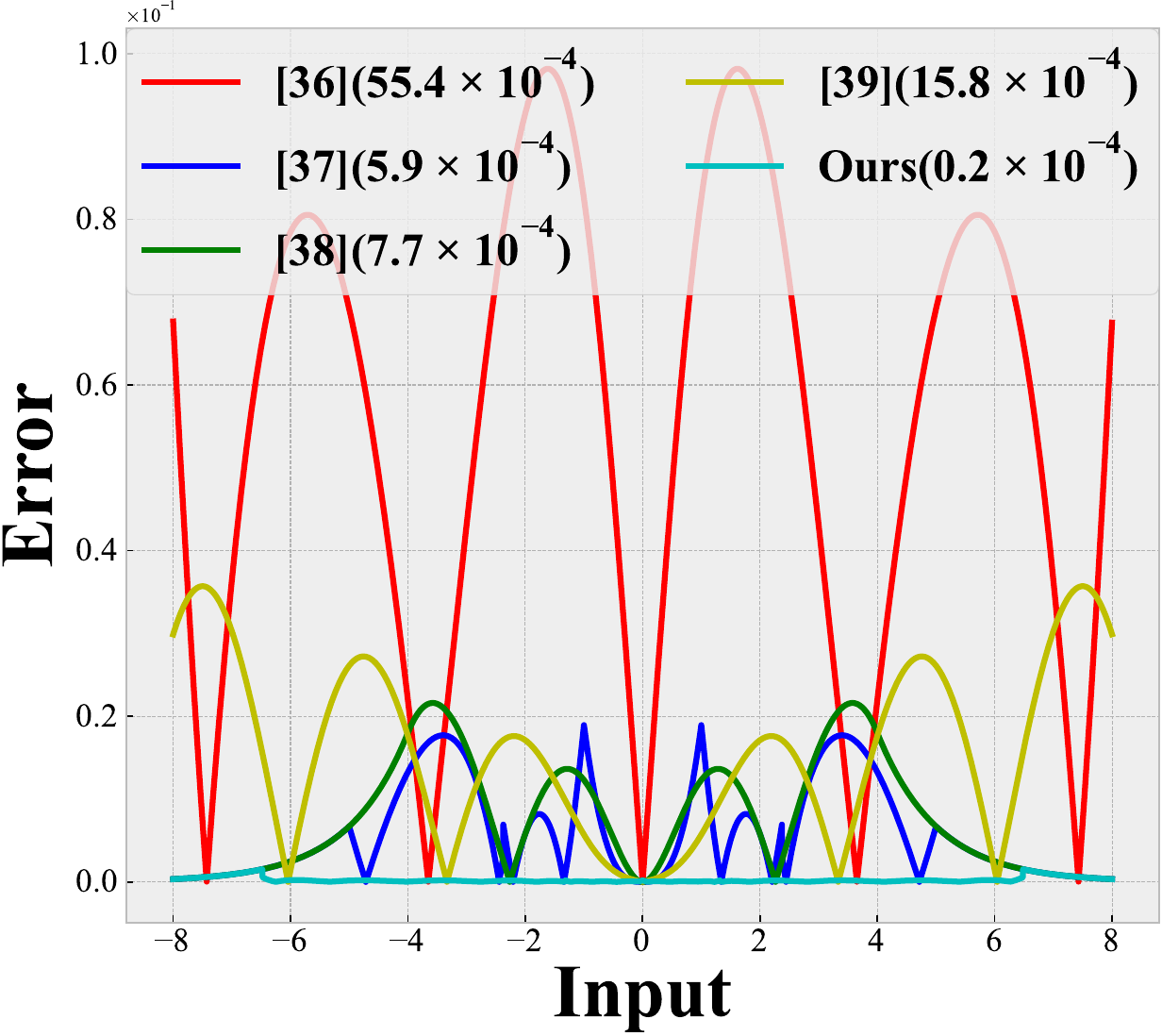}
		\label{fig: sigmoidalalysisb}}
\hfil

\subfloat[\scriptsize Plot: w/ dif. seg-point]{
		\includegraphics[width=1.45in]{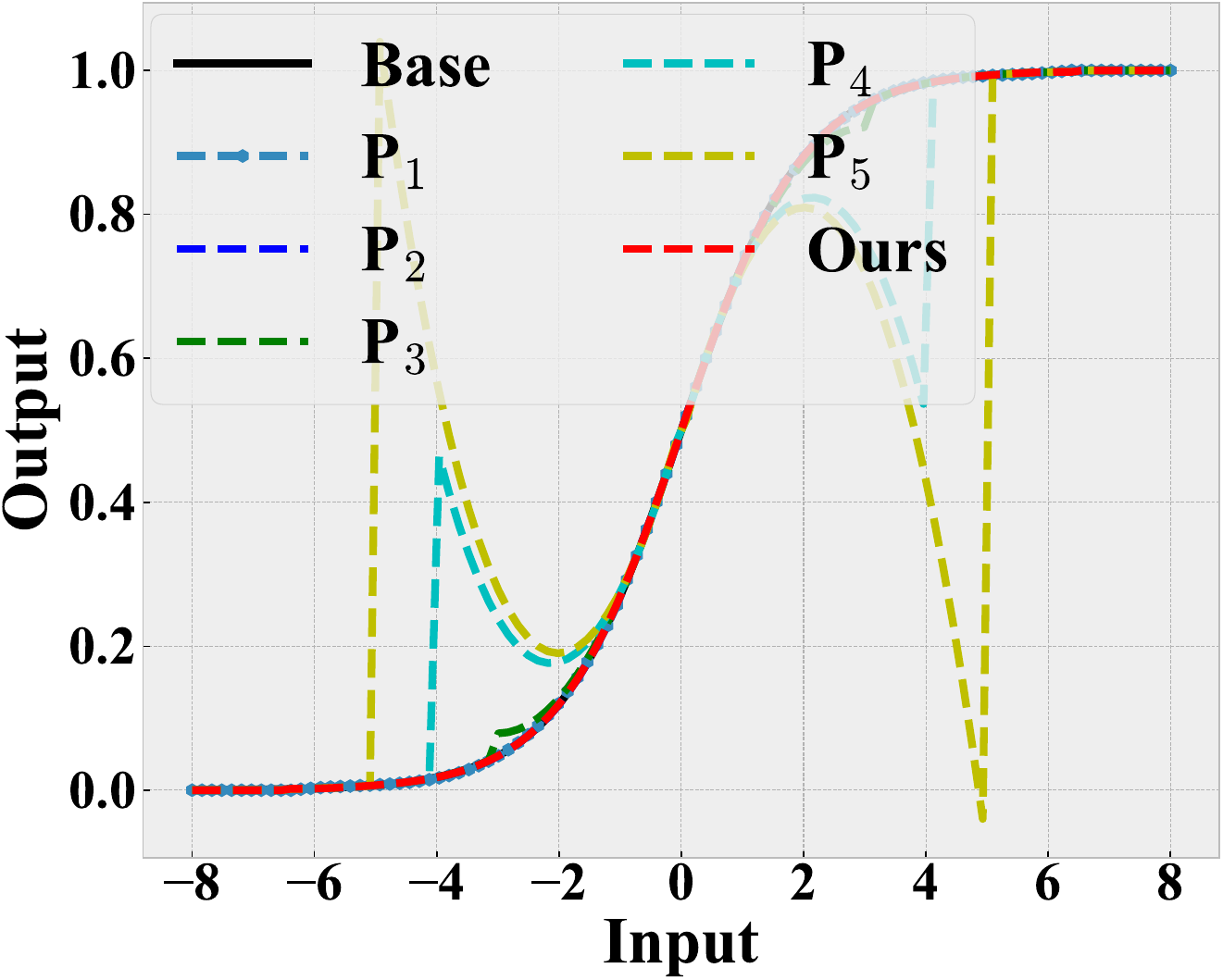}
		\label{fig: sigmoidalalysisc}}
\hfil
\subfloat[\scriptsize MAE: w/ dif. seg-point]{
        \includegraphics[width=1.35in]{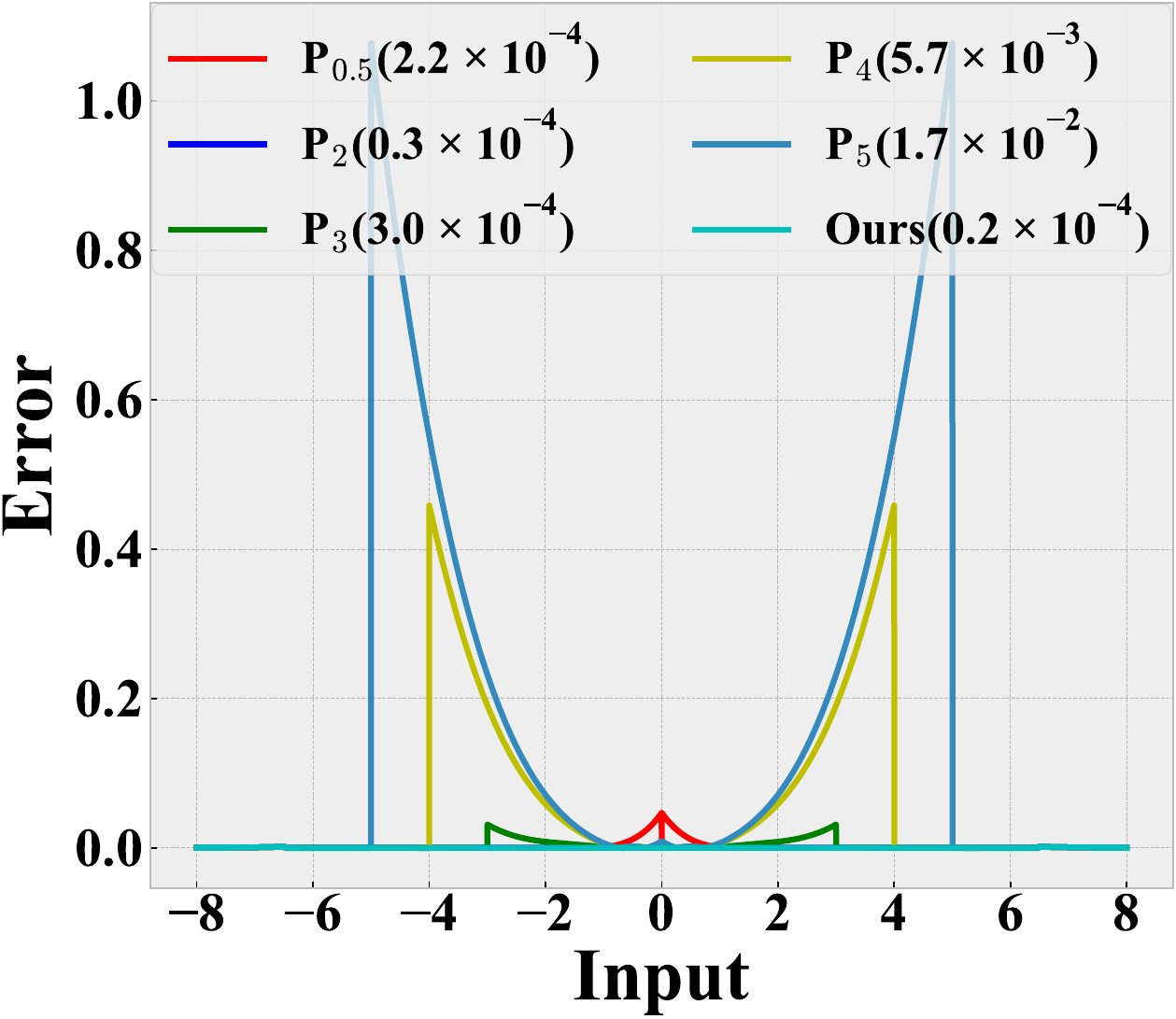}
		\label{fig: sigmoidalalysisd}}

	\caption{Comparative analysis of Sigmoid}
	\label{fig: sigmoidalalysis}
\end{figure}

%%%%%%%%%%%%%%%%%%%%%%input file %%%%%%%%%%%%%%%

To further underscore the extensibility of our method, We also fitted Sigmoid, Tanh, and Mish functions, and the results show that our method achieves a lower MAE and polynomial degree  As depicted in Fig. \ref{fig: sigmoidalalysisa} and Fig. \ref{fig: sigmoidalalysisb}, our fitting results exhibit significantly higher accuracy compared to \cite{sigmoidpoly2024, sigmoidAMI1997, sigmoidSONF2017, sigmoidFloris2023}. Furthermore, as shown in Fig. \ref{fig: sigmoidalalysisc} and Fig. \ref{fig: sigmoidalalysisd}, the superiority of our method is consistent across different partition points. These findings strongly support the effectiveness of selecting specific points for function fitting. We also evaluated the applicability of our proposed method to the Tanh and Mish activation functions. 

%%%%%%%%%%%%%%%%%%%%%%%%%%%%%%%%%inputfile$%%%%%%%%%%%

\begin{figure}[!h]
\centering
\subfloat[\scriptsize Plot: w/ dif. seg\&method]{
		\includegraphics[width=1.59in]{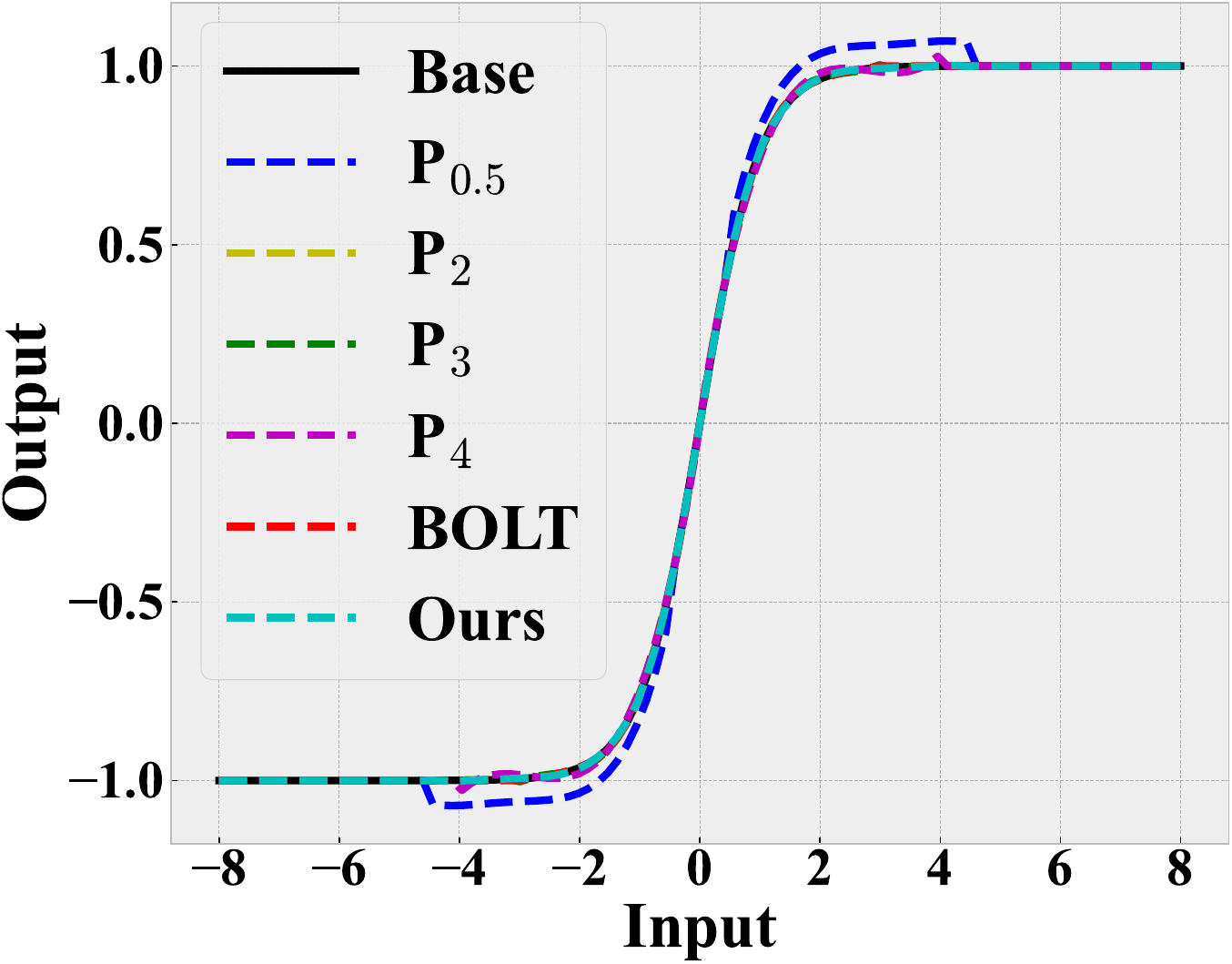}
		\label{}}
\hfil
\subfloat[\scriptsize MAE: w/ dif. seg\&method]{
        \includegraphics[width=1.40in]{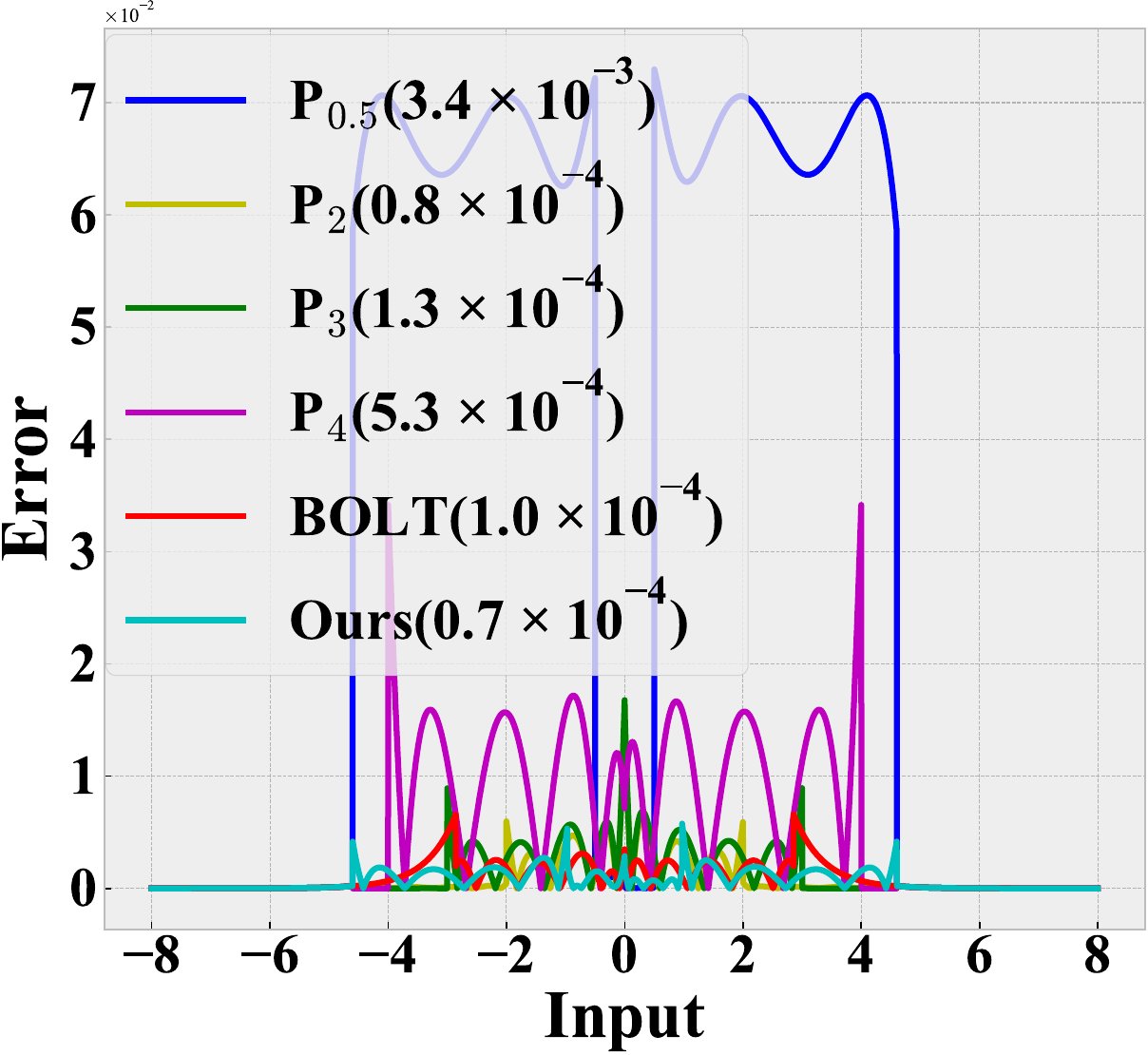}
		\label{}}
\hfil

	\caption{Comparative analysis of Tanh}
	\label{fig: tanhalalysis}
\end{figure}

%%%%%%%%%%%%%%%%%%%%%%%%%%%%%%%%%inputfile$%%%%%%%%%%%

Similar to Sigmoid's testing method, we evaluate the MAE of Tanh fitted at different points in the set $\{0.5, 
2, 3, 4\}$ and at the partition points selected by our method. As shown in Fig. \ref{fig: tanhalalysis}, the results demonstrate the effectiveness of our approach.
%%%%%%%%%%%%%%%%%%%%%%%%%%%%%%%%%inputfile$%%%%%%%%%%%

\begin{figure}[h!]
\centering
\subfloat[\scriptsize Plot: w/ dif. seg\&method]{
		\includegraphics[width=1.53in]{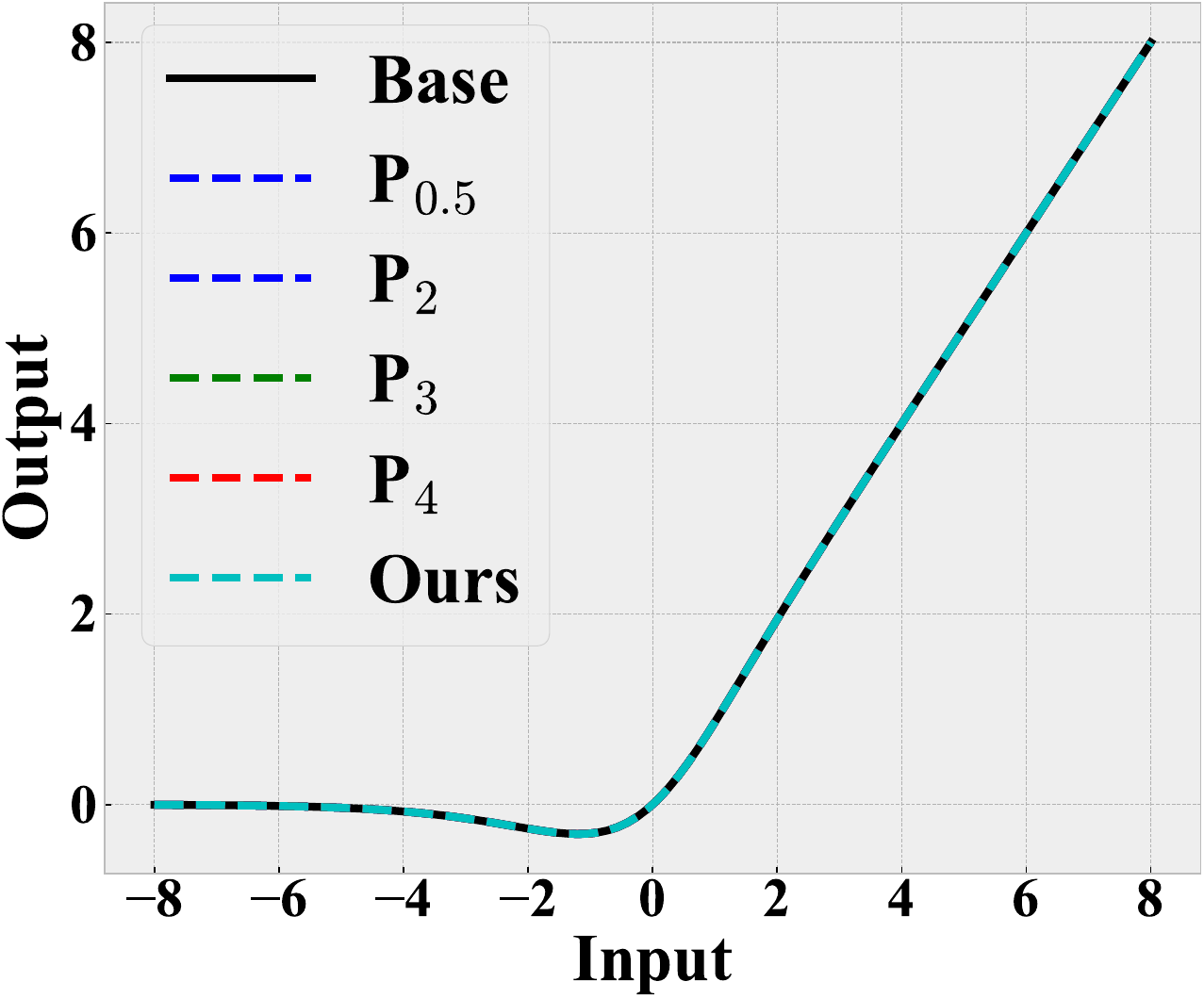}
		\label{}}
\hfil
\subfloat[\scriptsize MAE: w/ dif. seg\&method]{
        \includegraphics[width=1.47in]{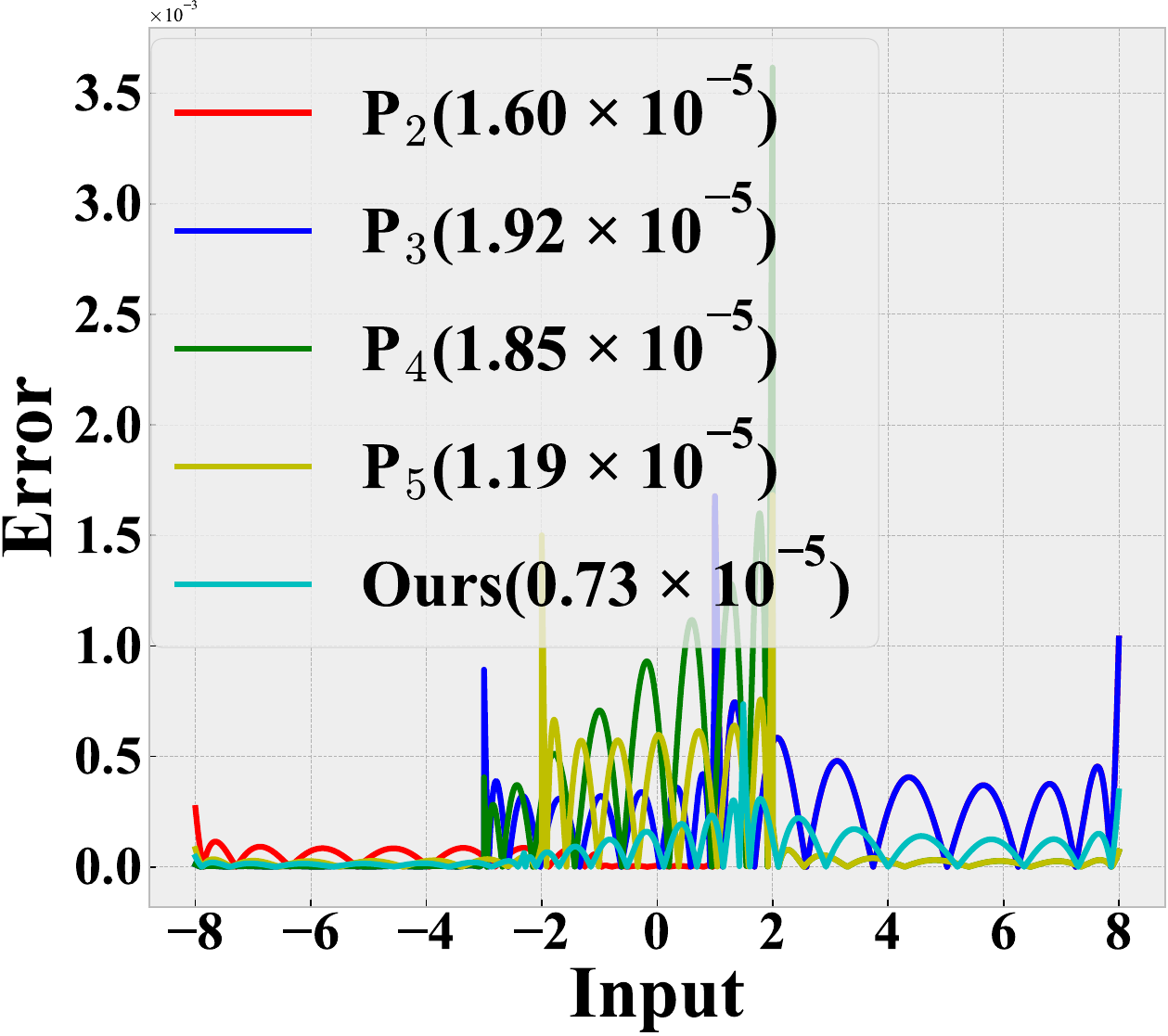}
		\label{}}
\hfil

	\caption{Comparative analysis of Mish}
	\label{fig: mishalalysis}
\end{figure}

%%%%%%%%%%%%%%%%%%%%%%%%%%%%%%%%%inputfile$%%%%%%%%%%%
As Tanh, BOLT uses a 5-degree polynomial to approximate it, we achieve higher accuracy with only a 4-degree polynomial. The evaluation of Mish in Fig. \ref{fig: mishalalysis} provides additional validation that our approach of selecting specific points for piecewise fitting can improve accuracy.
\IEEEpubidadjcol

\subsection{Evaluation on $\rm BERT_\mathsf{base}$ }

\subsubsection{Accuracy}
In Tab \ref{tab:end-to-end-acc}, we evaluated the F1 score on MRPC, RTE, and SST-2 accuracy. We followed the data setting in BOLT \footnote{Model wights are available here: \url{https://drive.google.com/drive/u/0/folders/13bBok39UevQ-6hDWHtBVtLJrYVo5VnsR},}. We only modified the GeLU activation function for a fair comparison while keeping other components unchanged. Our method does not require re-training the model as in MPCformer, nor fine-tuning as in BOLT. Baseline uses GeLU as described in Eq. \eqref{eq: GeLU}, and is otherwise identical to BOLT.
\begin{table}[h]
    \caption{Accuracy of floating-point plaintext, Baseline, BOLT, and FASTLMPI}
    \centering
    \setlength{\tabcolsep}{4pt}
    \scalebox{0.86}{
    \begin{tabular}{@{}cccccc@{}}
    \toprule

    \multirow{2}{*}{\textbf{Dataset}} & \multirow{2}{*}{\textbf{Metric}} & \multirow{2}{*}{\textbf{\cite{BertQuantized2019} (STD)}} & \multirow{2}{*}{\textbf{\shortstack{BERT \\ baseline}}} & \multirow{2}{*}{\textbf{\shortstack{BOLT \cite{BOLT2024} \\ w/o W.E.}}} & \multirow{2}{*}{\textbf{FASTLMPI}} \\
    & & & & & \\
    \midrule
    \textbf{MRPC}    & F1            & 90 (0.23) & 89.8  & 86.7 & 89.1 \\
    \textbf{RTE}     & Accuracy      & 69.7 (1.5) & 54.5 & 65.0 & 54.5 \\
  \textbf{SST-2}    & Accuracy      & 92.4 (0.59)  & 89.5  & 90.5 & 90.7\\
    %\textbf{STS-B}     & Pearson corr. & 89.6 (0.31) & 94.01 & 93.1 & 93.98 \\
    \bottomrule
    \end{tabular}
    }
    \label{tab:end-to-end-acc}
\end{table}

RET-BOLT-Acc.:=$65.0$, which may be due to the model parameters re-training by BOLT.

\subsubsection{End-to-End inference}

\begin{figure*}[!]
\centering
\subfloat[\scriptsize LAN (Sec.)]{
		\includegraphics[width=1.58in]{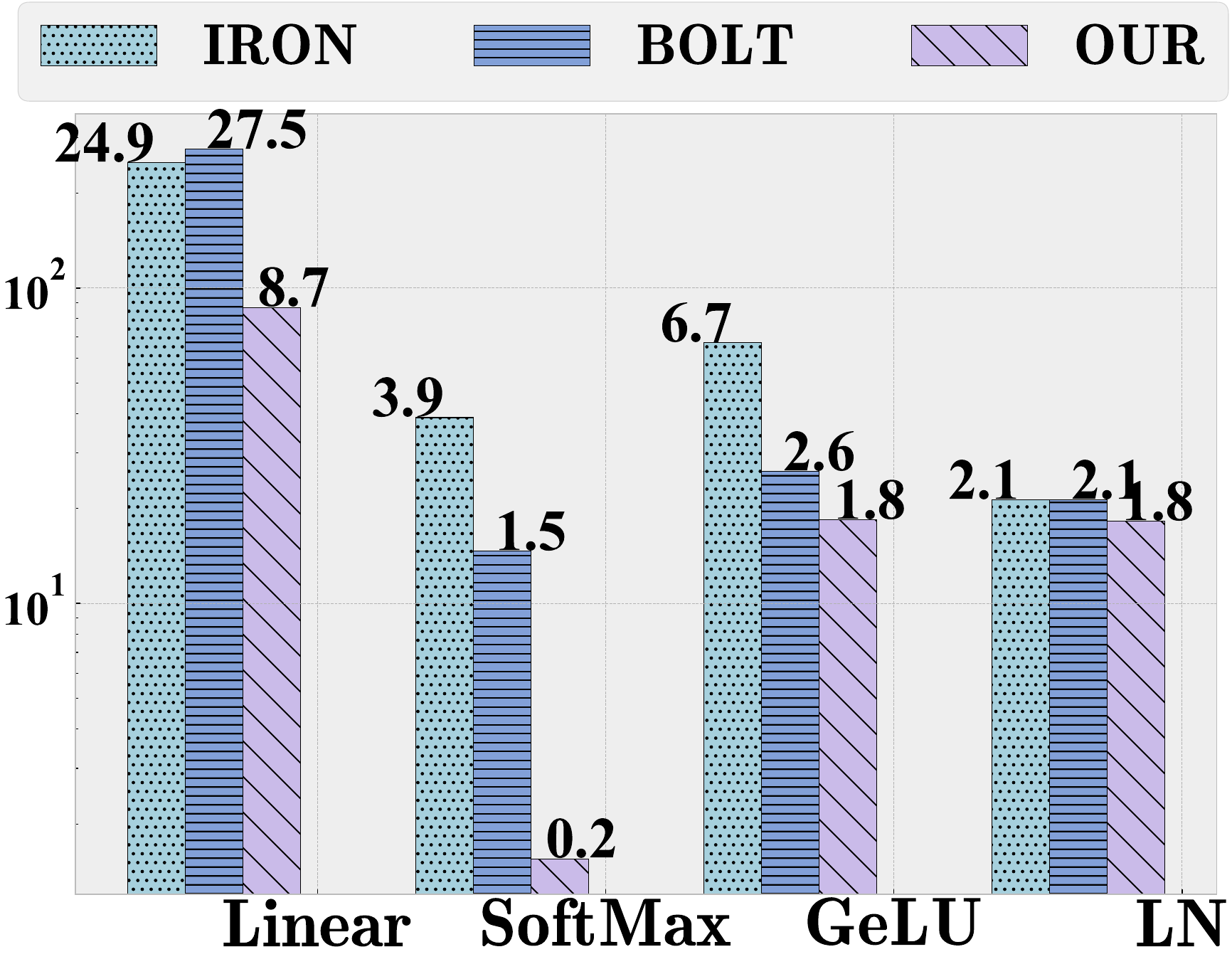}
		\label{seggelu:seg1}}
\hfil
\subfloat[\scriptsize WAN$_1$ (Sec.)]{
        \includegraphics[width=1.58in]{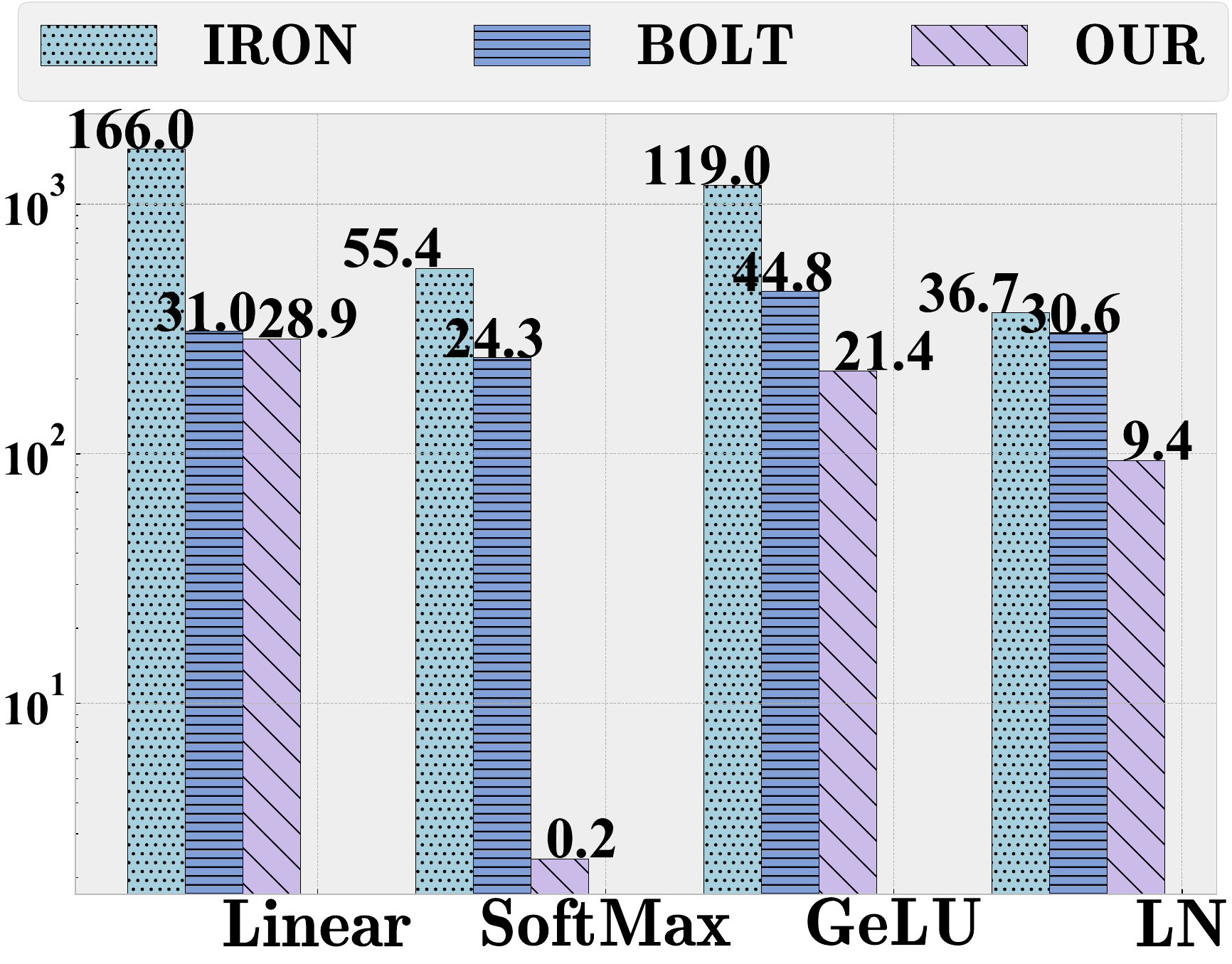}
		\label{seggelu:seg2}}
\hfil
\subfloat[\scriptsize WAN$_2$ (Sec.)]{
		\includegraphics[width=1.58in]{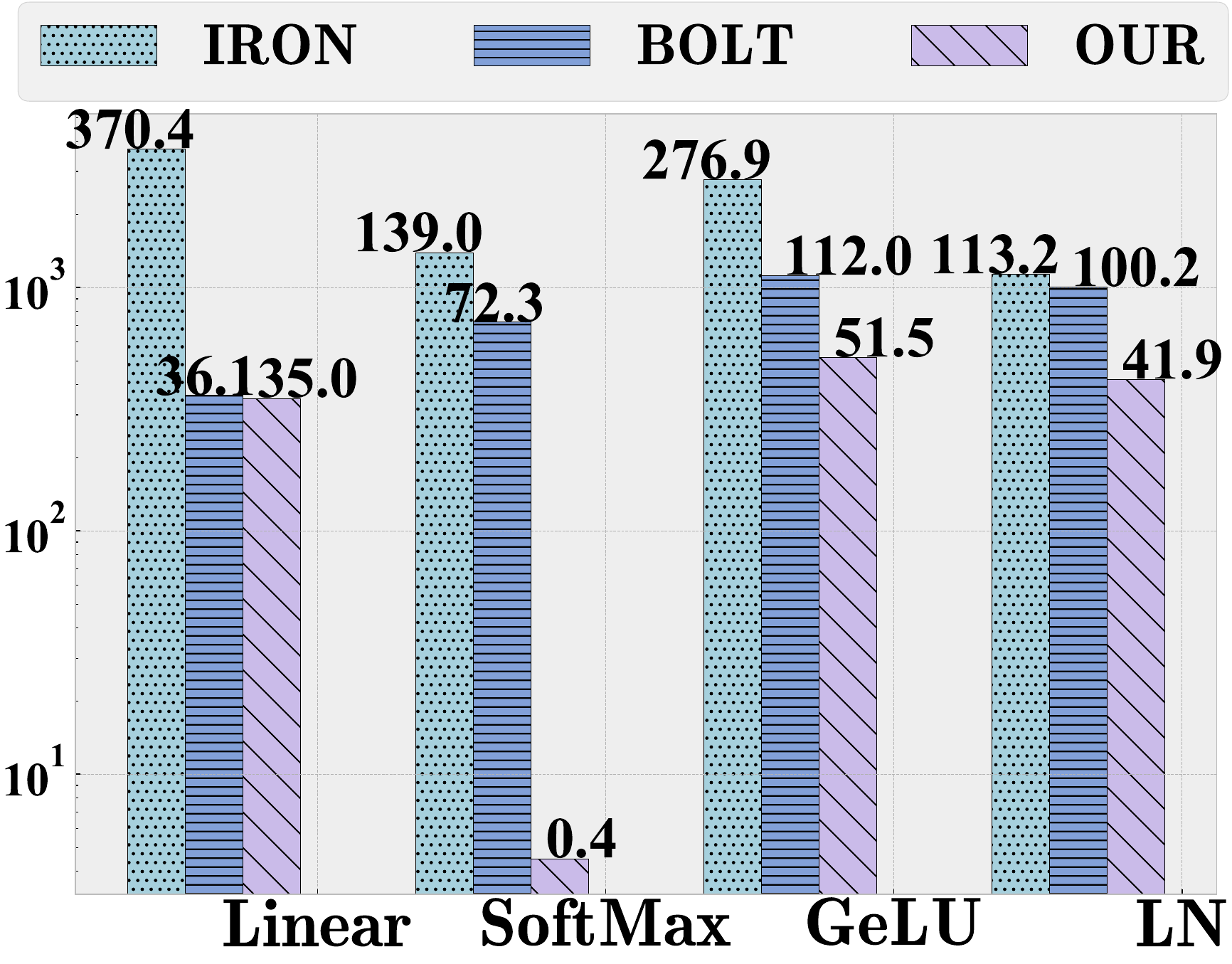}
		\label{seggelu:seg3}}
\hfil
\subfloat[\scriptsize End-to-End]{
        \includegraphics[width=1.76in]{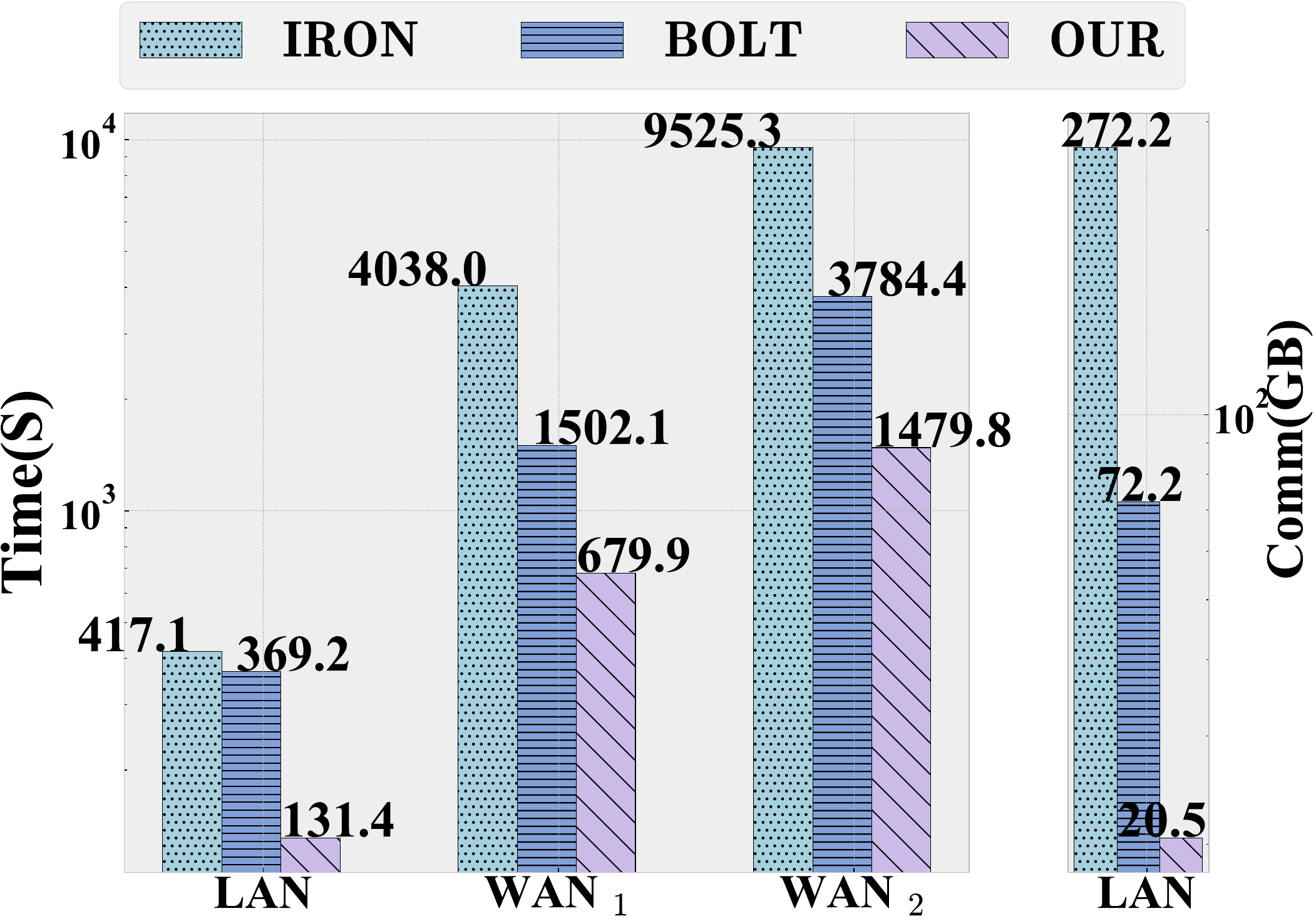}
	  \label{seggelu:seg4}}
	\caption{End-to-End linear \& non-linear operators' performance.}
	\label{fig: linear-nonlinear-costs}
\end{figure*}

In Tab. \ref{tab:end-end-runtime}, we compare the performance of FASTLMPI with other private TBM inference methods. Similar to NEXUS \cite{Nexus2024}, the end-to-end performance is the aggregation of the microbenchmarks. In summary, we have achieved a $2.8$-fold, $2.2$-fold, and $2.6$-fold improvement in inference time under three network settings and a reduction of communication costs by $72.2$\% compared to the BOLT. 
\begin{table}[h]
    \caption{End-to-end comparisons with existing private inference frameworks. The numbers of SIGMA, NEXUS, and BumbleBee are taken from their papers. Frameworks marked with "*" \textbf{are not} 2PC framework.}
    \centering
    \setlength{\tabcolsep}{10pt}
    \scalebox{0.76}{
    \begin{tabular}{@{}c|ccc|c@{}}
    \toprule
    
    \multirow{2}{*}{\textbf{Method}} & \multicolumn{3}{c|}{\textbf{Time (s)} } & \multirow{2}{*}{\textbf{Comm.(GB)}} \\
    \cline{2-4}
    & $\mathrm{LAN}$ & $\mathrm{WAN_1}$ & $\mathrm{WAN_2}$ & \\
    \midrule
    \textbf{NEXUS \cite{Nexus2024}} & $\approx$ 857  & $\approx$ 869 & \textbf{$\approx$ 881}  &  \textbf{0.16} \\
    \textbf{IRON \cite{IRON2022}}&
    417.08  & 4037.97 & 9525.28 & 272.22 \\
   \textbf{\shortstack{BOLT \cite{BOLT2024} w/o W.E.}} & 369.19 & 1502.12   & 3784.39 & 72.21  \\
    \textbf{\shortstack{SIGMA* \cite{gupta2023sigma}}} & $\approx$ 240.0  & - & -  & 34.37 \\
    \textbf{\shortstack{BumbleBee \cite{Bumblebee2023}}} & $\approx$ 204.0  & $\approx$ 971.0 &  $\approx$ 1578.0  & 8.58 \\
  \textbf{FASTLMPI} & 
   \textbf{131.36} &\textbf{679.90} & 1479.82  & 20.49  \\
    \bottomrule
    \end{tabular}
    }
    \label{tab:end-end-runtime}
\end{table}

% \textbf{\shortstack{BOLT \cite{BOLT2024} w/o W.E.}}
\textbf{It is noteworthy that} NEXUS exhibits a more significant advantage in scenarios with higher network latency. This is because NEXUS is a non-interactive approach that employs a unified HE technique. However, when network latency is low, FASTLMPI still holds a substantial advantage. We also evaluated the performance of IRON, BOLT, and FASTLMPI under LAN, WAN$_1$, and WAN$_2$ network conditions in the end-to-end inference task. Fig. \ref{fig: linear-nonlinear-costs} shows linear and non-linear operators' inference time and communication costs. $\Pi_\mathsf{softmax}$ benefits the most from the fine-grained collaborative technology of HE and SS, experiencing a $9-72\times$ speedup in runtime and a $260-769\times$ reduction in communication cost.

 \section{RELATED WORK}

With the proliferation of Deep Learning, private inference has become a key area of research. 

\textbf{Private inference with efficient protocols}: Many methods for privacy protection in large models are implemented through the design of efficient protocols, such as \cite{Bumblebee2023}, \cite{Nexus2024}, \cite{ScaleMPCfp2024}, \cite{Secformer2024}, \cite{kundu2023learning}, \cite{THEX2022}, \cite{hou2023ciphergpt}. Where \cite{Bumblebee2023} has optimized the matrix multiplication protocol via a specialized RLWE batching technique, provided methods for designing efficient and accurate protocols for the activation functions used in transformers, and has successfully run BumbleBee on five pre-trained transformer models. And NEXUS \cite{Nexus2024} designs a non-interactive private transformer inference using HE, achieving efficient inference under the WAN environment. They provide SIMD-based ciphertext slot folding and decompression techniques to reduce communication costs. Tan et al. \cite{tan2021cryptgpu} introduce CrytpGPU, which identifies a sequence of “GPU-friendly” cryptographic protocols to enable privacy-preserving evaluation of both linear and non-linear operations on the GPU.

\textbf{Private inference with hardware acceleration}: Due to the high computational resource consumption of some cryptographic primitives, some works adopt hardware acceleration methods to improve the efficiency of private computations, as in \cite{gupta2023sigma}, \cite{peng2023rrnet}, \cite{watson2022piranha}, \cite{jawalkar2023orca}, \cite{yang2023xnet}. Among them, SIGMA \cite{gupta2023sigma} enhances inference efficiency through GPU while maintaining model accuracy based on Function secret sharing (FSS) \cite{boyle2016function}. It is worth noting that SIGMA is the first work to design a 2-PC private TBM inference by using the FSS technique and accelerating it with GPU. RRNet \cite{peng2023rrnet} replaces computationally expensive ReLU operations with trainable polynomial activation functions, develops a cryptographic hardware scheduler and performance model for FPGA platforms, constructs a latency lookup table to optimize the scheduling of encrypted operations.

\textbf{Private inference with quantization and distillation}: Some studies also employ quantization and distillation methods to reduce the model size, thereby enhancing the efficiency of privacy-preserving computations. SecFormer\cite{Secformer2024} and MPCFormer\cite{MPCformer2024} utilize knowledge distillation to train a surrogate Transformer model with fewer parameters and simpler architecture, achieving performance comparable to the original Transformer. The surrogate model is both high-performing and compatible with SMPC. MPCViT\cite{zeng2023mpcvit}, based on the vision transformer, replaces the Softmax with an accurate and MPC-friendly ReLU Softmax. An MPC-aware differentiable NAS \cite{chitty2022neural} algorithm that can learn architectural parameters and model parameters simultaneously and finally uses knowledge distillation to improve the accuracy.refe\section{DISCUSSION}

\textbf{Extensibility of \SystemName}: A specific benchmark is provided on the BERT$_\mathsf{base}$ model. \SystemName, enabled by the four fundamental protocols we provide, can be applied to transformer-based language models including GPT \cite{GPT2024}, BERT$_\mathsf{large}$, and so on. Moreover, the offered activation functions (Sigmoid, Tanh, and Mish) are well-suited for traditional neural networks like Logistic regression \cite{2011logistic}, LSTM \cite{yu2019lstmreview}, and RNN \cite{rnn2020}.

\textbf{The gains from fine-grained co-design of HE and SS}: It is well known that, classical arts use HE in linear computations (Matmul), and SS in non-linear computations (Softmax, LayerNorm, GeLU). However, in some special cases, preprocessing ciphertexts in SIMD form before homomorphic operations can bring greater benefits than SS. This is our original intention: to selectively use HE in Softmax, LayerNorm, and GeLU while selectively using SS in Matmul, breaking the boundaries of operators of using HE solely in matrix multiplication and SS solely in non-linear operations, aiming to reduce computational and communication costs simultaneously. Our experiments have verified that this is feasible. 
rrin\section{CONSLUSION \& FUTURE WORK}

We introduce \SystemName, an optimized 2PC private TBM inference framework, which significantly reduces costs compared to previous methods. This advancement enhances the practicality of secure inference for the TBM. Moreover, we draw two key conclusions. Firstly, through the strategic collaboration of Homomorphic Encryption (HE) and secret sharing (SS), all homomorphic operations are transformed into ciphertext-plaintext multiplications. This allows exploration of Partial Homomorphic Encryption algorithms like Paillier \cite{paillier1999public}, \cite{gong2024efficient} and ECC-Elgamal \cite{elgamal1985public}, \cite{raju2003elliptic}, \cite{reegan2021highly} for more efficient HE\&SS-based protocols. Secondly, breaking the linear and nonlinear boundaries significantly improves performance, with a higher emphasis on HE making private TBM inference more suitable for WAN environments with poor network conditions, while a heavier reliance on SS provides an advantage in LAN environments.

In future work, creating a network environment-adaptive private TBM inference framework using feedback tuning mechanisms is a key objective, which can dynamically generate the optimal private TBM inference scheme tailored to the prevailing network conditions. The secure computing protocol used in the generation scheme can directly call the existing mature private computing library. But it is also possible to redesign the base protocol to be more efficient. So we believe that the feedback mechanism of the cost model and designing efficient protocols do not conflict and are equally important for driving the development of Privacy-preserving Machine Learning.

% Reference
\bibliographystyle{IEEEtran}
\bibliography{tifs_reference}

\newpage

%\centering
%\begin{ta\appendix
\section{Security against the semi-honest adversary}\label{appendix: proof}

\subsection{2PC functionality}
\begin{figure}[h]
    \centering\scalebox{0.80}{
    \begin{tikzpicture}
        \node[line width=1pt,draw, rectangle, fill=gray!20, rounded corners = 3pt, minimum height=0.4cm, minimum width=2.4cm] at (-0.53, 4.0) 
        {Functionality $\mathcal{F}^\mathsf{2pc}[f]$};
        \node[line width=1pt, draw, rectangle, rounded corners = 5pt, minimum height=3cm, minimum width=4cm, text width=9.8cm, align=left] at (2.8, 1) 
        {
            \textbf{Initialization}: \\
            ~~$\mathcal{F}^\mathsf{2pc}$ interacts with $\mathcal{A}$, $\mathcal{B}$ and the adversary $\mathsf{Adv}$. Let $f$ denote the functionality to be computed.\\
    
            \textbf{Inputs}:\\
            \begin{itemize}
            \resetlinenumber
            \internallinenumbers
             \item ~~Upon receiving $(\mathsf{Input}, \mathsf{sid}, \mathbf{A})$ from $\mathcal{A}$, send $(\mathsf{Input}, \mathsf{sid}, \mathcal{A})$ to $\mathsf{Sim}$ 
             
             \item ~~Upon receiving $(\mathsf{Input}, \mathsf{sid}, \mathbf{B})$ from $\mathcal{B}$, send $(\mathsf{Input}, \mathsf{sid}, \mathcal{B})$ to $\mathsf{Sim}$ 
            \end{itemize}
            \textbf{Execution}: 
            \begin{itemize}

            \internallinenumbers
            
            \item  Computes $\mathbf{Z} = f(\mathbf{A}, \mathbf{B})$, picks random value $\langle \mathbf{Z} \rangle_\mathcal{B}$, calculates $\langle \mathbf{Z} \rangle_\mathcal{A} = \mathbf{Z} - \langle \mathbf{Z} \rangle_\mathcal{B}$.

            \item  Sends $(\mathsf{Output},\mathsf{sid},\langle \mathbf{Z} \rangle_\mathcal{A})$ to $\mathcal{A}$, $(\mathsf{Output},\mathsf{sid},\langle \mathbf{Z} \rangle_\mathcal{B})$ to $\mathcal{B}$
            \end{itemize}  
        };

    \end{tikzpicture}
 }
    \caption{Functionality for 2PC.}
    \label{pro: functionality}
\end{figure}  
   
\subsection{Security of Matmul}
We define the functionality of Matmul (denoted as $\mathcal{F}^\mathsf{2pc}_\mathsf{matmul}$) as an instance of $\mathcal{F}^\mathsf{2pc}$, where $\mathcal{F}^\mathsf{2pc}_\mathsf{matmul}$ calculates $\mathbf{Z} = f(\mathbf{A}, \mathbf{B}) := \mathbf{A} \otimes \mathbf{B}$. If $\mathcal{B}$ is the corrupted party, instead of picking random $\langle \mathbf{Z} \rangle_{\mathcal{B}}$, $\mathcal{F}^\mathsf{2pc}_\mathsf{matmul}$ receives $(\mathsf{Modify}, \mathsf{sid}, \langle \mathbf{Z} \rangle_{\mathcal{B}})$ from $\mathsf{Sim}$, calculates $\langle \mathbf{Z} \rangle_{\mathcal{A}} = \mathbf{Z} - \langle \mathbf{Z} \rangle_{\mathcal{B}}$.

\begin{theorem}
    The protocol $\Pi_\mathsf{matmul}$ is securely computes functionality $\mathcal{F}^\mathsf{2pc}_\mathsf{matmul}$ in the presence of static semi-honest adversaries $\mathsf{Adv}$ controlling each of parity $\mathcal{A}$ or $\mathcal{B}$.
\end{theorem}
\begin{proof}\renewcommand{\qedsymbol}{}
Consider the first case that $\mathcal{A}$ is corrupted. In the protocol, $\mathcal{A}$ receives a single message $[\![ \langle C\rangle_\mathcal{A} ]\!]_\mathcal{A}$. We construct the simulator $\mathsf{Sim}$ as follows:
\begin{itemize}
\item $\mathsf{Sim}$  sends $(\mathsf{Input}, \mathsf{sid}, \mathbf{A})$ to $\mathcal{F}^\mathsf{2pc}_\mathsf{matmul}$. 
\item $\mathsf{Sim}$ receive $[\![ \langle \Tilde{\mathbf{A}_\mathsf{i}}\rangle_\mathcal{A} ]\!]_\mathcal{A}$ from the corrupted party $\mathcal{A}$.
\item Upon receiving $\langle \mathbf{Z} \rangle_\mathcal{A}$ from $\mathcal{F}^\mathsf{2pc}_\mathsf{matmul}$, $\mathsf{Sim}$ encrypt it as $\mathbf{c} = [\![ \langle \mathbf{Z} \rangle_\mathcal{A} ]\!]_\mathcal{A}$ and forward $\mathbf{c}$ to corrupted $\mathcal{A}$ 
\end{itemize}

$\mathsf{Indistinguishable.}$ Obviously, when the corrupted $\mathcal{A}$ receives $\mathbf{c} = [\![ \langle \mathbf{Z} \rangle_\mathcal{A} ]\!]_\mathcal{A}$, it will decrypt it and set $\langle \mathbf{Z} \rangle_\mathcal{A}$ as result, leading a correct output. For the incoming message $[\![ \langle \mathbf{C} \rangle_\mathcal{A} ]\!]_\mathcal{A}$ of the corrupted $\mathcal{A}$, due to that $\langle \mathbf{C} \rangle_\mathcal{A}$ is encrypted with one-time pad, with the random key $\mathbf{R}$,  $\mathcal{A}$  cannot distinguish it from the random value $\langle \mathbf{Z} \rangle_\mathcal{A}$ of the ideal world.

Consider the second case that $\mathcal{B}$ is corrupted. In the protocol, $\mathcal{B}$ receives a single message $[\![ \Tilde{A_i} ]\!]_\mathcal{A}$. We construct the simulator $\mathsf{Sim}$ as follows:
\begin{itemize}
\item $\mathsf{Sim}$ sends $(\mathsf{Input}, \mathsf{sid}, \mathbf{B})$ to functionality $\mathcal{F}^\mathsf{2pc}_\mathsf{matmul}$.

\item $\mathsf{Sim}$ picks random value $\mathbf{A}'$ and sends $[\![ \Tilde{\mathbf{A}'} ]\!]_\mathcal{A}$ to corrupted $\mathcal{B}$;

\item Upon receiving $[\![\langle \mathbf{C}\rangle_\mathcal{A}]\!]_\mathcal{A}$ from corrupted $\mathcal{B}$, $\mathsf{Sim}$ extracts $\mathbf{R}$ as $\mathbf{R} = \langle \mathbf{C}\rangle_\mathcal{A} - \mathbf{A}'\mathbf{B}$;
\item $\mathsf{Sim}$ sends $(\mathsf{Modify}, \mathsf{sid}, \mathbf{R})$ to functionality $\mathcal{F}^\mathsf{2pc}_\mathsf{matmul}$.
\end{itemize}

$\mathsf{Indistinguishable.}$ Obviously, when the corrupted $\mathcal{B}$ sets its output as $R$, $\mathcal{F}^\mathsf{2pc}_\mathsf{matmul}$ will output $\mathbf{AB} - \mathbf{R}$ to $\mathcal{A}$, leading a correct output. For the incoming message $[\![ \Tilde{\mathbf{A}_\mathsf{i}}]\!]_\mathcal{A}$ of the corrupted $\mathcal{B}$, if the adversary $\mathsf{Adv}$ can distinguish $[\![ \Tilde{\mathbf{A}'_\mathsf{i}} ]\!]_\mathcal{A}$ and $[\![ \Tilde{\mathbf{A}_\mathsf{i}}]\!]_\mathcal{A}$, we can construct a game to break the BFV scheme.

\end{proof}

\subsection{Security of SoftMax}
We define the functionality of SoftMax (denoted as $\mathcal{F}_\mathsf{softmax}^\mathsf{2pc}$) as an instance of $\mathcal{F}^\mathsf{2pc}$, where $\mathcal{F}^\mathsf{2pc}_\mathsf{softmax}$ calculates $\mathbf{Z} = f(\langle \mathbf{X} \rangle_\mathcal{A}, \langle \mathbf{X} \rangle_\mathcal{B}) := SoftMax(\frac{e^{x_i}}{\sum_\mathsf{j=0}^\mathsf{m-1}e^{x_\mathsf{ij}}})$. if $\mathcal{A}$ is the corrupted party, instead of picking randoms $\langle \mathbf{Z} \rangle_\mathcal{A}$, $\mathcal{F}^\mathsf{2pc}_\mathsf{softmax}$ receives $(\mathsf{Modify}, \mathsf{sid}, \langle \mathbf{Z} \rangle_{\mathcal{A}})$ from $\mathsf{Sim}$, calculates $\langle \mathbf{Z} \rangle_\mathcal{B} = \mathbf{Z} - \langle \mathbf{Z} \rangle_\mathcal{A}$.

\begin{theorem}
    The protocol $\Pi_\mathsf{softmax}$ is securely computes functionality $\mathcal{F}^\mathsf{2pc}_\mathsf{softmax}$ in the presence of static semi-honest adversaries $\mathsf{Adv}$ controlling each of parity $\mathcal{A}$ or $\mathcal{B}$.
\end{theorem}
\begin{proof}\renewcommand{\qedsymbol}{}
Consider the first case that $\mathcal{A}$ is corrupted. In the protocol, $\mathcal{A}$ receives messages $[\![\widetilde{\langle\mathbf{E}\rangle_\mathcal{B}}]\!]_\mathcal{B}$, $ [\![ \sum(\mathbf{E}) \odot \mathbf{v}]\!]_\mathcal{A}$,  $[\![ \widehat{\mathbf{V}} ]\!]_\mathcal{B}$. We construct the simulator $\mathsf{Sim}$ as follows:
\begin{itemize}
\item $\mathsf{Sim}$ sends $(\mathsf{Input}, \mathsf{sid}, \langle X \rangle_\mathcal{A})$ to $\mathcal{F}^\mathsf{2pc}_\mathsf{softmax}$. 
\item $\mathsf{Sim}$ picks random values $\widetilde{\langle\mathbf{E}\rangle_\mathcal{B}}'$, $ \widehat{\mathbf{V}}'$ and sends $[\![\widetilde{\langle\mathbf{E}\rangle_\mathcal{B}}']\!]_\mathcal{B}$,  $[\![ \widehat{\mathbf{V}}' ]\!]_\mathcal{B}$ to corrupted party $\mathcal{A}$.

\item Upon receiving $[\![\widetilde{\mathbf{E}} \oplus \mathbf{R}]\!]_\mathcal{B} $, $[\![ \mathbf{SR} ]\!]_\mathcal{A}$, $[\![\langle \mathbf{Y} \rangle_\mathcal{B}]\!]_\mathcal{B}$ from corrupted $\mathcal{A}$. $\mathsf{Sim}$ can extracts $\widetilde{\mathbf{E}} \oplus \mathbf{R}$, and $\mathbf{M}$ 

\item $\mathsf{Sim}$ sends ($\mathsf{Modeify}$, $\mathsf{Sid}$, $\widetilde{\mathbf{E}} \oplus \mathbf{R}$, ($\mathsf{Modeify}$, $\mathsf{Sid}$, $\mathbf{M}$) to functionality $\mathcal{F}^\mathsf{2pc}_\mathsf{softmax}$.

\end{itemize}

$\mathsf{Indistinguishable.}$ Obviously, when the corrupted $\mathcal{A}$ sets its output as $\mathbf{M}$, $\mathcal{F}^\mathsf{2pc}_\mathsf{softmax}$ will output $\mathbf{Z} - \mathbf{M}$ to $\mathcal{A}$, leading a correct output. For the incoming messages $[\![\widetilde{\langle\mathbf{E}\rangle_\mathcal{B}}]\!]_\mathcal{B}$,  $[\![ \widehat{\mathbf{V}} ]\!]_\mathcal{B}$ of the corrupted $\mathcal{A}$, if the adversary $\mathsf{Adv}$ can distinguish $[\![\widetilde{\langle\mathbf{E}\rangle_\mathcal{B}}']\!]_\mathcal{B}$,  $[\![ \widehat{\mathbf{V}}' ]\!]_\mathcal{B}$, we can construct a game to break the BFV scheme. And with random key $\mathbf{v}$, $\mathcal{A}$ can not distinguish it from the random value $\sum(\mathbf{E}) \odot \mathbf{v}$ of the ideal world.

Consider the second case that $\mathcal{B}$ is corrupted. In the protocol, $\mathcal{B}$ receives messages $[\![\widetilde{\mathbf{E}} \oplus \mathbf{R}]\!]_\mathcal{B} $, $[\![ \mathbf{SR} ]\!]_\mathcal{A}$, $[\![\langle \mathbf{Y} \rangle_\mathcal{B}]\!]_\mathcal{B}$. We construct the simulator $\mathsf{Sim}$ as follows:

\begin{itemize}
\item $\mathsf{Sim}$ sends $(\mathsf{Input}, \mathsf{sid}, B)$ to functionality $\mathcal{F}^\mathsf{2pc}_\mathsf{softmax}$.

\item $\mathsf{Sim}$ receives $[\![\widetilde{\langle\mathbf{E}\rangle_\mathcal{B}}]\!]_\mathcal{B}$,  $[\![ \widehat{\mathbf{V}} ]\!]_\mathcal{B}$ from the corrupted party $\mathcal{B}$.

\item $\mathsf{Sim}$ picks random value $\mathbf{SR}'$ and sends $[\![ {\mathbf{SR}'} ]\!]_\mathcal{A}$ to corrupted $\mathcal{B}$;

\item Upon receiving $ [\![ \sum(\mathbf{E}) \odot \mathbf{v}]\!]_\mathcal{A}$, from the corrupted $\mathcal{B}$, $\mathsf{Sim}$ extracts $\sum(\mathbf{E}) \odot \mathbf{v}$,

\item $\mathsf{Sim}$ sends $(\mathsf{Modify}, \mathsf{sid}, \sum(\mathbf{E}) \odot \mathbf{v})$ to functionality $\mathcal{F}^\mathsf{2pc}_\mathsf{softmax}$.
\end{itemize}

$\mathsf{Indistinguishable.}$ Obviously, when the corrupted $\mathcal{B}$ receives $c= \langle \mathbf{Y} \rangle_\mathcal{B}$, it will decrypt it and set $\langle \mathbf{Y} \rangle_\mathcal{B}$ as a result, leading a correct output. For the incoming message $[\![\widetilde{\mathbf{E}} \oplus \mathbf{R}]\!]_\mathcal{B} $, $[\![\langle \mathbf{Y} \rangle_\mathcal{B}]\!]_\mathcal{B}$ of the corrupted $\mathcal{B}$, due that $\widetilde{\mathbf{E}} \oplus \mathbf{R}$ and $\langle \mathbf{Y} \rangle_\mathcal{B}$ is encrypted with one-time pad, with the random key $\mathbf{R}$ and $\mathbf{M}$, $\mathcal{B}$ cannot distinguish it from the random value $\langle Z \rangle_\mathcal{B}$ and $\widetilde{\mathbf{E}}$ of the ideal world. For the incoming $[\![ \mathbf{SR} ]\!]_\mathcal{A}$ of corrupted $\mathcal{B}$, if the adversary $\mathsf{Adv}$ can distinguish $[\![ \mathbf{SR} ]\!]_\mathcal{A}$ and $[\![ \mathbf{SR}' ]\!]_\mathcal{A}$, we can construct a game to break the BFV scheme.

\end{proof}

\subsection{Security of LayerNrom}

We define the functionality of LayerNorm (denoted as $\mathcal{F}_\mathsf{ln}^\mathsf{2pc}$) as an instance of $\mathcal{F}^\mathsf{2pc}$, where $\mathcal{F}^\mathsf{2pc}_\mathsf{ln}$ calculates $\mathbf{Z} = f(\langle \mathbf{X} \rangle_\mathcal{A}, \langle \mathbf{X} \rangle_\mathcal{B}) := LayerNorm(\mathbf{\gamma} \cdot \frac{x_i - \mu}{\sigma} + \mathbf{\beta})$. if $\mathcal{B}$ is the corrupted party, instead of picking randoms $\langle \mathbf{Z} \rangle_\mathcal{B}$, $\mathcal{F}^\mathsf{2pc}_\mathsf{ln}$ receives $(\mathsf{Modify}, \mathsf{sid}, \langle \mathbf{Z} \rangle_{\mathcal{B}})$ from $\mathsf{Sim}$, calculates $\langle \mathbf{Z} \rangle_\mathcal{A} = \mathbf{Z} - \langle \mathbf{Z} \rangle_\mathcal{B}$.

\begin{theorem}
    The protocol $\Pi_\mathsf{ln}$ is securely computes functionality $\mathcal{F}^\mathsf{2pc}_\mathsf{ln}$ in the presence of static semi-honest adversaries $\mathsf{Adv}$ controlling each of parity $\mathcal{A}$ or $\mathcal{B}$.
\end{theorem}

\begin{proof}\renewcommand{\qedsymbol}{}
The proof process is the same as that of SoftMax.
\end{proof}

\section{Segmental Approximation of Nonlinear Functions} 
\label{appen: non-linear appro}
%%%%%%%%%%%%%%%%%%%%%%%%%%%%%%%%%inputfile$%%%%%%%%%%%
\begin{figure*}%[h!]
    \centering
    \centering\scalebox{0.66}{
    \includegraphics[width=1.10\linewidth]{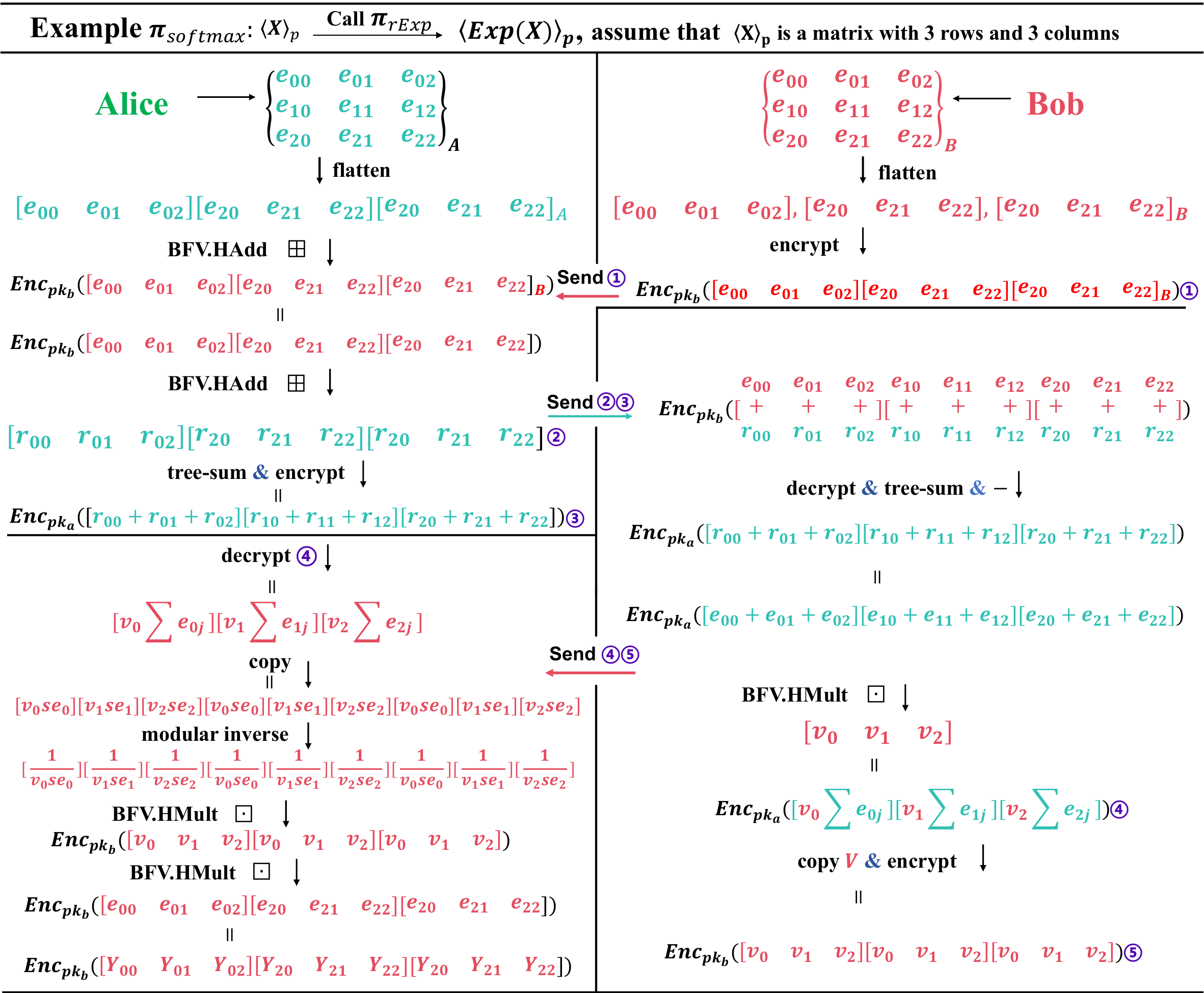}}
    \caption{The example of $\Pi_\mathsf{softmax}$}
    \label{Fig: ieSoftmax}
\end{figure*}

%%%%%%%%%%%%%%%%%%%%%%%%%%%%%%%%%inputfile$%%%%%%%%%%%

\subsection{Approximation of GeLU}
\label{appen: geluinterval}
The piecewise polynomial approximations F1 and F2 for the seg5GeLU function are as follows:

\begin{equation}\footnotesize
\begin{split}
F_1(x)=&-0.568686678 -  0.529288810 x  - 0.183509590 x^2 \\&- 0.028070202 x ^ 3  -0.001597741 x ^ 4\\
F_2(x)=&0.001193207 +  0.5 x  + 0.385858026x^2  -0.045101361 x ^ 4\\
F_3(x)=&-0.438406187 + 1.340789252x - 0.087184212 x ^ 2\\& + 0.007334718 x ^ 3
\label{eq:seg5GeLU_parm}
\end{split}
\end{equation}

\subsection{Approximation of Sigmoid}
The Sigmoid function is:
\begin{equation}
    y = Sigmoid(x) = \frac{1}{1 + e ^{-x}}
\end{equation}
which is regarding $(0, 0.5)$ symmetry, so we can just concern the situation where $x\geq0$, and for $x<0$, $Sigmoid(x)=1-Sigmoid(-x)$. Its derivatives is $y'=y(1-y)>0$, so we can get its second derivative: $y''= y'(1-2y)$; And its third derivative: $y'''=y'(6y^2-6y+1)$. Let $y''=0$, we can get $x=0$, and let $y'''=0$, due to $y'>0$, so let $6y^2-6y+1=0$, that is $y=\frac{3\pm\sqrt{3}}{6}$, and we obtain the inflection point $x_1=\ln(2\pm\sqrt{3})$. Because of $x\geq0$, so $x_1=\ln(2+\sqrt{3})$. We can obtain the other inflection point $x_2 = 6.48$, when $y'' < 10^{-5}$. Consequently, we can derive a piecewise approximation to the Sigmoid function:

\begin{equation}
seg4Sigmoid(x) = 
\begin{cases}
    F_1(x), & \text{if } 0 \leq x  <  x_1 \\
    F_2(x), & \text{if } x_1 \leq 0 < x_2 \\
    1, & \text{if } x \geq x_2
\end{cases}
\label{eq:seg4sigmoid}
\end{equation}
where the interval functions $F_1$ and $F_2$ are given below:
\begin{equation}\footnotesize
\begin{split}
F_1(x)=&0.4998102695 +  0.2527736008 x - 0.0086980795^2 \\ - 0.0127621849 x ^ 3
F_2(x)=&0.4489827105 + 0.3642809155 - 0.0948498277 x ^ 2\\& + 0.0113621587 x ^ 3 -0.0005220290 x ^ 4
\label{eq:seg4sigmoid_parm}
\end{split}
\end{equation}
Seg4Sigmoid for the negative half-axis can be easily obtained based on the symmetry.

\subsection{Approximation of Tanh}
Tanh's graphics is similar to Sigmoid's, the Tanh function is:
\begin{equation}
    y = \tanh(x) = \frac{e^x-e^{-x}}{e^x+e^{-x}}
\end{equation}
Tanh is regarding $(0, 0)$ symmetry, similar to Sigmoid, we can just concern the situation where $x\geq0$, and for $x<0$, $\tanh(x)=-\tanh(-x)$. Its derivative is $y'=1-y^2>0$, so we can get its second derivative: $y''= 2yy'$; And its third derivative: $y'''=2y'(3y^2-1)$. Let $y''=0$, we can get $x=0$, and let $y'''=0$, due to $y'>0$, so $3y^2-1=0$, that is $y=\pm\frac{1}{\sqrt{3}}$, and we obtain the point $x_1=\ln\frac{\sqrt{3}\pm1}{\sqrt2}$. Because of $x\geq0$, so $x_1=\ln\frac{\sqrt{3}+1}{\sqrt2}$. We can locate the segmentation point $x_1' = \ln\frac{\sqrt{3}+2}{\sqrt2}$ on either side of $x_1$ which results in the smallest error. We can also obtain the other inflection point $x_2 = 4.60$ when $y'' < 10^{-5}$. Consequently, we can derive a piecewise approximation to the Tanh function:

\begin{equation}
seg4Tanh(x) = 
\begin{cases}
    F_1(x), & \text{if } 0 \leq x  <  x_1' \\
    F_2(x), & \text{if } x_1' \leq 0 < x_2 \\
    1, & \text{if } x \geq x_2
\end{cases}
\label{eq:seg4sigmoid}
\end{equation}
where the interval functions $F_1$ and $F_2$ are given below:
\begin{equation}\footnotesize
\begin{split}
F_1(x)=&-0.0018890324 +  1.0384417257 x - 0.1695016932x^2\\&  -0.1084776546 x ^ 3\\
F_2(x)=&0.0800126966 + 1.0756763251 x -0.4766182792x ^ 2\\& + 0.0938427835 x ^ 3 -0.0068823466 x ^ 4
\label{eq:seg4sigmoid_parm}
\end{split}
\end{equation}

Compared with Bolt, which employs a piecewise polynomial of degree 5, our model, using a fourth-degree polynomial, still achieves higher fitting accuracy.

\subsection{Approximation of Mish}
Mish's graphics is similar to GeLU's, but it has no symmetry with its derivative. The Mish function is:
\begin{equation}
    y = Mish(x) = x\tanh(\ln(1+e^x))
\end{equation}
Because $y''$ and $y'''$ are too complex and hard to get their zero point, we use dichotomy to get $y''$'s approximate zero point: $x_1=-2.2563763963607935$ and $x_2=1.4905711794854284$. Consequently, we can derive a piecewise approximation to the Mish function:

\begin{equation}
seg4Mish(x) = 
\begin{cases}
    0, & \text{if } x < -8 \\ 
    F_1(x), & \text{if } -8 \leq x  <  x_1 \\
    F_2(x), & \text{if } x_1 \leq x < x_2 \\
    F_3(x), & \text{if } x_2 \leq x < 8 \\
    x, & \text{if } x \geq 8
\end{cases}
\label{eq:seg5mish}
\end{equation}
where the interval functions $F_1$, $F_2$ and $F_2$ are given below:
\begin{equation}\footnotesize
\begin{split}
F_1(x)=&-0.1150272397 + 0.5194677655 x +0.4293028981x^2\\& + 0.1459472737 x ^ 3 + 0.0271015218x^4 + 0.0028988426 x^5\\&+0.0001685503 x^6 + 0.0000041415 x^7\\
F_2(x)=&0.0000929623+0.5993108159x+0.3185423599x^2\\&-0.0135480666x^3-0.0420248186^x4-0.0022342097^x5\\&+0.0043057993x^6+0.0008690923x^7\\
F_3(x)=&-0.2470775212+1.0311064672x+0.1227243900x^2\\&-0.0757410810x^3+0.0200857395x^4-0.0027959123x^5\\&+0.0002003775x^6-0.000005848x^7
\label{eq:seg5mish_parm}
\end{split}
\end{equation}

\subsection{Example of the secure Softmax $\Pi_\mathsf{softmax}$ \& $\Pi_\mathsf{ln}$ } 
\label{appendix: softmax&ln}
Here are two examples, with Fig. \ref{Fig: ieSoftmax} and Fig. \ref{Fig: ieln} corresponding to $\Pi_\mathsf{softmax}$ and $\Pi_\mathsf{ln}$, respectively.

\begin{figure*}%[b!]
    \centering
    \centering\scalebox{0.76}{
    \includegraphics[width=1.0\linewidth]{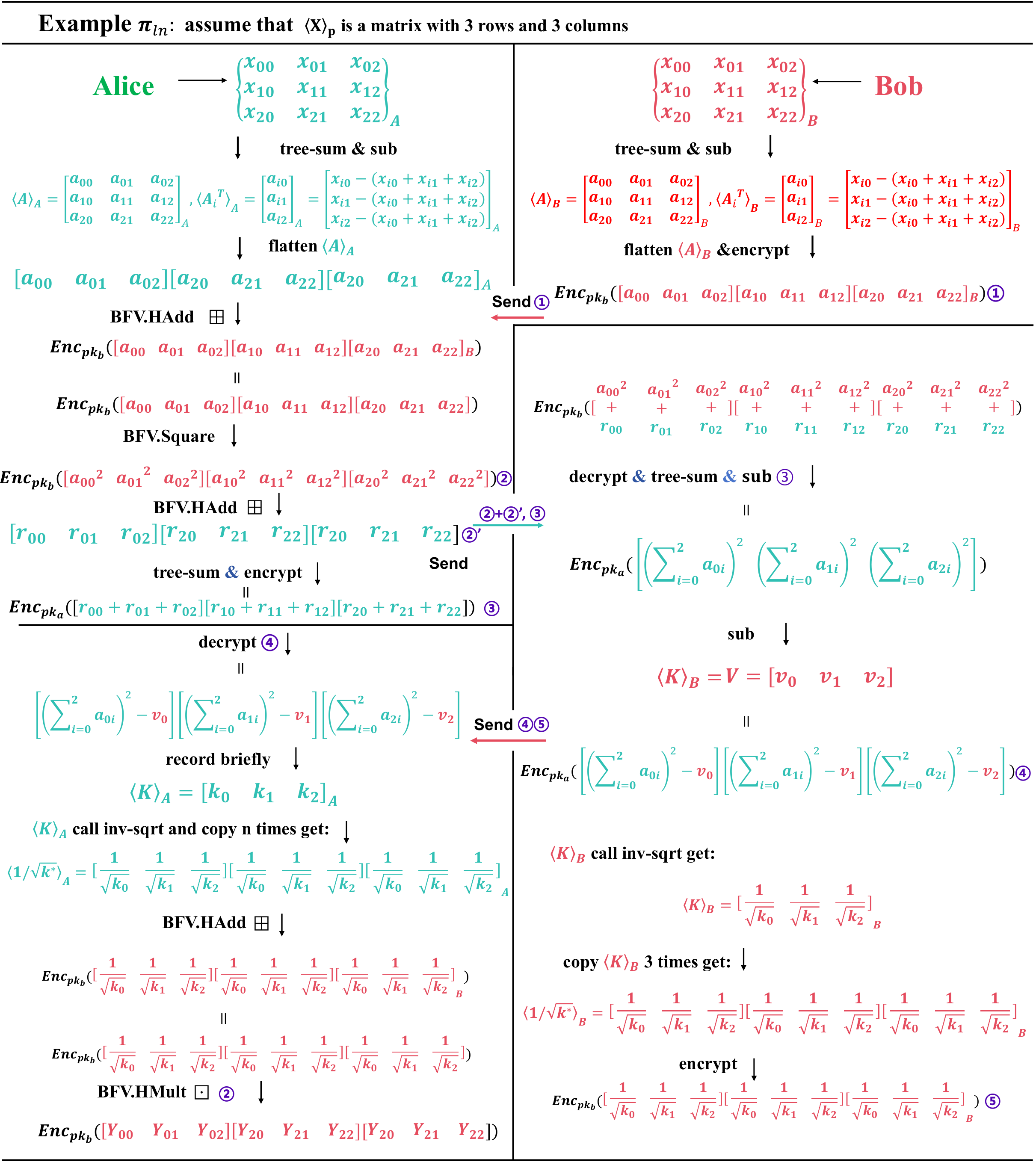}}
    \caption{The example of $\Pi_{ln}$}
    \label{Fig: ieln}
\end{figure*}

%\appendix

% \section{Qualitative examples}
% \subsection{Telecom Task Completion}
%\newpage
%\section*{Appendices}

\end{document}